\documentclass[reprint,amsmath,amssymb,aps,onecolumn,pra,nofootinbib,longbibliography,floatfix,superscriptaddress]{revtex4-2}

\usepackage{dcolumn}
\usepackage{bm}
\usepackage{relsize}
\usepackage[utf8]{inputenc}
\usepackage[english]{babel}
\usepackage[T1]{fontenc}
\usepackage{graphicx}
\usepackage{xcolor}
\usepackage[normalem]{ulem}
\usepackage{soul}
\usepackage[caption=false]{subfig}
\usepackage{nicematrix}
\usepackage{setspace}

\usepackage{tikz}
\usepackage{tikz-cd}
\usepackage{pgfplotstable}
\usepgfplotslibrary{colorbrewer}
\usetikzlibrary{decorations.pathreplacing, cd, shapes.geometric, arrows, calc, patterns}

\usepackage{amsthm}
\usepackage{amsfonts}
\usepackage{mathtools}
\usepackage{array}
\usepackage{tabulary}
\usepackage{braket}
\usepackage{stackrel}
\usepackage{stmaryrd}
\usepackage{dsfont}
\usepackage{accents}

\usepackage[linesnumbered, ruled, vlined]{algorithm2e}
\SetKwInput{KwInput}{Input}                
\SetKwInput{KwOutput}{Output}              

\usepackage[
    colorlinks=true,
    urlcolor=blue,
    linkcolor=blue,
    citecolor=blue,
    filecolor=blue,
]{hyperref}
\usepackage[capitalise]{cleveref}

\usepackage{datetime}
\usepackage{todonotes}
\setuptodonotes{size=\footnotesize\singlespacing}

\newtheorem{theorem}{Theorem}
\newtheorem{corollary}[theorem]{Corollary}
\newtheorem{lemma}[theorem]{Lemma}
\newtheorem{definition}[theorem]{Definition}
\newtheorem{remark}[theorem]{Remark}
\newtheorem{example}[theorem]{Example}

\newtheorem{proposition}[theorem]{Proposition}

\definecolor{cbDarkGray}{HTML}{404040}    
\definecolor{cbLightGray}{HTML}{D9D9D9}   
\definecolor{cbMediumGray}{HTML}{8C8C8C}  
\definecolor{cbOrange}{HTML}{E69F00}      
\definecolor{cbSkyBlue}{HTML}{56B4E9}     
\definecolor{cbBluishGreen}{HTML}{009E73} 
\definecolor{cbVermillion}{HTML}{D55E00}  
\definecolor{cbBlue}{HTML}{0072B2}        
\definecolor{cbRedPurple}{HTML}{CC79A7}   
\definecolor{lightorange}{HTML}{ffb000}
\definecolor{darkorange}{HTML}{fe6100}
\definecolor{pink}{HTML}{dc267f}
\definecolor{darkblue}{HTML}{785ef0}
\definecolor{lightblue}{HTML}{648fff}

\DeclareMathOperator{\ev}{ev}

\DeclareMathOperator{\ord}{ord}
\DeclareMathOperator{\ann}{ann}

\DeclareMathOperator{\lcm}{lcm}

\DeclareMathOperator{\Aut}{Aut}
\DeclareMathOperator{\BB}{BB}
\DeclareMathOperator{\rank}{rank}
\DeclareMathOperator{\Id}{Id}

\newcommand{\im}{\mathrm{im}\,}
 
\newcommand{\F}{\mathbb{F}}
\newcommand{\Z}{\mathbb{Z}}

\newcommand{\Sp}{\mathrm{Sp}}
\newcommand{\SL}{\mathrm{SL}}
\newcommand{\GL}{\mathrm{GL}}
\newcommand{\idlGens}[1]{\langle#1\rangle} 
\newcommand{\wt}{\mathrm{wt}}
\newcommand{\suppreg}[1]{\widehat{#1}}

\makeatletter
\renewcommand{\todo}[2][]{%
    \@todo[caption={#2}, #1]{\begin{spacing}{0.5}#2\end{spacing}}%
} 
\makeatother 

\makeatletter
\renewcommand*\l@section{\@dottedtocline{1}{0em}{3em}}
\renewcommand*\l@subsection{\@dottedtocline{2}{3em}{3em}}
\renewcommand*\l@subsubsection{\@dottedtocline{3}{6em}{3.5em}}
\makeatother

\begin{document}


\title{Spectral Theory of Semisimple Bivariate Bicycle Codes}

\author{Eric Sabo}
\email{eric.sabo@xanadu.ai}
\thanks{Corresponding author}
\affiliation{Xanadu, Toronto, Ontario M5G 2C8, Canada}

\author{Mahir Bilen Can}
\affiliation{Xanadu, Toronto, Ontario M5G 2C8, Canada}

\author{David Marquis}
\affiliation{Xanadu, Toronto, Ontario M5G 2C8, Canada}

\date{\today}

\begin{abstract}
    Extending the classical theory of two-dimensional cyclic codes, we develop an algebraic approach to bivariate bicycle codes. Using Frobenius-orbit idempotents, formulas for logical dimensions are derived and lower bounds on minimum distances are established. A systematic theory of code symmetries is formulated to construct a structured block-monomial subgroup of coordinate permutations.
    Several explicit examples show how to generate these codes from first principles without relying on numerical searches. 
    An appendix extends the analysis to BCH-based product constructions. 
\end{abstract}

\maketitle
\tableofcontents
\clearpage
\begin{table}[p]
\caption{Principal notation used in the paper.}
\label{tab:notation}
\centering
\renewcommand{\arraystretch}{1.03}
\begin{tabular}{
@{}
>{\begin{minipage}[t]{0.23\textwidth}\raggedright\arraybackslash}l<{\end{minipage}}
>{\begin{minipage}[t]{0.71\textwidth}\raggedright\arraybackslash}l<{\end{minipage}}
@{}
}
\textbf{Notation} & \textbf{Meaning}\\

\multicolumn{2}{@{}l}{\textbf{Ambient algebra}}\\
$p$, $q=p^e$
  & Characteristic and field size; sections over a prime field assume $q=p$.\\
$\ell,m$
  & Orders of the two cyclic directions; $\dim_{\F_q}R=\ell m$.\\
$R$, $R^2$
  & $R=\F_q[x,y]/\idlGens{x^\ell-1,y^m-1}$ and the two-block BB module.\\
$K$; $\alpha,\beta$
  & A splitting field and fixed primitive $\ell$-th and $m$-th roots of unity.\\
$\iota$
  & Coordinate inversion, $c(x,y)\mapsto c(x^{-1},y^{-1})$.\\
$M_c$; $A,B$
  & Multiplication by $c$; $A=M_a$ and $B=M_b$.\\
$\idlGens{a,b}$, $\ann_R(I)$
  & Generated ideal and annihilator $\{r\in R:rI=0\}$.\\
$\idlGens{b:a}$
  & Colon ideal $\{u\in R:au\in\idlGens b\}$.\\
$\wt_R$, $d_R$
  & Hamming weight in the monomial basis and the associated minimum distance.\\
$\mu$
  & Multiplier map $\mu:(i, j) \mapsto (q \cdot i \pmod \ell, q \cdot j \pmod m)$\newline\\[1pt]

\multicolumn{2}{@{}l}{\textbf{Orbits, idempotents, and regions}}\\
$\mathcal Z(f)$, $\mathcal Z(I)$
  & Zero sets of a polynomial or ideal on the root grid.\\
$O_{(i,j)}$, $\Omega$
  & A $q$-Frobenius orbit in $\Z_\ell\times\Z_m$ and the set of all such orbits.\\
$e_O$, $e_{\mathcal S}$
  & Orbit idempotent and selected sum
    $e_{\mathcal S}=\sum_{O\in\mathcal S}e_O$.\\
$J_{\mathcal S}$
  & Support-selected ideal $J_{\mathcal S}=e_{\mathcal S}R$.\\
$c_O$, $S_c$
  & Component $c_O=e_Oc$ and active support
    $S_c=\{O:c_O\ne0\}$.\\
$\suppreg{\mathcal T}_{a,b}$,
$\suppreg{\mathcal U}_{a,b}$,
$\suppreg{\mathcal F}_{a,b}$
  & Active-support regions where respectively both, only $a$, or neither of
    $a,b$ is nonzero; interchange $a,b$ for
    $\suppreg{\mathcal U}_{b,a}$.\\
$\mathcal T_{a,b}$,
$\mathcal U_{a,b}$,
$\mathcal F_{a,b}$
  & Zero-set regions where respectively both, only $a$, or neither of
    $a,b$ vanishes; interchange $a,b$ for $\mathcal U_{b,a}$.\\
Region duality
  & Fourier transformations $\suppreg{\mathcal T}\leftrightarrow\mathcal F$,
    $\suppreg{\mathcal F}\leftrightarrow\mathcal T$, and
    $\suppreg{\mathcal U}_{a,b}\leftrightarrow\mathcal U_{b,a}$.\newline\\[1pt]

\multicolumn{2}{@{}l}{\textbf{BB codes and distance bounds}}\\
$Q=\BB(a,b)$
  & BB CSS code determined by $a,b\in R$.\\
$H_X,H_Z$
  & $H_X=[A\ B]$ and
    $H_Z=[B^\top\ {-}A^\top]
         =[M_{\iota(b)}\ {-}M_{\iota(a)}]$.\\
$\Psi_X,\Psi_Z$
  & $\Psi_X(u,v)=au+bv$ and
    $\Psi_Z(u,v)=\iota(b)u-\iota(a)v$.\\
$H_1,H^1$
  & Logical-$Z$ and $X$ spaces.\\
$[\![2\ell m,k,d]\!]_q$
  & Physical length, logical dimension, and quantum minimum distance.\\
$\mathcal C(S)$
  & 2D cyclic code whose defining zero-orbit set is $S$.\\
$E_a,E_b$
  & Stabilizer-excluded single-block distance terms from
    $\ann\idlGens a$ and $\ann\idlGens b$.\\
$N_{a,b}$, $N^*_{a,b}$
  & Colon-ideal mixed-block term and its stabilizer-excluded refinement.\newline\\[1pt]

\multicolumn{2}{@{}l}{\textbf{Later sections}}\\
$\Aut_{\mathrm{perm}}(Q)$, $\Aut_{\mathrm{LC}}(Q)$
  & Permutation and local-Clifford--permutation automorphism groups.\\
$\lambda_O$
  & The slope on orbit $O$.\\
$\widetilde R,\widetilde Q$; $\rho,\pi,w$
  & Cover ring and code; ring projection, orbit projection, and wrapping
    multiplicity.\\
$\mathsf{Prod}(C_x,C_y)$
  & Product-code ideal
    $\idlGens{g_x(x)g_y(y)}\cong C_x\otimes C_y$.\\
\end{tabular}
\end{table}
\clearpage
\section{Introduction}
The realization of scalable fault-tolerant quantum computing rests on quantum error-correcting codes that combine high error thresholds with low hardware overhead. In this regard, bivariate bicycle (BB) codes have emerged as a leading family of quantum low-density parity-check (QLDPC) codes for near-term architectures, offering encoding rates and distances far beyond the surface code at comparable check weights while retaining a planar, translation-invariant layout.

To date, the discovery of good BB codes has been driven almost entirely by large-scale numerical searches over the pair of defining polynomials. These sweeps have produced an abundance of codes with good parameters, but the algebraic mechanisms that fix those parameters have remained largely obscured. The result is a gap between knowing that a good code exists at a given length and being able to design one deterministically, with its dimension and distance guaranteed from first principles rather than checked after the fact.
In this paper, we close part of that gap by placing BB codes inside the classical theory of two-dimensional (2D) cyclic codes. In the semisimple case $\gcd(\ell m, q) = 1$ over a finite field $\F_q$, the defining ring $R = \F_q[x,y]/\idlGens{x^\ell - 1, y^m - 1}$ decomposes, by the Chinese Remainder Theorem and the Wedderburn structure theorem, into a product of finite fields indexed by the $q$-Frobenius orbits of $\Z_\ell \times \Z_m$. Quantities of interest such as logical dimension, minimum distance bounds, metachecks, and covering-code parameters can be read off from how the two defining polynomials project onto these orbit components.

Our first contribution is an easy-to-compute dimension formula. By partitioning the $q$-Frobenius orbits into regions where both, one, or neither of the defining polynomials vanish, we show that the number of logical qudits is controlled entirely by the common-zero region, the orbits where both defining polynomials vanish. Furthermore, we prove that orbits where only one polynomial vanishes are homologically trivial, and we establish an exact algebraic duality between the $X$- and $Z$-logical spaces via the canonical coordinate involution.

Our second and central contribution concerns the minimum distance. Existing bounds proceed by isolating the logical contribution of each block separately, which can miss low-weight logical operators supported across both blocks. We introduce a colon-ideal bound that accounts for these mixed-block logicals directly. By refining this bound with an alternating stabilizer exclusion, we obtain a sharper, computable lower bound on the quantum distance from the minimum Hamming weights of classical 2D cyclic codes. Combined with the spectral theory, this pinpoints which roots control the logical dimension and which control the distance, making precise the tension between the two.

Beyond static parameters, the spectral picture clarifies the symmetries of BB codes. Specializing to the prime subfield $\F_p$ to analyze local Clifford operations, we show that coordinate-permutation automorphisms are not governed by the zero sets alone; they also depend on the relative ratio of the two check polynomials. To capture this linear dependence, we introduce the notion of a spectral slope on the coupled active-support region and prove that any valid automorphism requires an exact matching of these slopes across orbits. Within this structure, we prove the unconditional existence of the signed block-swapping $ZX$-duality required for fold-transversal logic, characterize phase-type $CZ$ gates, and describe metachecks for the canonical set of redundant stabilizers.

We also apply our algebraic approach to code lifts and projections, which relate a base code to its covering-graph descendants. 
Earlier works extended dimension lower bounds to covering degrees that divide the field characteristic. We analyze the commutative algebra of the check ideals directly. We provide an independent algebraic proof using quotient ring surjections. This approach yields an exact dimension formula based on the strictly new common roots of the cover code over any finite field $\F_q$.

Instead of presenting tables of codes, we provide a method to understand existing codes in the literature. For a fixed grid, the number of BB codes is finite. The spectral approach gives a procedure to enumerate them and determine their properties without ever writing down the stabilizers. A systematic search for new codes based on this procedure will be the subject of a follow-up document. We illustrate these principles throughout with minimal working examples. These examples are designed to demonstrate specific mathematical results rather than to showcase any particular code. We find that working with coprime BB codes allows for easier discussion. We believe that the best codes utilize the full two-dimensional space in nontrivial ways. Combining the various techniques used throughout the paper may lead to codes with numerous desirable features.

This paper assumes a working knowledge of algebra. Nevertheless, we begin by recalling some of the foundational definitions that underlie the important distinctions made throughout the paper. Section~\ref{sec:1dcyclic} reviews the theory of cyclic codes, aligning with standard references such as~\cite{HuffmanPless}. Section~\ref{sec:2dcyclic} extends this to two dimensions. We review definitions and previous results on BB codes in Section~\ref{sec:BBpast} before presenting our spectral approach in Section~\ref{sec:BBnew}. Section~\ref{sec:BBauts} constructs and analyzes a structured block-monomial subgroup of the automorphism group of BB codes over $\F_p$.
Section~\ref{sec:lifts-and-projections} applies the theory to covering codes, and Section~\ref{sec:examples} works out explicit examples. The appendix~\ref{sec:bch_product_codes} contains algebraic constructions for BB codes, in particular a product construction with a BCH-based distance lower bound.
\section{Mathematical Preliminaries \& Notation}
We adopt the following conventions. 
\begin{itemize}
    \item All rings are commutative with unity.
    
    \item $\F_q$ will denote the finite field with $q = p^e$ elements. We will later restrict to the case $q = p$ is prime.
    
    \item The letter $R$ is reserved for the polynomial ring $R = \F_q[x, y] / \langle x^\ell - 1,  y^m - 1 \rangle$, where $\ell$ and $m$ are positive integers. 
    
    \item Unless otherwise stated, we will assume that $\gcd(q, \ell m) = 1$. The reader will be warned when we specialize $p$.
    
    \item As an $\F_q$-vector space, $R \cong \F_q^{\ell m}$ 
  with respect to the standard monomial basis 
  $\{x^i y^j \mid 0 \leq i < \ell,\, 0 \leq j < m\}$. 
  For any element $u = \sum_{i, j} u_{i, j} x^i y^j \in R$ with coefficients $u_{i, j} \in \F_q$, 
  its \emph{Hamming weight} is the number of its nonzero coefficients:
  \[
    \wt_R(u) = |\{(i, j) \mid u_{i, j} \neq 0\}|.
  \]
\end{itemize}

BB codes, and their generalizations, can be independently viewed as ideals in group algebras, ideals in rings, or submodules. While we may map between the different representations, each mathematical structure comes with a different set of tools and techniques that make some results more natural when viewed as one object compared to another. An analogous situation in classical coding theory is the relationship between 2D cyclic codes and quasi-cyclic codes.

For a finite group $G$ and a field $K$, the \emph{group algebra of $G$ over $K$}, denoted $K[G]$, is the $K$-algebra consisting of all formal expressions of the form $\sum_{g \in G} a_g g$ for $g \in G, a_g \in K$, where addition, multiplication, and scalar action are defined as follows:
\begin{itemize}
    \item Addition: $\displaystyle \sum_{g \in G} a_g g + \sum_{g \in G} b_g g = \sum_{g \in G} (a_g + b_g) g$. 
    \item Multiplication: $\displaystyle \left(\sum_{g \in G} a_g g \right) \cdot \left(\sum_{g \in G} b_g g\right) = \sum_{g \in G} \left(\sum_{v w = g} (a_v b_w)\right) g$.
    \item Scalar action of $K$: $\displaystyle \alpha \left(\sum_{g \in G} a_g g\right) = \sum_{g \in G}(\alpha a_g) g$ for $\alpha \in K$.
\end{itemize}
If $G$ is abelian, then $K[G]$ is a commutative ring. 
For example, if $C_\ell$ is the finite cyclic group of order $\ell$, then its group algebra over a field $K$ is given by 
\begin{equation*}
    K[C_\ell] = \{a_0 1_{C_\ell} + a_1 g + \cdots + a_{\ell - 1} g^{\ell - 1} \mid (a_0, \dots, a_{\ell - 1}) \in K^\ell \}.
\end{equation*}
Let $\langle x^\ell - 1 \rangle$ be the ideal generated by the polynomial $x^\ell - 1\in K[x]$.
It is not difficult to check that the map
\begin{equation*}
    C_\ell \to K[x] / \langle x^\ell - 1 \rangle, \quad g \mapsto x \pmod{\langle x^\ell - 1 \rangle},    
\end{equation*}
extended by linearity to $K[C_\ell]$ is a $K$-algebra isomorphism, 
\begin{equation*}
    K[C_\ell] \cong K[x] / \langle x^\ell -1 \rangle. 
\end{equation*}
More generally, if a finite abelian group $G$ is expressed as a product of cyclic groups, $G \cong C_{\ell_1} \times \cdots \times C_{\ell_r}$, then we obtain a $K$-algebra isomorphism
\begin{equation*}
    K[G] \cong K[x_1, \dots, x_r] / \langle x_1^{\ell_1} - 1, \dots, x_{r}^{\ell_r} - 1 \rangle. 
\end{equation*}
The ring $R$ is isomorphic to the group algebra $\F_q[\Z_\ell \times \Z_m]$ independently of the assumption $\gcd(p, \ell m) = 1$. From now on, we will stick with the standard polynomial ring perspective, unless explicitly noted.

In our context, the relevant ambient space is the product space $R^2 = R \times R$. This space carries two distinct algebraic structures and it is crucial to distinguish between them:
\begin{enumerate}
    \item[(1)] We can treat $R \times R$ as a product ring with component-wise multiplication.
    A subset $P \subseteq R \times R$ is an \emph{ideal} if it absorbs multiplication by any pair $(r_1, r_2)$:  
    \begin{equation*}
        \text{for all $(u, v) \in P$ and $(r_1, r_2) \in R^2$, we have $(r_1, r_2) \cdot (u, v) = (r_1 u, r_2 v) \in P$}.
    \end{equation*}

    \item[(2)] We can treat $R \times R$ as a \emph{free module} of rank 2 over $R$.
    A subset $Q \subseteq R \times R$ is a \emph{submodule} if it is closed under scaling by a single element $r \in R$ (acting diagonally):
    \begin{equation*}
        \text{for all $(u, v) \in Q$ and $r \in R$, we have $r \cdot (u, v) = (ru, rv) \in Q$}.
    \end{equation*}
\end{enumerate}

Evidently, every ideal of the product ring is an $R$-submodule, but the converse is \emph{not} true. 
This distinction leads to the definition of quasi-cyclic codes.

\begin{definition}\label{def:QCcodesofindex2}
    A classical linear code $\mathcal{C} \subseteq R \times R$ is called a \emph{quasi-cyclic code of index 2} if it is an $R$-submodule of $R^2$.
\end{definition}
\noindent Intuitively, this means that if a pair of polynomials $(u, v)$ represents a codeword in $R^2$, then shifting both simultaneously (multiplying by $r \in R$) produces another codeword.

\begin{definition}
    A finite dimensional algebra $A$ over a field $K$ is \emph{Frobenius} if there exists a linear functional $\varepsilon : A \to K$ such that the bilinear form $\langle a, b \rangle_F = \varepsilon(ab)$ is nondegenerate. If $\varepsilon(ab) = \varepsilon(ba)$, the algebra is \emph{symmetric}.
\end{definition}

We apply this to our ring $R$, where $K=\F_q$. Define the involution $\iota$ and the trace functional $\varepsilon$:
\begin{equation*}
    \iota(x^i y^j) = x^{-i} y^{-j}, \qquad \varepsilon \left(\sum a_{ij} x^i y^j \right) = a_{00},
\end{equation*}
and two bilinear forms on $R$:
\begin{enumerate}
    \item[(1)] \emph{The Frobenius Pairing:} $\langle a, b \rangle_F = \varepsilon(ab)$.
    \item[(2)] \emph{The Euclidean Pairing:} $\langle a, b\rangle_E = \varepsilon(a \iota(b))$.
\end{enumerate}
The interplay between these two forms gives rise to algebraic properties governing BB codes.

\begin{lemma}
    We maintain our notation from the previous paragraph. Then 
    \begin{enumerate}
        \item[(1)] $R$ is a symmetric Frobenius algebra with respect to $\langle \cdot, \cdot \rangle_F$.
        \item[(2)] The Euclidean pairing $\langle \cdot, \cdot \rangle_E$ is a nondegenerate, symmetric form such that the monomial basis is orthonormal.
    \end{enumerate}
\end{lemma}
\begin{proof}
    (1) Nondegeneracy of $\langle \cdot, \cdot \rangle_F$: For $0 \neq a = \sum a_g g$, if $a_k \neq 0$, then $\varepsilon(a \cdot k^{-1}) = a_k \neq 0$. Symmetry follows from the commutativity of $R$.
    
    (2) Orthonormality: $\langle x^i y^j, x^{i^\prime} y^{j^\prime} \rangle_E = \varepsilon(x^{i - i^\prime} y^{j - j^\prime}) = \delta_{i, i^\prime} \delta_{j, j^\prime}$.
\end{proof}

\begin{definition} 
    For $c \in R$, define the (left) multiplication map $M_c : R \to R$ via $M_c(f) = c f$.
\end{definition}
\noindent Relative to the standard monomial basis of $R$, we will also view $M_c$ as its representing matrix.

\begin{proposition}
\label{prop:algebraic-properties}
For every $c\in R$ and every ideal $I\subseteq R$, the following hold:
    \begin{enumerate}
        \item[(1)] For any ideal $I \subseteq R$, $\ann\idlGens{\ann\idlGens{I}} = I$, and $\dim_{\F_q} I  + \dim_{\F_q} \ann\idlGens{I} = \ell m$.
        
        \item[(2)] The adjoint of the multiplication operator $M_c$ with respect to the Euclidean pairing is $M_{\iota(c)}$. Equivalently, with respect to the standard monomial basis $\{x^i y^j \mid (i,j)\in \Z_\ell\times \Z_m\}$, the matrix transpose is given by $M_c^{\top} = M_{\iota(c)}$, where $\iota(c) = c(x^{-1}, y^{-1})$.
    \end{enumerate}
\end{proposition}

\begin{proof}
(1) This is a standard property of Frobenius algebras. Write 
$R^* = \operatorname{Hom}_{\F_q}(R, \F_q)$ for the $\F_q$-linear dual. The Frobenius 
functional $\varepsilon: R \to \F_q$ induces an $\F_q$-linear isomorphism
\[
  \Phi: R \longrightarrow R^*, \quad \Phi(a)(b) = \varepsilon(ab),
\]
which is an isomorphism of $R$-modules where $R$ acts on $R^*$ by the 
multiplication operators $(r \cdot h)(b) = h(M_r(b))$ for $h\in R^*$ and $r\in R$. 
Consequently, for any ideal $I \subseteq R$, the annihilator $\ann(I)$ 
corresponds to the orthogonal complement of $I$ under the non-degenerate 
Frobenius pairing $\langle \cdot, \cdot \rangle_F$, yielding 
$\ann\langle \ann\langle I \rangle \rangle = I$ and 
\[
  \dim_{\F_q} I + \dim_{\F_q} \ann\langle I \rangle = \dim_{\F_q} R = \ell m.
\]

    (2) We first compute the adjoint directly using the Euclidean pairing:
    \[
    \langle M_c f, g \rangle_E = \varepsilon( (cf) \iota(g) ) = \varepsilon( f (c \iota(g)) ).
    \]
    Note that $c \iota(g) = \iota( \iota(c) g )$. Thus,
    \[
    \varepsilon( f \iota( \iota(c) g ) ) = \langle f, \iota(c) g \rangle_E = \langle f, M_{\iota(c)} g \rangle_E.
    \]
    Since the standard monomial basis is orthonormal with respect to the Euclidean pairing $\langle \cdot, \cdot \rangle_E$, the adjoint $M_{\iota(c)}$ corresponds exactly to the matrix transpose $M_c^{\top}$. 
\end{proof}

\section{1D Cyclic Codes}
\label{sec:1dcyclic}
We assume the reader is familiar with classical cyclic codes and include a brief refresher to introduce notation and terminology which we align with the suggested reference~\cite{HuffmanPless} throughout the paper, as feasible. Throughout this section, we fix $n \in \mathbb{N}$, and let $R_x := \F_q[x] /  \idlGens{x^n - 1}$. The images of the monomials $1, x, \dots, x^{n-1}$ under the canonical map $\pi : \F_q[x] \to R_x$ are referred to as the \emph{standard basis} of $R_x$.

Let $f = (f_0, f_1, \dots, f_{n - 1}) \in \F_q^n$. Using the standard isomorphism between $\F^n_q$ and the ring $R_x$, we identify the vector $f$ with the polynomial $f(x) = f_0 + f_1 x + \dots + f_{n - 1} x^{n - 1} \in R_x$.

\begin{definition} 
	A linear code $\mathcal{C} \subseteq \F_q^n$ is \emph{cyclic} if for any codeword $(f_0, f_1, \dots, f_{n - 1}) \in \mathcal{C}$, it also contains its cyclic right-shift $(f_{n - 1}, f_0, \dots, f_{n - 2})$.
\end{definition}
\noindent Multiplication by $x$, $x f(x) = f_0 x + f_1 x^2 + \dots f_{n - 1} x^n$, corresponds to a shift in the codeword if and only if $x^n = 1$. Hence, cyclic codes are naturally viewed algebraically with respect to the ring $\F_q[x]/\idlGens{x^n - 1}$.
\begin{proposition}
    \label{bi}
	Cyclic codes of length $n$ are in bijection with the ideals of the ring $\F_q[x]/\idlGens{x^n - 1}$.
\end{proposition}

Let $p$ be the characteristic of $\F_q$. If $p$ divides $n$, we write $n = n^\prime p^a$ where $\gcd(n^\prime, p) = 1$. It follows from the Frobenius map that
\begin{equation}\label{repeated}
    x^n - 1 = (x^{n^\prime} - 1)^{p^a}.
\end{equation}
Consequently, the irreducible factors of $x^n - 1$ are identical to those of $x^{n^\prime} - 1$, occurring with multiplicity $p^a$. Codes of this form are called \emph{repeated-root} cyclic codes. We restrict to the case $\gcd(q, n) = 1.$ 

Ideals of this ring are generated by divisors of $x^n - 1$. To factor $x^n - 1$ over $\F_q[x]$, recall that over the splitting field, $\displaystyle x^n - 1 = \prod_{\alpha^n = 1} (x - \alpha)$, where the product is taken over all $n$-th roots of unity. While not over the splitting field, some of these terms need to be grouped together into irreducible factors, $\displaystyle x^n - 1 = \prod \mathrm{min}_\alpha(x)$, where $\mathrm{min}_\alpha(x)$ is the minimal polynomial for $\alpha$ over the appropriate base field. It follows from the binomial theorem that for $f(x) \in \F_q[x]$, $f(x^q) = f(x)^q$. Hence, for $\alpha$ a root of $f(x)$ in some extension field of $\F_q$, $f(\alpha^q) = f(\alpha)^q = 0$, implying $\alpha, \alpha^q, \alpha^{q^2}, \dots$ are all roots of $f(x)$. This sequence stops when $\alpha^{q^r} = \alpha$ for some natural number $r$. Let $K$ be the splitting field of $x^n - 1$ with $\mathrm{gcd}(q, n) = 1$, and let $\beta$ be a primitive element of $K$. Then $\alpha = \beta^d$ is a primitive $n$-th root of unity with $d = (|K| - 1)/n$, where $|K| = q^{\mathrm{ord}_n(q)}$. Then $\alpha^{q^r} = \alpha \to \beta^{dq^r - d} = 1$, or $dq^r \equiv d \mod (|K| - 1)$.
\begin{definition}
    The \emph{$q$-cyclotomic coset of $i$ modulo $n$} is the set of exponents defined by
    \begin{equation*}
        C_i = \{i, iq, iq^2, \dots, iq^{r - 1} \} \pmod{n},
    \end{equation*}
    where $r$ is the smallest positive integer such that $iq^r \equiv i \pmod{n}$.
\end{definition}
\noindent The minimal polynomial for $\alpha$ thus must also contain these other powers of $\alpha$,
\begin{equation*}
	\mathrm{min}_{\alpha^i}(x) = \prod_{j \in C_i} (x - \alpha^j).
\end{equation*}
If the minimal polynomial contained any other roots, it would also need to contain all of its $q$ powers and we could separate all of these new terms into a polynomial which divides $\mathrm{min}_\alpha(x)$, contradicting the irreducibility of the minimal polynomial. Hence, over $\F_q$, we have
\begin{equation}\label{overFq}
	x^n - 1 = \prod_{C_d} \mathrm{min}_{\alpha^d}(x).
\end{equation}
The assumption $\mathrm{gcd}(q, n) = 1$ ensures there are no repeated roots in the factorization.

Let $x^n - 1 = g(x)h(x)$ for some $g(x), h(x) \in \F_q[x]$. By Proposition \ref{bi}, $\mathcal{C} = \idlGens{g(x)}$ is a cyclic code with \emph{generator polynomial} $g(x)$ viewed as an ideal of $\F_q[x]/\idlGens{x^n - 1}$. Let $c(x) = a(x) g(x) \in \mathcal{C}$. Then, $c(x) h(x) = a(x) g(x) h(x) = a(x) (x^n - 1) \equiv 0$, earning $h(x)$ the name \emph{check polynomial}. By the same logic, the ideal generated by $h(x)$ defines another cyclic code, $\mathcal{C}^\prime = \idlGens{h(x)}$. If $\deg g(x) = r$, then $k = \deg h(x) = n - r$. Let $c^\prime(x) = b(x) h(x) \in \mathcal{C}^\prime$. 
Then $\deg a(x) \leq k - 1$ 
and $\deg b(x) \leq r - 1$, so the product 
$c(x) c^\prime(x) = a(x) g(x) b(x) h(x) = a(x) b(x) (x^n - 1)$ satisfies 
$\deg(a(x) b(x)) < n - 1$. In particular, the coefficient of $x^{n - 1}$ 
in $a(x) b(x)$ is zero. Directly expanding the product in terms of their 
coefficients $c_j$ and $c^\prime_j$, the coefficient of $x^{n - 1}$ in 
$c(x) c^\prime(x)$ is $0 = \sum_{j = 0}^{n - 1} c_j c^\prime_{n - 1 - j}$, 
which is the Euclidean dot product of the codeword $c$ and the reverse 
of $c^\prime$. It follows that codewords of $\mathcal{C}$ are orthogonal 
to the reverse of codewords in $\mathcal{C}^\prime$, meaning $\mathcal{C}^\perp$ 
is contained in this latter space. A further dimensional argument shows 
that these two spaces are in fact equal. Recalling that the reciprocal 
of a polynomial $f(x)$ of degree $m$ is $f^*(x) = x^m f(x^{-1})$, we have 
$\mathcal{C}^\perp = \idlGens{h^*(x)/h(0)}$, where we introduce the 
normalization factor $1/h(0)$ to ensure the generator polynomial is monic.

\begin{lemma}
	If $h(x) \mid x^n - 1$, then $h^*(x)/h(0)$ also divides $x^n - 1$. In particular, the dual code of a cyclic code is cyclic.
\end{lemma}

Following the above discussion, we get a handful of basic results.
\begin{lemma}
	Let $\mathcal{C}$ be a cyclic code over $\F_q$ with generator polynomial $g(x)$ of degree $r$ and check polynomial $h(x)$ of degree $k$. Then generator and parity-check matrices for $\mathcal{C}$ are given by
	\begin{equation}\label{cycgenmat}
    		G = \begin{pmatrix}
    				g_0 & g_1 & \dots & \dots & \dots & \dots & g_r & & & \\
    				 & g_0 & g_1 & \dots & \dots & \dots & \dots & g_r & & \\
    				 & & \ddots & \ddots & \ddots & \ddots & \ddots & \ddots & \ddots & \\
    				 & & & g_0 & g_1 & \dots & \dots & \dots & \dots & g_r
    			\end{pmatrix},
    	\end{equation}
    	and
    	\begin{equation}
    		H = \begin{pmatrix}
    				h_k & h_{k - 1} & \dots & \dots & \dots & \dots & h_0 & & & \\
    				 & h_k & h_{k - 1} & \dots & \dots & \dots & \dots & h_0 & & \\
    				 & & \ddots & \ddots & \ddots & \ddots && \ddots & \ddots &  \\
    				 & & & h_k & h_{k - 1} & \dots & \dots & \dots & \dots & h_0
    			\end{pmatrix},
    	\end{equation}
    	respectively. 
\end{lemma}
\noindent 
The $k$ rows of $G$ and the $r$ rows of $H$ are linearly independent, 
although we could perform an extra $r$ (respectively, $k$) cyclic shifts to 
complete either into a full $n \times n$ circulant matrix.

\begin{lemma}
	Let $\mathcal{C}$ be a cyclic code of length $n$ with check polynomial $h(x)$ of degree $k$. Then $\mathcal{C}$ is an $[n, k]_q$-linear code.\footnote{When $q = 2$, we leave off the subscript.}
\end{lemma}

There are a couple of useful bookkeeping tools commonly used to describe cyclic codes.
\begin{definition} 
    Let $\mathcal{C} = \idlGens{g(x)}$ be a cyclic code of length $n$ over $\F_q$. 
    Then $\mathcal{T} = \bigcup_d C_d$, where the union is over the $q$-cyclotomic 
    cosets $C_d$ modulo $n$ corresponding to the roots of $g(x)$, is called the 
    \emph{defining set} of $\mathcal{C}$.
\end{definition}
\noindent Let $\alpha$ be a primitive $n$-th root of unity in the splitting field $K$. 
As is clear from the definition, $\mathcal{T}$ completely determines $g(x)$ and 
vice versa, with $g(x) = \prod_{j \in \mathcal{T}} (x - \alpha^j)$.
\begin{definition} 
    \label{def:zset}
    The set of powers of $\alpha$ that are roots of $g(x)$ is called the 
    \emph{zero set} of $\mathcal{C}$, denoted 
    $\mathcal{Z}(\mathcal{C}) = \{\alpha^j \mid j \in \mathcal{T} \}$, and the 
    elements of this set are called \emph{zeros} of the code.
\end{definition}

If $x^n - 1 = g(x) h(x)$, then the roots of $h(x)$ are the roots of $x^n - 1$ 
not contained in $g(x)$. Since $f^*(x) = x^{\deg f} f(x^{-1})$, the roots of 
$f^*(x)$ are the inverses of the roots of $f(x)$.
\begin{lemma}
    Let $\mathcal{T}_{\mathcal{C}}$ be the defining set of a cyclic code $\mathcal{C}$ 
    of length $n$. Then
    \begin{equation*}
        \mathcal{T}_{\mathcal{C}^\perp} = \{ -i \pmod n \mid i \in \{0, \dots, n - 1\} \setminus \mathcal{T}_{\mathcal{C}} \}.
    \end{equation*}
\end{lemma}

\begin{example}\label{ex:1dcycpart1}
    Consider the ring $\F_2[x] / \idlGens{x^7 - 1}$. The polynomial $x^7 - 1$ factors over $\F_2$ as
    \begin{equation*}
        x^7 - 1 = (x + 1)(x^3 + x + 1)(x^3 + x^2 + 1).
    \end{equation*}
    The 2-cyclotomic cosets modulo 7 are $C_0 = \{0\}$, $C_1 = \{1, 2, 4\}$, and $C_3 = \{3, 5, 6\}$. There are eight possible cyclic codes:
    \begin{gather*}
        \mathcal{T}_0 = \emptyset, \quad \mathcal{T}_1 = C_0, \quad \mathcal{T}_2 = C_0 \cup C_1, \quad \mathcal{T}_3 = C_0 \cup C_3,\\
        \mathcal{T}_4 = C_0 \cup C_1 \cup C_3, \quad \mathcal{T}_5 = C_1, \quad \mathcal{T}_6 = C_1 \cup C_3, \quad \mathcal{T}_7 = C_3.
    \end{gather*}
    The corresponding generator polynomials are
    \begin{gather*}
        g_0(x) = 1, \quad g_1(x) = 1 + x, \quad g_2(x) = 1 + x^2 + x^3 + x^4, \quad g_3(x) = 1 + x + x^2 + x^4,\\
        g_4(x) = 0, \quad g_5(x) = 1 + x + x^3, \quad g_6(x) = 1 + x + x^2 + x^3 + x^4 + x^5 + x^6,\\
        g_7(x) = 1 + x^2 + x^3.
    \end{gather*}
\end{example}

Let $s$ be an integer such that $\gcd(s, n) = 1$. The \emph{multiplier} 
$\mu_s : \{0, 1, \dots, n - 1\} \to \{0, 1, \dots, n - 1\}$ is a permutation 
of the coordinate indices defined by $\mu_s(i) = s \cdot i \pmod n$. Under 
the natural isomorphism $\F_q^n \cong \F_q[x]/\idlGens{x^n - 1}$, 
the action of $\mu_s$ on the coordinate indices induces an action on the 
polynomials. For any codeword $c(x) \in \F_q[x]/\idlGens{x^n - 1}$, 
this is given by $\mu_s(c(x)) \equiv c(x^s) \pmod{x^n - 1}$. The map $\mu_s$ 
is called a \emph{multiplier automorphism} of a cyclic code $\mathcal{C}$ 
if $\mu_s(\mathcal{C}) = \mathcal{C}$. This property can be characterized 
entirely by the code's defining set: $\mu_s \in \mathrm{Aut}(\mathcal{C})$ 
if and only if $\mu_s(\mathcal{T}_{\mathcal{C}}) = \mathcal{T}_{\mathcal{C}}$. 
Since $\mu_s$ purely permutes the coordinates, it preserves the Hamming weight.

One of the most successful families of cyclic codes is the class of BCH codes, 
initiated by Bose, Ray-Chaudhuri, and Hocquenghem. BCH codes are hand-designed 
to guarantee a chosen minimum distance. In short, a high minimum weight can be 
guaranteed by choosing a defining set with a large number of consecutive elements. 
On the other hand, $\dim \mathcal{C} = n - | \mathcal{T} |$, implying that smaller 
defining sets are preferable. BCH codes provide a way to construct cyclic codes 
with both high minimum distance and high dimension by choosing $\mathcal{T}$ as 
small as possible such that it is a union of cyclotomic cosets containing $\delta - 1$ 
consecutive elements.
\begin{definition} 
    A BCH code $\mathcal{C}$ over $\F_q$ of length $n$ and \emph{designed distance} 
    $2 \leq \delta < n$ is a cyclic code with defining set 
    $\mathcal{T} = C_b \cup C_{b + 1} \cup \dots \cup C_{b + \delta - 2}$ whose zeros 
    are generated by a primitive $n$-th root of unity $\alpha \in \F_{q^m}$, where 
    $m = \mathrm{ord}_n(q)$.
\end{definition}
\noindent The dual of a BCH code is, in general, not a BCH code, as the remaining 
cyclotomic cosets defining $h(x)$ need no longer be consecutive.

\begin{theorem} 
    \label{thm:BCH-bound}
    Let $\mathcal{C}$ be a cyclic code of length $n$ over $\F_q$ with defining set 
    $\mathcal{T}(\mathcal{C})$. If $\mathcal{T}(\mathcal{C})$ contains $\delta - 1$ 
    consecutive exponents modulo $n$, given by 
    $b, b + 1, \dots, b + \delta - 2 \pmod n$, then $d(\mathcal{C}) \geq \delta$.
\end{theorem}

\begin{corollary}
    A BCH code with designed distance $\delta$ has minimum distance $d \geq \delta$.
\end{corollary}

\begin{example}\label{ex:1dcycpart2}
    Continuing Example~\ref{ex:1dcycpart1}, applying the BCH bound to these defining sets gives
    \begin{gather*}
        \delta_0 = 1, \quad \delta_1 = 2, \quad \delta_2 = 4, \quad \delta_3 = 4,\\
        \delta_4 = \infty, \quad \delta_5 = 3, \quad \delta_6 = 7, \quad \delta_7 = 3.
    \end{gather*}
    Checking degrees, the resulting codes have parameters
    \begin{gather*}
        \mathcal{C}_0: [7, 7, 1], \quad \mathcal{C}_1: [7, 6, 2], \quad \mathcal{C}_2: [7, 3, 4], \quad \mathcal{C}_3: [7, 3, 4],\\
        \mathcal{C}_4: [7, 0, -], \quad \mathcal{C}_5: [7, 4, 3], \quad \mathcal{C}_6: [7, 1, 7], \quad \mathcal{C}_7: [7, 4, 3],
    \end{gather*}
    where we have computed the true minimum distances exactly for comparison against the BCH bound.
\end{example}

Consider the factorization of $x^n - 1$ in~\eqref{overFq}. The minimal 
polynomials are pairwise coprime and irreducible over $\F_q$. The Chinese 
Remainder Theorem gives
\begin{equation}\label{CRT1D}
    R = \F_q[x] / \idlGens{x^n - 1} \cong \bigoplus_{C_d} \F_q[x]/ \idlGens{\min_{\alpha^i}(x)},
\end{equation}
where the sum is taken over all cyclotomic cosets $C_d$, and $\alpha^i$ 
is a root representative for $C_d$. Each summand is a field, so 
$R \cong \bigoplus_{C_d} \F_{q^{|C_d|}}$. This is the Wedderburn decomposition 
of the ring. Let $\hat{e}_i$ denote the canonical primitive idempotents of the 
direct sum, where $\hat{e}_i$ has the identity element $1$ in the $i$th field 
component and $0$ in all others. These satisfy $\hat{e}_i^2 = \hat{e}_i$, 
$\hat{e}_i \hat{e}_j = 0$ for $i \neq j$, and $\sum_i \hat{e}_i = 1$. 
Furthermore, $\hat{e}_i R$ is isomorphic to the $i$th field, giving 
$R \cong \bigoplus_i \hat{e}_i R$.

It remains to find the inverse image of $\hat{e}_i$ under the map in~\eqref{CRT1D}. 
We will find it useful to associate each idempotent with its corresponding 
(minimal polynomial) cyclotomic coset rather than its index in some fixed 
ordering of~\eqref{CRT1D}, and we correspondingly denote the pullbacks by 
$e_{C_d}(x)$. There are two ways to compute $e_{C_d}(x)$. First, solve 
$1 = a_{C_d}(x) h_{C_d}(x) + b_{C_d}(x) g_{C_d}(x)$ for $a_{C_d}(x)$ using 
the extended Euclidean algorithm, where $g_{C_d}$ and $h_{C_d}$ are the generator 
and parity-check polynomials of the code with defining set $\mathcal{T} = C_d$, 
respectively. Then, $e_{C_d}(x) = a_{C_d}(x) h_{C_d}(x) \pmod{x^n - 1}$. 
Alternatively, over the splitting field with root $\alpha$,
\begin{equation}
    e_{C_d}(x) = \frac{1}{n} \sum^{n - 1}_{j = 0} \left( \sum_{i \in C_d} \alpha^{-ij} \right) x^j.
\end{equation}
The generator idempotent of $C$ is $e=\sum_{C_{d}\notin \mathcal{T}}e_{C_{d}}$. Its complementary ideal is $ann(C)=(1-e)R$. Since Euclidean duality inverts the root indices, the generator idempotent of $C^{\perp}$ is $e^{\perp}=1-\iota(e)=\iota(1-e)$, where $\iota(f)(x)=f(x^{-1})$. Independently, we always have $R=eR\oplus(1-e)R$.

It follows from the Euclidean algorithm derivation that
\begin{equation}\label{idemind}
    e(\alpha^i) = \begin{cases}
        0 & \text{if $\alpha^i$ is a zero of the code,}\\
        1 & \text{otherwise.}
    \end{cases}
\end{equation}
As such, $e(x)$ acts as an indicator function for roots, and therefore for 
$q$-cyclotomic cosets, of the code. If we think of each cyclotomic coset as a 
frequency component of the code, then $g(x)$ and $e(x)$ generate the code in 
the time and frequency domains, respectively.

\begin{example}\label{ex:1dcycpart3}
    Continuing Example~\ref{ex:1dcycpart2}, the generator polynomials for the ideals corresponding to $C_0$, $C_1$, and $C_3$ are
    \begin{equation*}
        g_{C_0}(x) = 1 + x, \quad g_{C_1}(x) = 1 + x + x^3, \quad g_{C_3}(x) = 1 + x^2 + x^3.
    \end{equation*}
    The main idea in the algorithm to construct the idempotents of the ideals generated by the polynomials $g_{C_0}, g_{C_1}, g_{C_3}$ is to compute the parity-check polynomials $h_{C_i}(x) = (x^7 - 1) / g_{C_i}(x)$. By the extended Euclidean algorithm, there exists a polynomial $a_i(x) \in \F_2[x]$ such that $a_i(x) h_{C_i}(x) \equiv 1 \pmod{g_{C_i}(x)}$. If we set $e_{C_i} = a_i(x) h_{C_i}(x)$, then we obtain a polynomial that satisfies~\eqref{idemind}.

    Explicitly, $h_{C_1}(x) = (x + 1)(x^3 + x^2 + 1) = x^4 + x^2 + x + 1$. Dividing $h_{C_1}(x)$ by $g_{C_1}(x) = x^3 + x + 1$ over $\F_2$ yields $x^4 + x^2 + x + 1 = x(x^3 + x + 1) + 1$. Rearranging this gives $1 = 1 \cdot h_{C_1}(x) + x \cdot g_{C_1}(x)$. Thus, $a_1(x) = 1$, and our idempotent is
    \begin{equation*}
        e_{C_1}(x) = 1 \cdot h_{C_1}(x) = 1 + x + x^2 + x^4.
    \end{equation*}
    
    Next, $h_{C_3}(x) = (x + 1)(x^3 + x + 1) = x^4 + x^3 + x^2 + 1$. The extended Euclidean algorithm with $g_{C_3}(x) = x^3 + x^2 + 1$ yields $1 = (1 + x^2) h_{C_3}(x) + x^3 g_{C_3}(x)$. Thus, $a_3(x) = 1 + x^2$, and the idempotent is
    \begin{equation*}
        e_{C_3}(x) = (1 + x^2)(x^4 + x^3 + x^2 + 1) = 1 + x^3 + x^5 + x^6.
    \end{equation*}

    Similarly, $h_{C_0}(x) = \sum_{j = 0}^6 x^j$. Evaluating at the root $x = 1$ gives $7 \equiv 1 \pmod 2$, meaning $a_0(x) = 1$ and
    \begin{equation*}
        e_{C_0}(x) = 1 + x + x^2 + x^3 + x^4 + x^5 + x^6.
    \end{equation*}

    The generator idempotents of the codes in the running example, given by $e_i(x) = \sum_{C \notin \mathcal{T}_i} e_C(x)$, are therefore:
    \begin{gather*}
        e_0(x) = e_{C_0} + e_{C_1} + e_{C_3} = 1 \\
        e_1(x) = e_{C_1} + e_{C_3} = x + x^2 + x^3 + x^4 + x^5 + x^6 \\
        e_2(x) = e_{C_3} = 1 + x^3 + x^5 + x^6 \\
        e_3(x) = e_{C_1} = 1 + x + x^2 + x^4 \\
        e_4(x) = 0 \\
        e_5(x) = e_{C_0} + e_{C_3} = x + x^2 + x^4 \\
        e_6(x) = e_{C_0} = 1 + x + x^2 + x^3 + x^4 + x^5 + x^6 \\
        e_7(x) = e_{C_0} + e_{C_1} = x^3 + x^5 + x^6.
    \end{gather*}
\end{example}

As seen in the examples, all possible cyclic codes are constructed by including or not including cyclotomic cosets (idempotents). This induces a lattice of containment of codes. Recall that in a commutative ring, for principal ideals $\idlGens{a}$ and $\idlGens{b}$, the inclusion $\idlGens{a} \subseteq \idlGens{b}$ implies $b \mid a$.
\begin{corollary}
    Let $\mathcal{C}_1$ and $\mathcal{C}_2$ be cyclic codes over $\F_q$ with generator polynomials $g_1(x)$ and $g_2(x)$ and defining sets $\mathcal{T}_{\mathcal{C}_1}$ and $\mathcal{T}_{\mathcal{C}_2}$, respectively. Then $\mathcal{C}_1 \subseteq \mathcal{C}_2$ if and only if $g_2(x) \mid g_1(x)$ if and only if $\mathcal{T}_{\mathcal{C}_2} \subseteq \mathcal{T}_{\mathcal{C}_1}$.
\end{corollary}

\begin{example}
    Continuing Example~\ref{ex:1dcycpart3}, the relationships of containment between the cyclic codes form a Boolean lattice with eight elements; the Hasse diagram is shown below.
    
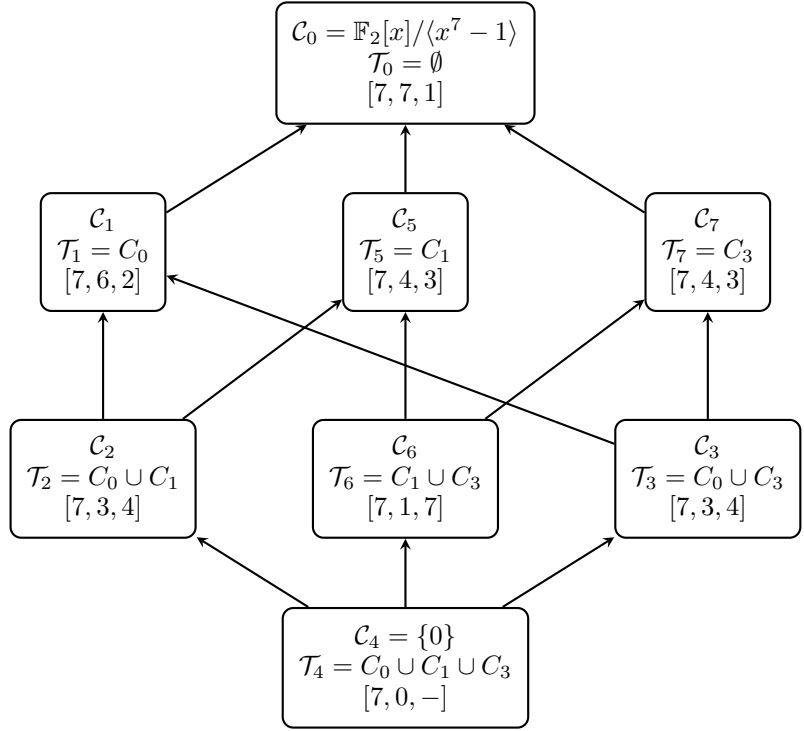
\begin{center}
        \begin{tikzpicture}[
            >=stealth,
            thick,
            box/.style={align=center, inner sep=6pt, draw, rounded corners, fill=white}
        ]
        
            \node[box] (C0) at (0, 4.5) {$\mathcal{C}_0 = \F_2[x]/\idlGens{x^7-1}$ \\ $\mathcal{T}_0 = \emptyset$ \\ $[7, 7, 1]$};
        
            \node[box] (C1) at (-4, 2) {$\mathcal{C}_1$ \\ $\mathcal{T}_1 = C_0$ \\ $[7, 6, 2]$};
            \node[box] (C5) at (0, 2) {$\mathcal{C}_5$ \\ $\mathcal{T}_5 = C_1$ \\ $[7, 4, 3]$};
            \node[box] (C7) at (4, 2) {$\mathcal{C}_7$ \\ $\mathcal{T}_7 = C_3$ \\ $[7, 4, 3]$};
        
            \node[box] (C2) at (-4, -1) {$\mathcal{C}_2$ \\ $\mathcal{T}_2 = C_0 \cup C_1$ \\ $[7, 3, 4]$};
            \node[box] (C6) at (0, -1) {$\mathcal{C}_6$ \\ $\mathcal{T}_6 = C_1 \cup C_3$ \\ $[7, 1, 7]$};
            \node[box] (C3) at (4, -1) {$\mathcal{C}_3$ \\ $\mathcal{T}_3 = C_0 \cup C_3$ \\ $[7, 3, 4]$};
        
            \node[box] (C4) at (0, -3.5) {$\mathcal{C}_4 = \{0\}$ \\ $\mathcal{T}_4 = C_0 \cup C_1 \cup C_3$ \\ $[7, 0, -]$};
        
            \draw[->] (C4) -- (C2);
            \draw[->] (C4) -- (C3);
            \draw[->] (C4) -- (C6);
        
            \draw[->] (C2) -- (C1);
            \draw[->] (C2) -- (C5);
        
            \draw[->] (C3) -- (C1);
            \draw[->] (C3) -- (C7);
        
            \draw[->] (C6) -- (C5);
            \draw[->] (C6) -- (C7);
        
            \draw[->] (C1) -- (C0);
            \draw[->] (C5) -- (C0);
            \draw[->] (C7) -- (C0);
        
        \end{tikzpicture}
    \end{center}
\end{example}

We close the review by noting that there are other, often tighter, lower bounds on the minimum distance. The BCH bound can be strengthened when the defining set contains several parallel strings of consecutive elements. One standard generalization is the Hartmann--Tzeng bound~\cite{HartmannTzeng1972}.
\begin{theorem}[Hartmann--Tzeng Bound]
    \label{thm:HT-bound}
    Let $\mathcal{C}$ be a cyclic code of length $n$ over $\F_q$ with defining set $\mathcal{T}(\mathcal{C})$. Suppose that there exist integers $b, m_1, m_2, \delta$, and $s$ such that $\delta \geq 2$, $s \geq 0$, $\gcd(m_1, n) = 1$, and $b + i m_1 + j m_2 \in \mathcal{T}(\mathcal{C})$ for all $0 \leq i \leq \delta - 2$ and $0 \leq j \leq s$. If $\gcd(m_2, n) < \delta$, then $d(\mathcal{C}) \geq \delta + s$.
\end{theorem}
\noindent When $s = 0$, the Hartmann--Tzeng bound reduces to the BCH bound. Instead of requiring only one run of $\delta - 1$ consecutive roots, it uses several shifted runs of such roots. Thus, the Hartmann--Tzeng bound can be viewed as a multi-string BCH bound.

A further generalization is the Roos bound~\cite{Roos1982, Roos1983}. While there are several equivalent formulations in the literature, it can be conceptually framed as a product-set bound for cyclic codes.

\begin{theorem}[Roos Bound]
    \label{thm:roos-bound}
    Let $\mathcal{C}$ be a cyclic code of length $n$ over $\F_q$ with zero set $\mathcal{Z}(\mathcal{C})$, and let $\alpha$ be a primitive $n$-th root of unity. Suppose $\mathcal{Z}(\mathcal{C})$ contains the product set $MN = \{\mu\nu \mid \mu \in M, \nu \in N\}$, where $M = \{\alpha^b, \alpha^{b+1}, \dots, \alpha^{b+\delta-2}\}$ consists of $\delta - 1$ consecutive powers of $\alpha$, and $N = \{\alpha^{i_0}, \alpha^{i_1}, \dots, \alpha^{i_s}\}$ is a set of $s + 1$ powers of $\alpha$ with $0 = i_0 < i_1 < \dots < i_s$.
    Let $\overline{N}$ be the smallest consecutive set of roots containing $N$. If the number of missing consecutive roots (``holes'') in $N$ satisfies
    \begin{equation}
        |\overline{N}| - |N| \le \delta - 2 \quad \Longleftrightarrow \quad i_s - i_0 \le \delta + s - 2,
    \end{equation}
    then the minimum distance of $\mathcal{C}$ is bounded below by $d(\mathcal{C}) \ge |M| + |N| = \delta + s$.
\end{theorem}

\noindent Under suitable consecutiveness hypotheses on $M$ and the size constraint on $N$, this condition yields a stronger lower bound on the minimum distance. Both the BCH bound and the Hartmann--Tzeng bound can be recovered through specific choices of these subsets.

Another important lower bound is the shift bound of van Lint and Wilson~\cite{vanLintWilson1986}. This bound is often stronger than the BCH, Hartmann--Tzeng, and Roos bounds, but it is somewhat less elementary to state because it is formulated in terms of independent sets with respect to the zero set of a polynomial.
\begin{theorem}[Shift Bound]
    \label{thm:shift-bound}
    Let $0 \neq f(x) \in K[x]$, where $K$ is a field, and let $S = \{\theta \in K \mid f(\theta) = 0 \}$. If $A$ is an independent set with respect to $S$ in the sense of van Lint and Wilson, then $\wt(f) \geq |A|$. Consequently, if every nonzero codeword $c(x) \in \mathcal{C}$ has a zero set containing a set $S$ for which one can construct an independent set $A$ of size $D$, then $d(\mathcal{C}) \geq D$.
\end{theorem}
\noindent The shift bound provides a general method for proving lower bounds on the minimum distance of cyclic codes directly from their zero structure.

Finally, for the duals of BCH codes, we can use character-sum estimates. A BCH code of length $n$ over $\F_q$ is called \emph{primitive} if $n = q^m - 1$ for some integer $m \geq 1$. The classical Carlitz--Uchiyama bound~\cite{Carlitz1957} provides estimates for exponential sums and yields lower bounds on the weights of dual BCH codewords. Specifically, for a primitive binary BCH code of length $n = 2^m - 1$ and designed distance $\delta = 2t + 1$, the Carlitz--Uchiyama theorem restricts the weights of the nonzero codewords in the dual code, which immediately gives a lower bound on its minimum distance.
\begin{theorem}[Carlitz--Uchiyama Bound]
    \label{thm:carlitz-uchiyama}
    Let $\mathcal{C}$ be a primitive binary BCH code of length $n = 2^m - 1$ and designed distance $\delta = 2t + 1$. If $w$ is the weight of any nonzero codeword in the dual code $\mathcal{C}^\perp$, then
    \begin{equation*}
        2^{m - 1} - (t - 1) 2^{m / 2} \leq w \leq 2^{m - 1} + (t - 1) 2^{m / 2}.
    \end{equation*}
\end{theorem}

\section{2D Cyclic Codes}
\label{sec:2dcyclic}
The theory of two-dimensional cyclic codes is more complex than the one-dimensional case.
\begin{definition}
    A two-dimensional (2D) cyclic code is an ideal in $R = \F_q[x, y] / \idlGens{x^\ell - 1, y^m - 1}$.
\end{definition}
\noindent The ring $\F_q[x, y]$ is not a Euclidean domain and ideals are not principal. A single code may therefore be defined by a number of ``generator'' polynomials $\idlGens{g_1, \dots, g_r}$. The majority of the existing 2D cyclic code literature is concentrated on establishing canonical generators and determining the code parameters through polynomial algorithms such as Gröbner basis \cite{cox2015ideals}. In the special case that $R$ is semisimple, we are able to mimic the one-dimensional case. Relevant results on idempotent generators may be found sprinkled throughout the classical group algebra codes literature. Given the relationship between two-dimensional cyclic codes ($R$) and BB codes ($R \times R$), we include a self-contained, comprehensive theory of semisimple 2D cyclic codes that mimics the 1D case as closely as possible. This theory will be used in the following sections to analyze and construct BB codes.

By the Chinese Remainder Theorem, if $\gcd(\ell, m) = 1$, the underlying coordinate group $\Z_\ell \times \Z_m$ is isomorphic to the cyclic group $\Z_{\ell m}$. As a result, the ring $R$, which is isomorphic to the group algebra $\F_q[\Z_\ell \times \Z_m]$, is isomorphic to the 1D group algebra $\F_q[\Z_{\ell m}] \cong \F_q[z] / \idlGens{z^{\ell m} - 1}$.  In this scenario, the 2D grid effectively ``unrolls'' into a single line, and the study of 2D cyclic codes reduces entirely to the classical theory of 1D cyclic codes of length $\ell m$. Therefore, a genuinely two-dimensional theory only emerges when $\gcd(\ell, m) \neq 1$.

Rigorously, let $R = \F_q[x, y] / \idlGens{x^\ell - 1, y^m - 1}$ and consider the case where $\gcd(\ell, m) = 1$. Set $s = \ell m$, define a ring homomorphism $\Phi:\F_q[x, y] \longrightarrow \F_q[z]/ \idlGens{z^s - 1}$ by its action on generators,
    $\Phi(x) = z^m$, $\Phi(y) = z^\ell$,
and extend $\F_q$-linearly and multiplicatively. For a polynomial $f(x,y) = \sum_{i, j} a_{i, j }x^i y^j$, this gives
\begin{equation*}
    \Phi(f) = \sum_{i, j} a_{i, j} z^{i m + j \ell} \pmod{z^s - 1}.
\end{equation*}
Since
\begin{equation*}
    \Phi(x^\ell - 1) = (z^m)^\ell - 1 = z^{\ell m} - 1 = z^s - 1 \equiv 0,
\end{equation*}
\begin{equation*}
    \Phi(y^m - 1) = (z^{\ell})^m - 1 = z^{\ell m} - 1 = z^s - 1 \equiv 0,
\end{equation*}
the ideal $\idlGens{x^\ell - 1, y^m - 1}$ is contained in $\ker\Phi$. Hence $\Phi$ induces a well-defined ring homomorphism
\begin{equation*}
    \overline{\Phi}: \F_q[x, y] / \idlGens{x^\ell - 1, y^m - 1} \longrightarrow \F_q[z]/\idlGens{z^s - 1}
\end{equation*}
given by
\begin{equation*}
    \overline{\Phi}\big(\overline{f(x, y)}\big) = f\big(z^{m}, z^{\ell}\big) \pmod{z^s - 1},
\end{equation*}
where $\overline{f(x,y)}$ stands for the image of $f(x, y)$ in $R$.

Since $\gcd(\ell, m) = 1$, there exist integers $u, v \in \Z$ such that $u \ell + v m = 1$. In $\F_q[z] / \idlGens{z^s - 1}$, the element $z$ is a unit, so we have 
\begin{equation*}
    z = (z^m)^v (z^\ell)^u = \overline{\Phi}(x)^v \, \overline{\Phi}(y)^u.
\end{equation*}
Therefore $z \in \im \overline{\Phi}$. It follows that $\overline{\Phi}$ is surjective since $\F_q[z] / \idlGens{z^s - 1}$ is generated as a ring by $z$. In this case, all 2D cyclic codes in $R$ are essentially 1D cyclic codes in $\F_q[z] / \idlGens{z^s - 1}$.

Assuming that $\gcd(\ell, m) \neq 1$, the structure of the ring $R = \F_q[x, y] / \idlGens{x^\ell - 1, y^m - 1}$ fundamentally bifurcates based on whether it is semisimple. By examining the formal derivatives, $x^\ell - 1$ and $y^m - 1$ have distinct roots over the algebraic closure of $\F_q$ if and only if the characteristic of $\F_q$ ($p$) divides neither $\ell$ nor $m$. Since $q$ is a power of the characteristic, $R$ is semisimple if and only if both $\gcd(\ell, q) = 1$ and $\gcd(m, q) = 1$. This is concisely  captured by a single requirement: $\gcd(\ell m, q) = 1$. Consequently, the study of 2D cyclic codes naturally splits into two distinct cases: the semisimple case where $\gcd(\ell m, q) = 1$, and the non-semisimple (or repeated-root) case where $\gcd(\ell m, q) \neq 1$. For the rest of the paper, we will assume $R$ is semisimple and leave the non-semisimple case to future work.

\subsection{Orbits}
Let $\alpha$ and $\beta$ be primitive $\ell$-th and $m$-th roots of unity, respectively, in a suitable extension field of $\F_q$.
\begin{definition}
    The common zero set of the polynomials $x^\ell - 1$ and $y^m - 1$ is the \emph{root grid}
    \begin{equation*}
        \mathcal{G}_{\alpha, \beta} = \{(\alpha^i, \beta^j) \mid 0 \leq i < \ell, \ 0 \leq j < m\}.
    \end{equation*}
    For an ideal $I \subseteq R$, the \emph{zero set} of $I$ is
    \begin{equation*}
        \mathcal{Z}(I) = \{(\mu, \nu) \in \mathcal{G}_{\alpha,\beta} \mid f(\mu, \nu) = 0 \ \text{for all } f \in I\}.
    \end{equation*}
    Equivalently, $(\mu, \nu) \in \mathcal{Z}(I)$ if and only if every codeword polynomial vanishes when evaluated at $(\mu, \nu)$. 
\end{definition}

\begin{remark}
    In general, an ideal of a polynomial ring is not uniquely determined by its set of common roots. Hilbert's  Nullstellensatz says this happens if and only if the ideal is radical.
\end{remark}

Orbits will generalize $q$-cosets to two dimensions.
\begin{definition}[$q$-Frobenius Orbits]
    \label{def:Frobenius-orbits}
    The \emph{$q$-Frobenius orbit} of $(i, j) \in \Z_\ell \times \Z_m$ is
    \begin{equation*}
    O_{(i, j)} = \{(q^t i, q^t j) \in \Z_\ell \times \Z_m \mid t \geq 0\}.
    \end{equation*}
    The set of all $q$-Frobenius orbits in $\Z_\ell \times \Z_m$ will be denoted by $\Omega = \mathrm{Orb}_q(\Z_\ell \times \Z_m)$.
\end{definition}
\noindent Similar to cyclotomic cosets, we will often shorten ``$q$-Frobenius orbit'' as appropriate to avoid repetition when the meaning is apparent from context.

\begin{example}\label{example:5x7grid}
    Figure~\ref{fig:orbit-grid-4-5-7-colored} shows the 4-Frobenius orbits of the grid $\Z_5 \times \Z_7$.
  
    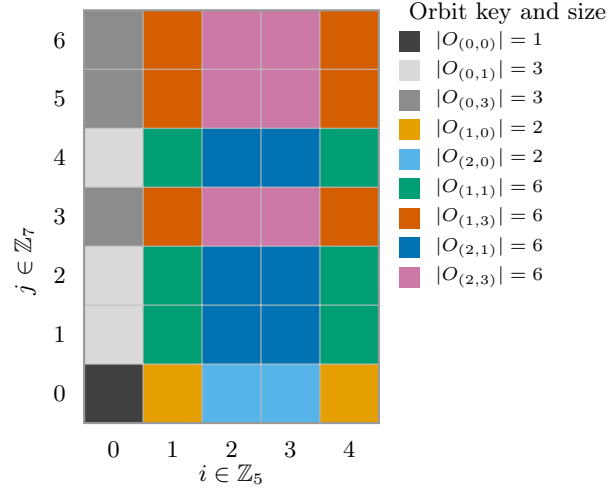
\begin{figure}[htp]
        \centering
        \begin{tikzpicture}[x=0.78cm,y=0.78cm]

            \tikzset{
              orb0/.style={fill=cbDarkGray},    
              orbA/.style={fill=cbLightGray},   
              orbB/.style={fill=cbMediumGray},  
              orbC/.style={fill=cbOrange},      
              orbD/.style={fill=cbSkyBlue},     
              orbE/.style={fill=cbBluishGreen}, 
              orbF/.style={fill=cbVermillion},  
              orbG/.style={fill=cbBlue},        
              orbH/.style={fill=cbRedPurple}    
            }
            
            \newcommand{\cell}[3]{\fill[#3] (#1,#2) rectangle ++(1,1);}
            
            \cell{0}{0}{orb0}
            
            \foreach \i/\j in {0/1,0/4,0/2} { \cell{\i}{\j}{orbA} }
            
            \foreach \i/\j in {0/3,0/5,0/6} { \cell{\i}{\j}{orbB} }
            
            \foreach \i/\j in {1/0,4/0} { \cell{\i}{\j}{orbC} }
            
            \foreach \i/\j in {2/0,3/0} { \cell{\i}{\j}{orbD} }
            
            \foreach \i/\j in {1/1,4/4,1/2,4/1,1/4,4/2} { \cell{\i}{\j}{orbE} }
            
            \foreach \i/\j in {1/3,4/5,1/6,4/3,1/5,4/6} { \cell{\i}{\j}{orbF} }
            
            \foreach \i/\j in {2/1,3/4,2/2,3/1,2/4,3/2} { \cell{\i}{\j}{orbG} }
            
            \foreach \i/\j in {2/3,3/5,2/6,3/3,2/5,3/6} { \cell{\i}{\j}{orbH} }
            
            \draw[step=1,gray!50] (0,0) grid (5,7);
            \draw[gray!80,thick] (0,0) rectangle (5,7);
            
            \foreach \i in {0,...,4} {
              \node[below] at (\i+0.5,-0.15) {\small $\i$};
            }
            \foreach \j in {0,...,6} {
              \node[left] at (-0.15,\j+0.5) {\small $\j$};
            }
            \node[below] at (2.5,-0.55) {\small $i\in\Z_5$};
            \node[left,rotate=90] at (-1,3.5) {\small $j\in\Z_7$};
            
            \begin{scope}[shift={(5.35,0.35)}]
            \node[anchor=west] at (0,6.6) {\small Orbit key and size};
            
            \fill[orb0] (0,5.95) rectangle (0.35,6.30);
            \node[anchor=west] at (0.45,6.13) {\scriptsize $|O_{(0,0)}|=1$};
            
            \fill[orbA] (0,5.45) rectangle (0.35,5.80);
            \node[anchor=west] at (0.45,5.63) {\scriptsize $|O_{(0,1)}|=3$};
            
            \fill[orbB] (0,4.95) rectangle (0.35,5.30);
            \node[anchor=west] at (0.45,5.13) {\scriptsize $|O_{(0,3)}|=3$};
            
            \fill[orbC] (0,4.45) rectangle (0.35,4.80);
            \node[anchor=west] at (0.45,4.63) {\scriptsize $|O_{(1,0)}|=2$};
            
            \fill[orbD] (0,3.95) rectangle (0.35,4.30);
            \node[anchor=west] at (0.45,4.13) {\scriptsize $|O_{(2,0)}|=2$};
            
            \fill[orbE] (0,3.45) rectangle (0.35,3.80);
            \node[anchor=west] at (0.45,3.63) {\scriptsize $|O_{(1,1)}|=6$};
            
            \fill[orbF] (0,2.95) rectangle (0.35,3.30);
            \node[anchor=west] at (0.45,3.13) {\scriptsize $|O_{(1,3)}|=6$};
            
            \fill[orbG] (0,2.45) rectangle (0.35,2.80);
            \node[anchor=west] at (0.45,2.63) {\scriptsize $|O_{(2,1)}|=6$};
            
            \fill[orbH] (0,1.95) rectangle (0.35,2.30);
            
            \node[anchor=west] at (0.45,2.13) {\scriptsize $|O_{(2,3)}|=6$};
            \end{scope}
        \end{tikzpicture}
        \caption{The $4$-Frobenius orbit decomposition of $\Z_5 \times \Z_7$. Cells with the same shading lie in the same orbit.}
        \label{fig:orbit-grid-4-5-7-colored}
    \end{figure}
\end{example}

As before, we can say something about the parameters of the cyclic code by counting roots. We collect a few results in this direction. For $i \in \Z_\ell$ and $j \in \Z_m$, define 
\begin{equation*}
    \mathbf{t}_\ell (i) = \frac{\ell}{\gcd(\ell, i)}, \qquad \mathbf{t}_m(j) = \frac{m}{\gcd(m, j)}.
\end{equation*}
\begin{lemma}\label{lem:orbit-length-rigorous}
    Let $(i, j)\in \Z_\ell \times \Z_m$. Then, $|O_{(i, j)}| = \mathrm{lcm} \bigl( \mathrm{ord}_{\mathbf{t}_\ell (i)}(q), \mathrm{ord}_{\mathbf{t}_m(j)}(q)\bigr)$.
\end{lemma}
\begin{proof}
    Let $g = \gcd(\ell, i)$ so that $i = g u$ for some integer $u$ with $\gcd(u, \ell / g) = 1$. Write $d = \mathbf{t}_\ell (i) = \ell / g$.  Then
    \begin{equation*}
        q^t i\equiv i \pmod \ell \iff \ell \mid (q^t - 1) i \iff \ell \mid (q^t - 1) g u \iff d \mid (q^t - 1) u.
    \end{equation*}
    Since $\gcd(u, d) = 1$, $u$ is a unit modulo $d$, so $d \mid (q^t - 1) u$ is equivalent to $d \mid (q^t - 1)$, or $q^t \equiv 1 \pmod d$. Therefore the least positive $t$ with $q^t i \equiv i \pmod \ell$ is $\mathrm{ord}_d(q)$.
    
    Denoting $e=\mathbf{t}_m(j)$, 
    the same argument shows that the least positive $t^\prime$ with $q^{t^\prime} j \equiv j \pmod m$ is
    $\mathrm{ord}_{e}(q)$. The orbit length is the least $t$ satisfying both congruences. But this is the least
    common multiple of these two orders.
\end{proof}

This result generalizes the size of the $q$-cosets above. If $i = 0$, then $\mathbf{t}_\ell (i) = 1$ and $\mathrm{ord}_{\mathbf{t}_\ell (i)}(q) = 1$, hence $|O_{(0, j)}| = \mathrm{ord}_{ \mathbf{t}_m(j)}(q)$. Similarly, $|O(i, 0)| = \mathrm{ord}_{\mathbf{t}_\ell (i)}(q)$. 
\medskip

To systematically count the roots and determine the size of the orbits, it is helpful to partition the index set into strata based on the divisors defined above. 
For $d \mid \ell, \, e \mid m$ we define 
\begin{equation*}
    \Lambda_{d, e} = \{(i, j) \in \Z_\ell \times \Z_m \mid \mathbf{t}_\ell(i) = d, \mathbf{t}_m(j) = e\}.
\end{equation*} 
By grouping together the coordinate pairs that share the same respective additive orders, we can count the elements in each stratum directly.

\begin{lemma}\label{lem:strata-sizes-rigorous}
    For each $d \mid \ell$ and $e \mid m$, we have $|\Lambda_{d, e}| = \phi(d) \phi(e)$, where $\phi(\cdot)$ is Euler's totient function.
\end{lemma}
\begin{proof}
    The condition $\mathbf{t}_\ell (i) = d$ is equivalent to $\gcd(\ell, i) = \ell / d$, so, $i = \frac{\ell}{d}u$ with $\gcd(u, d) = 1$. Since distinct $u \in (\Z / d \Z)^\times$ yield distinct residues $i \in \Z_\ell$, there are exactly $\phi(d)$ choices for $i$. Similarly, there are $\phi(e)$ choices for $j$, and the choices are independent. The result follows.
\end{proof}

\begin{proposition}\label{prop:orbit-census-by-divisors}
    We follow the notation of Lemma~\ref{lem:strata-sizes-rigorous}. Then, all points in $\Lambda_{d, e}$ have the same orbit length $L(d, e) = \lcm(\ord_d(q), \ord_e(q))$.
    Consequently, the number of $q$-Frobenius orbits contained in $\Lambda_{d,e}$ is
    \begin{equation*}
        \frac{|\Lambda_{d, e}|}{L(d, e)} = \frac{\phi(d) \phi(e)}{\lcm(\ord_d(q), \ord_e(q))}.
    \end{equation*}
    In particular, the multiset of orbit sizes in $\Z_\ell \times \Z_m$ is determined by the values $\ord_d(q)$ for $d \mid \ell$ and $\ord_e(q)$ for $e \mid m$.
\end{proposition}
\begin{proof}
    If $(i, j) \in \Lambda_{d, e}$ then $d(i) = d$ and $e(j) = e$, so by Lemma~\ref{lem:orbit-length-rigorous} we have 
    \begin{equation*}
        |O_{(i, j)}| = \lcm(\ord_d(q), \ord_e(q)) = L(d, e),
    \end{equation*}
    which depends only on $(d, e)$ and not on $(i, j)$ within the stratum. Since $\Lambda_{d, e}$ is a disjoint union of orbits all of the same size $L(d, e)$, the number of orbits in $\Lambda_{d, e}$ is $|\Lambda_{d, e}| / L(d, e)$, and Lemma~\ref{lem:strata-sizes-rigorous} gives $|\Lambda_{d, e}| = \phi(d) \phi(e)$.
\end{proof}

Let $\varphi : \Z_\ell \times \Z_m \to \Z_\ell \times \Z_m$ be the permutation corresponding to the Frobenius automorphism so that $\varphi(i, j) = (qi, qj)$. Let $\mathrm{Fix}(\varphi^t)$ denote the number of points of $\Z_\ell \times \Z_m$ fixed by $\varphi^t$, 
\begin{equation*}
    \mathrm{Fix}(\varphi^t) = \#\{(i, j) \in \Z_\ell \times \Z_m \mid \varphi^t(i, j) = (i, j)\}.
\end{equation*}

\begin{lemma}\label{lem:fixed-points-rigorous}
    Let $t$ be a nonnegative integer. Then we have 
    \begin{align*}
        \mathrm{Fix}(\varphi^t) &= \#\{(i, j) \in \Z_\ell \times \Z_m \mid q^t i \equiv i \!\!\! \pmod \ell, \, q^t j \equiv j \!\!\! \pmod m \}\\
        &= \gcd(\ell, q^t - 1) \gcd(m, q^t - 1).
    \end{align*}
\end{lemma}
\begin{proof}
    The condition $\varphi^t(i, j) = (i, j)$ is equivalent to
    \begin{equation*}
        (q^t - 1) i \equiv 0 \pmod \ell, \quad (q^t - 1) j \equiv 0 \pmod m.
    \end{equation*}
    For a general modulus $n$ and integer $a$, the congruence $a x \equiv 0 \pmod n$ has exactly $\gcd(n, a)$ solutions modulo $n$. Applying this with $(n, a) =(\ell, q^t - 1)$ gives $\gcd(\ell, q^t - 1)$ choices for $i$, and similarly $\gcd(m, q^t - 1)$ choices for $j$. Independence of the coordinates yields the product.
\end{proof}

Define $s = \lcm(\ord_\ell(q), \ord_m(q))$. Then $\varphi^s = \mathrm{id}_{\Z_\ell \times \Z_m}$, so the cyclic group $\langle \varphi \rangle \cong \Z_s$ acts on $\Z_\ell \times \Z_m$. The following lemma is well-known.
\begin{lemma}[Burnside's Lemma]
\label{lem:burnside}
    If a finite group $G$ acts on a finite set $X$, the number of orbits is
    \begin{equation*}
        \#(X / G) = \frac{1}{|G|} \sum_{g \in G} \mathrm{Fix}(g).
    \end{equation*}
\end{lemma}

\begin{proposition}\label{prop:burnside}
    The number of $q$-Frobenius orbits in $\Z_\ell \times \Z_m$ is
    \begin{equation*}
        \#\Omega = \frac{1}{s} \sum_{t = 0}^{s - 1} \mathrm{Fix}(\varphi^t) = \frac{1}{s} \sum_{t = 0}^{s - 1} \gcd(\ell, q^t - 1) \gcd(m, q^t - 1).
    \end{equation*}
\end{proposition}
\begin{proof}
    Apply Burnside's Lemma (Lemma~\ref{lem:burnside}) to the cyclic group $G = \langle \varphi \rangle$ acting on $\Z_\ell \times \Z_m$.
\end{proof}

Since $\varphi^s = \mathrm{id}_{\Z_\ell \times \Z_m}$, every orbit length divides $s$. For each divisor $d\mid s$ set
\begin{equation*}
    A(d) = \mathrm{Fix}(\varphi^d) = \gcd(\ell, q^d - 1) \gcd(m, q^d - 1),
\end{equation*}
and let $B(d)$ be the number of points whose orbit has exactly length $d$.
\begin{lemma}\label{lem:AdBd}
    For every $d \mid s$ one has
        $A(d) = \sum_{r \mid d} B(r)$.
\end{lemma}
\begin{proof}
    A point $x \in \Z_\ell \times \Z_m$ is fixed by $\varphi^d$ if and only if $\varphi^d(x) = x$, which holds if and only if the orbit length of $x$ divides $d$. Thus, among the points fixed by $\varphi^d$, the orbit length can be any divisor $r$ of $d$, and summing the number of points of each exact orbit length $r$ gives $A(d) = \sum_{r \mid d}B(r)$.
\end{proof}

\begin{proposition}\label{prop:mobius-rigorous}
    For each $d \mid s$,
    \begin{equation*}
        B(d) = \sum_{r \mid d} \mu(r) A \left(\frac{d}{r} \right), \quad\text{and hence}\quad \#\{\text{orbits of length } d\} = \frac{B(d)}{d},
    \end{equation*}
    where $\mu$ is the M\"obius function.
\end{proposition}
\begin{proof}
    Lemma~\ref{lem:AdBd} gives $A(d )= \sum_{r \mid d} B(r)$ on the divisor poset of $d$. M\"obius inversion on the divisor lattice implies
    \begin{equation*}
        B(d) = \sum_{r \mid d} \mu(r) A(d / r).
    \end{equation*}
    Finally, each orbit of length $d$ contains exactly $d$ points, so the number of such orbits is $B(d) / d$.
\end{proof}

In the semisimple case, the ring $R$ decomposes into a direct sum of finite fields. The structure of this decomposition is entirely dictated by the $q$-Frobenius orbits of the index set. There is a one-to-one correspondence between the distinct $q$-Frobenius orbits and the minimal idempotent generators of $R$, and the length of each orbit determines the degree of the corresponding field extension. Since BB codes are defined by polynomials in $R$, their dimensions and minimum distances depend on how these polynomials project onto the ring's irreducible components.

\begin{example}
    Let $q = 2$, $\ell = 3$, and $m = 3$. These choices satisfy the semisimplicity condition $\gcd(\ell m, q) = \gcd(9, 2) = 1$, but since $\gcd(\ell, m) \neq 1$, the ring does not trivially collapse into a 1D structure. The total number of points in our index set is $9$. We will demonstrate the three propositions above by applying them to $R$ to count the $q$-Frobenius orbits and their lengths.

    \begin{enumerate} 
        \item Proposition~\ref{prop:orbit-census-by-divisors}: The divisors of both $\ell$ and $m$ are $\{1, 3\}$. The multiplicative orders of $q = 2$ modulo these divisors are $\ord_1(2) = 1$ and $\ord_3(2) = 2$. We evaluate the strata $\Lambda_{d, e}$ for all combinations of divisors:
        \begin{itemize}
            \item For $d = 1, e = 1$: $L(1, 1) = \lcm(1, 1) = 1$. The number of points is $\phi(1)\phi(1) = 1$. This yields $1 / 1 = 1$ orbit.
            \item For $d = 1, e = 3$: $L(1, 3) = \lcm(1, 2) = 2$. The number of points is $\phi(1)\phi(3) = 2$. This yields $2 / 2 = 1$ orbit.
            \item For $d = 3, e = 1$: $L(3, 1) = \lcm(2, 1) = 2$. The number of points is $\phi(3)\phi(1) = 2$. This yields $2 / 2 = 1$ orbit.
            \item For $d = 3, e = 3$: $L(3, 3) = \lcm(2, 2) = 2$. The number of points is $\phi(3)\phi(3) = 4$. This yields $4 / 2 = 2$ orbits.
        \end{itemize}
        Summing the orbits from each stratum, we find exactly $1 + 1 + 1 + 2 = 5$ distinct $q$-Frobenius orbits.
        \item Proposition~\ref{prop:burnside}: The maximum orbit length is bounded by $s = \lcm(\ord_3(2), \ord_3(2)) = 2$. We count the fixed points $\mathrm{Fix}(\varphi^t) = \gcd(3, 2^t - 1) \gcd(3, 2^t - 1)$ for $0 \leq t < 2$:
        \begin{itemize}
            \item For $t = 0$: $\gcd(3, 0) \gcd(3, 0) = 3 \times 3 = 9$.
            \item For $t = 1$: $\gcd(3, 1) \gcd(3, 1) = 1 \times 1 = 1$.
        \end{itemize}
        Applying Burnside's formula, the total number of orbits is $\frac{1}{2}(9 + 1) = 5$.
        \item Proposition~\ref{prop:mobius-rigorous}: We evaluate $A(d) = \gcd(3, 2^d - 1)\gcd(3, 2^d - 1)$ for the divisors of $s = 2$, which are $d \in \{1, 2\}$:
        \begin{itemize}
            \item $A(1) = \gcd(3, 1) \gcd(3, 1) = 1$.
            \item $A(2) = \gcd(3, 3) \gcd(3, 3) = 9$.
        \end{itemize}
        We apply M\"obius inversion to find $B(d)$, the exact number of points belonging to orbits of length $d$:
        \begin{itemize}
            \item $B(1) = \mu(1)A(1) = (1)(1) = 1$. This means there is $1 / 1 = 1$ orbit of length $1$.
            \item $B(2) = \mu(1)A(2) + \mu(2)A(1) = (1)(9) + (-1)(1) = 8$. This means there are $8 / 2 = 4$ orbits of length $2$.
        \end{itemize}
        Therefore, we have one orbit of length $1$ (the zero point) and four orbits of length $2$, totaling 5 orbits. 
    \end{enumerate}
\end{example}

\subsection{Idempotents}
Let $K\supseteq \F_q$ be a splitting field containing a primitive $\ell$-th root of unity $\alpha$
and a primitive $m$-th root of unity $\beta$, and set
\begin{equation*}
    R_K = K \otimes_{\F_q} R \cong K[x, y] / \idlGens{x^\ell - 1, y^m - 1}.
\end{equation*}
Since $\gcd(q, \ell) = \gcd(q, m) = 1$, both $x^\ell - 1$ and $y^m - 1$ are squarefree in characteristic
$\mathrm{char}(K)$ and split completely over $K$ into distinct linear factors. Hence the evaluation map
\begin{equation*}
    \operatorname{ev}: R_K \longrightarrow K^{\ell m}, \qquad f \longmapsto \bigl(f(\alpha^i, \beta^j)\bigr)_{(i, j) \in \Z_\ell \times \Z_m}
\end{equation*}
is a $K$-algebra isomorphism, identifying $R_K$ with a direct product of $\ell m$ copies of $K$.

For $(i, j) \in \Z_\ell \times \Z_m$, we denote by $e_{i, j}$ the element of $R_K$ that corresponds to the function defined by
\begin{equation}
    \label{2d_idemp_prop}
    e_{i, j}(\alpha^{i^\prime}, \beta^{j^\prime}) =
    \begin{cases}
        1, & (i^\prime, j^\prime) = (i, j),\\
        0, & \text{otherwise.}
    \end{cases}
\end{equation}
Then the set $\{e_{i, j}\}_{(i, j) \in \Z_\ell \times \Z_m}$ is a \emph{complete set of primitive orthogonal idempotents} in $R_K$. This means that the following four properties hold:
\begin{enumerate}
    \item $e_{i, j}^2 = e_{i, j}$ for every $(i, j)\in \Z_\ell \times \Z_m$,
    \item for every $(i, j) \neq (i^\prime, j^\prime)$ from $\Z_\ell \times \Z_m$, we have $e_{i, j} e_{i^\prime, j^\prime} = 0$,
    \item for every $(i, j) \in \Z_\ell \times \Z_m$, $e_{i, j}$ cannot be written as a sum of two other nonzero orthogonal idempotents, 
    \item $\displaystyle \sum_{(i, j) \in \Z_\ell \times \Z_m} e_{i, j} = 1$.
\end{enumerate}

Let $\sigma : K \to K$ be the $q$-Frobenius field automorphism defined by $\sigma(a) = a^q$, extended coefficient-wise to $R_K$ (still denoted $\sigma$). Applying $\sigma$ to the roots of unity yields $\sigma(\alpha) = \alpha^q$ and $\sigma(\beta) = \beta^q$. Consequently, $\sigma$ acts on the primitive idempotents by shifting their indices exactly according to the permutation $\varphi$ defined earlier: $\sigma(e_{i,j}) = e_{\varphi(i,j)} = e_{qi,qj}$ (with indices reduced modulo $l$ and $m$, respectively). Thus, the field automorphism $\sigma$ permutes the primitive idempotents precisely along the $q$-Frobenius orbits governed by $\varphi$.

There are two explicit formulas for the idempotents; the first is based on Lagrange interpolation and the second on the discrete Fourier transform.  
\begin{lemma}
\label{lem:lagrange-eij}
    For each $(i, j) \in \Z_\ell \times \Z_m$, the primitive idempotent is given by
    \begin{equation}\label{eq:lagrange-eij}
        e_{i, j}(x, y) = L_i(x) M_j(y) \in R_K, 
    \end{equation}
    where
    \begin{equation*}
        L_i(x) = \prod_{\substack{u \in \Z_\ell\\ u \neq i}} \frac{x - \alpha^u}{\alpha^i - \alpha^u} \in K[x] / \idlGens{x^\ell - 1}, \qquad M_j(y) = \prod_{\substack{v \in \Z_m\\ v \neq j}} \frac{y - \beta^v}{\beta^j - \beta^v} \in K[y] / \idlGens{y^m - 1}.
    \end{equation*}
    In particular,
    \begin{equation*}
    e_{i, j}\left(\alpha^{i^\prime}, \beta^{j^\prime}\right) = \delta_{i, i^\prime} \delta_{j, j^\prime} \qquad \forall (i^\prime, j^\prime) \in \Z_\ell \times \Z_m.
    \end{equation*}
\end{lemma}
\begin{proof}
    Since $x^\ell - 1 = \prod_{u \in \Z_\ell} (x - \alpha^u)$ has distinct roots, the polynomials $\{L_i(x)\}_{i \in \Z_\ell}$ form the standard (Lagrange) basis for functions on $\{\alpha^u\}_{u \in \Z_\ell}$. Indeed, we have $L_i\left(\alpha^{i^\prime}\right) = \delta_{i, i^\prime}$ for all $i^\prime \in \Z_\ell$. Similarly, $M_j\left(\beta^{j^\prime}\right) = \delta_{j, j^\prime}$ for all $j^\prime \in \Z_m$. Therefore, we have
    \begin{equation*}
        (L_iM_j)(\alpha^{i^\prime}, \beta^{j^\prime}) = L_i(\alpha^{i^\prime}) M_j(\beta^{j^\prime}) = \delta_{i, i^\prime} \delta_{j, j^\prime}.
    \end{equation*}
    Hence $L_i(x) M_j(y)$ has the defining evaluation property of $e_{i, j}$ under $\ev$. By the uniqueness of the coordinate idempotents in $K^{\ell m}$ and the fact that $\ev$ is an isomorphism, it follows that $e_{i, j} = L_i(x) M_j(y)$ in $R_K$.
\end{proof}

\begin{lemma}
\label{lem:fourier-eij}
    Let $\alpha$ be a primitive $\ell$-th root of unity and $\beta$ a primitive $m$-th root of unity in $K$. Let $(i, j) \in \Z_\ell \times \Z_m$. Then, 
    \begin{equation}\label{eq:fourier-eij}
        e_{i, j} = \frac{1}{\ell m} \sum_{a \in \Z_\ell} \sum_{b \in \Z_m} \alpha^{-ia} \beta^{-jb} x^a y^b \in R_K.
    \end{equation}
\end{lemma}
\begin{proof}
    We will evaluate the element 
    \begin{equation*}
    E = \frac{1}{\ell m} \sum_{a \in \Z_\ell} \sum_{b \in \Z_m} \alpha^{-ia} \beta^{-jb} x^a y^b \in R_K
    \end{equation*}
    at an arbitrary point $(\alpha^{i^\prime},\beta^{j^\prime})$:
    \begin{equation*}
    E(\alpha^{i^\prime}, \beta^{j^\prime}) = \frac{1}{\ell m} \sum_{a \in \Z_\ell} \sum_{b \in \Z_m} \alpha^{-ia} \beta^{-jb} \alpha^{i^\prime a} \beta^{j^\prime b} = \frac{1}{\ell m} \left(\sum_{a \in \Z_\ell} \alpha^{(i^\prime - i)a} \right) \left(\sum_{b \in \Z_m} \beta^{(j^\prime - j)b}\right).
    \end{equation*}
    Each inner sum is $\ell$ ($m$) if the exponent is $0$ modulo $\ell$ ($m$), and $0$ otherwise. 
    This follows from the fact that these inner sums are geometric series over all roots of unity. 
    Hence, we see that  $E\left(\alpha^{i^\prime}, \beta^{j^\prime}\right) = \delta_{i, i^\prime} \delta_{j, j^\prime}$, so $E$ has the defining evaluation property of $e_{i, j}$.  Therefore $E = e_{i, j}$ in $R_K$.
\end{proof}

\begin{example}
\label{ex:orbit-idempotents-lm-3-A}
Let $q=2$, let $\ell=m=3$, and let
\[
K=\F_4=\F_2(\omega),
\qquad
\omega^2+\omega+1=0.
\]
Then $\omega^3=1$ and $\omega$ is a primitive third root of unity. We set 
\[
\alpha=\beta=\omega.
\]
Since $\operatorname{char}(K)=2$, the scalar $1/9$ appearing in the Fourier formula is equal to $1$ in $K$. Hence Lemma~\ref{lem:fourier-eij} gives
\[
e_{i,j}
=
\sum_{a=0}^{2}\sum_{b=0}^{2}
\omega^{-ia}\omega^{-jb}x^ay^b.
\]
For example,
\[
e_{0,0}
=
(1+x+x^2)(1+y+y^2).
\]
Indeed,
\[
e_{0,0}(\omega^{i^\prime},\omega^{j^\prime})
=
(1+\omega^{i^\prime}+\omega^{2i^\prime})
(1+\omega^{j^\prime}+\omega^{2j^\prime}).
\]
The factor $1+\omega^{i^\prime}+\omega^{2i^\prime}$ is equal to $1$ if $i^\prime=0$ and is equal to $0$ if $i^\prime=1,2$. The same holds for the $y$-factor. 
Therefore, 
\[
e_{0,0}(\omega^{i^\prime},\omega^{j^\prime})
=
\delta_{0,i^\prime}\delta_{0,j^\prime}.
\]
Similarly, for $(i,j)=(1,2)$, the Fourier formula gives
\[
e_{1,2}
=
\left(1+\omega^{-1}x+\omega^{-2}x^2\right)
\left(1+\omega^{-2}y+\omega^{-4}y^2\right).
\]
Using $\omega^{-1}=\omega^2$, $\omega^{-2}=\omega$, and $\omega^{-4}=\omega^2$, this becomes
\[
e_{1,2}
=
(1+\omega^2x+\omega x^2)(1+\omega y+\omega^2y^2).
\]
Evaluating at $(\omega^{i^\prime},\omega^{j^\prime})$, we obtain
\[
e_{1,2}(\omega^{i^\prime},\omega^{j^\prime})
=
\left(\sum_{a=0}^{2}\omega^{(i^\prime-1)a}\right)
\left(\sum_{b=0}^{2}\omega^{(j^\prime-2)b}\right).
\]
Each geometric sum is zero unless the corresponding exponent is $0$ modulo $3$. 
Hence, we have 
\[
e_{1,2}(\omega^{i^\prime},\omega^{j^\prime})
=
\delta_{1,i^\prime}\delta_{2,j^\prime}.
\]
This agrees with the Lagrange interpolation formula from Lemma~\ref{lem:lagrange-eij}. For instance,
\[
L_1(x)
=
\prod_{\substack{u\in\Z_3\\u\neq 1}}
\frac{x-\omega^u}{\omega-\omega^u}
=
1+\omega^2x+\omega x^2,
\]
and
\[
M_2(y)
=
\prod_{\substack{v\in\Z_3\\v\neq 2}}
\frac{y-\omega^v}{\omega^2-\omega^v}
=
1+\omega y+\omega^2y^2.
\]
Thus
\[
L_1(x)M_2(y)
=
(1+\omega^2x+\omega x^2)(1+\omega y+\omega^2y^2)
=
e_{1,2}.
\]
\end{example}

The formulas \eqref{eq:lagrange-eij} and \eqref{eq:fourier-eij} exhibit coefficients such as $(\alpha^i - \alpha^u)^{-1}$ and $\alpha^{-ia} \beta^{-jb}$, which generally lie in $K$ but not in $\F_q$. However, some primitive idempotents are contained in $R$.
\begin{lemma}\label{lem:fixed-points-RK}
    Let $\sigma$ denote the $q$-Frobenius on $K$ extended coefficientwise to $R_K$. Then the $\sigma$-invariant of $R_K$ is precisely the $\F_q$-algebra $R$,
    \begin{equation*}
        R = (R_K)^{\langle\sigma\rangle} = \{f \in R_K \mid  \sigma(f) = f\}.
    \end{equation*}
\end{lemma}
\begin{proof}
    The extension $K / \F_q$ is Galois with group generated by $\sigma$.  Since $R_K = K \otimes_{\F_q} R$ and $\sigma$ acts trivially on $R$ and as Frobenius on $K$, the fixed subring is exactly $\F_q \otimes R \cong R$. Equivalently, $f \in R_K$ is fixed by $\sigma$ if and only if all its coefficients are fixed by Frobenius, so they all lie in $\F_q$.
\end{proof}

\begin{corollary}\label{cor:eij-in-R}
    For $(i, j) \in \Z_\ell \times \Z_m$, we have 
    \begin{equation*}
        e_{i, j} \in R \iff \ell \mid i(q - 1) \quad \text{and} \quad m \mid j(q - 1).
    \end{equation*}
    (Equivalently $e_{i, j}\in R$ if and only if $O_{(i, j)}$ has size $1$.)
\end{corollary}
\begin{proof}
    If $e_{i, j} \in R$, then $e_{i, j}$ is fixed by $\sigma$ by Lemma~\ref{lem:fixed-points-RK}. But $\sigma(e_{i, j}) = e_{qi, qj}$, so $e_{i, j} = e_{qi, qj}$, which forces $(qi, qj) \equiv (i, j)$ since the primitive idempotents in $K^{\ell m}$ are distinct coordinate vectors.
    
    Conversely, if $(qi, qj) \equiv (i, j)$ then $\sigma(e_{i, j}) = e_{i, j}$, hence $e_{i, j}\in R$ again by Lemma~\ref{lem:fixed-points-RK}.
\end{proof}

The situation for orbit idempotents is simpler.
\begin{proposition}\label{prop:orbit-idempotents-in-R}
    For every Frobenius orbit $O \in \Omega$, the orbit idempotent
    \begin{equation*}
        e_O = \sum_{(i, j) \in O} e_{i, j} \in R_K
    \end{equation*}
    lies in $R$ and is fixed by Frobenius: $\sigma(e_O) = e_O$. Moreover, the family $\{e_O\}$ consists of pairwise orthogonal idempotents in $R$ with
    \begin{equation*}
        1 = \sum_{O \in \Omega} e_O.
    \end{equation*}
\end{proposition}
\begin{proof}
    Since $\sigma$ permutes the summands $e_{i, j}$ along the orbit, the sum is invariant: $\sigma(e_O) = e_O$.  By Lemma~\ref{lem:fixed-points-RK} this implies $e_O \in R$. Orthogonality and the sum-to-one property follow from the corresponding statements for the $\{e_{i, j}\}$ together with the fact that the orbits form a partition of $\Z_\ell \times \Z_m$.
\end{proof}

\begin{example} 
\label{ex:orbit-idempotents-lm-3-B}
We continue with Example~\ref{ex:orbit-idempotents-lm-3-A}. 

The Frobenius automorphism $\sigma:z\mapsto z^2$ acts on the indices by
\[
(i,j)\longmapsto (2i,2j)
\]
in $\Z_3\times \Z_3$. Hence the Frobenius orbits are
\[
O_{0,0}=\{(0,0)\},
\]
\[
O_x=\{(1,0),(2,0)\},
\qquad
O_y=\{(0,1),(0,2)\},
\]
\[
O_{+}=\{(1,1),(2,2)\},
\qquad
O_{-}=\{(1,2),(2,1)\}.
\]
The corresponding orbit idempotents are
\[
e_{O_{0,0}}=e_{0,0}
=
(1+x+x^2)(1+y+y^2),
\]
\[
e_{O_x}=e_{1,0}+e_{2,0}
=
(x+x^2)(1+y+y^2),
\]
\[
e_{O_y}=e_{0,1}+e_{0,2}
=
(1+x+x^2)(y+y^2),
\]
\[
e_{O_{+}}=e_{1,1}+e_{2,2}
=
x+x^2+y+y^2+xy+x^2y^2,
\]
and
\[
e_{O_{-}}=e_{1,2}+e_{2,1}
=
x+x^2+y+y^2+xy^2+x^2y.
\]
Each of these polynomials has coefficients in $\F_2$, and therefore lies in
\[
R=\F_2[x,y]/\langle x^3-1,y^3-1\rangle.
\]
For example, using Lemma~\ref{lem:fourier-eij}, we have
\[
e_{1,1}+e_{2,2}
=
\sum_{a,b\in\Z_3}
\left(\omega^{-a-b}+\omega^{-2a-2b}\right)x^ay^b.
\]
If $a+b\equiv 0\pmod 3$, then the coefficient is $1+1=0$. If
$a+b\not\equiv 0\pmod 3$, then the coefficient is $\omega+\omega^2=1$.
Thus
\[
e_{1,1}+e_{2,2}
=
\sum_{\substack{a,b\in\Z_3\\a+b\not\equiv 0}}
x^ay^b
=
x+x^2+y+y^2+xy+x^2y^2.
\]
The same calculation gives the other orbit idempotents above.

Finally, these orbit idempotents are pairwise orthogonal and sum to $1$:
\[
1
=
e_{O_{0,0}}+e_{O_x}+e_{O_y}+e_{O_+}+e_{O_-}.
\]
This illustrates Proposition~\ref{prop:orbit-idempotents-in-R}: although the individual
primitive idempotents $e_{i,j}$ usually live only over the splitting field $K$ (see Example~\ref{ex:orbit-idempotents-lm-3-A}), their
Frobenius-orbit sums descend to the base ring $R$.
\end{example}

Proposition~\ref{prop:orbit-idempotents-in-R} induces a product decomposition of $R$ by the orbit idempotents.
\begin{proposition}[Wedderburn]
    \label{prop:product-decomp}
    The map
    \begin{equation*}
        \Psi : R \longrightarrow \prod_{O \in \Omega} e_OR, \qquad r \longmapsto (e_Or)_{O}
    \end{equation*}
    is an isomorphism of rings. In particular,
    \begin{equation*}
        R \cong \prod_{O \in \Omega} e_O R.
    \end{equation*}
\end{proposition}
\begin{proof}
    Since $\sum_O e_O = 1$ and $e_O e_{O^\prime} = 0$ for $O \neq O^\prime$, we have for every $r \in R$
    \begin{equation*}
        r = r \cdot 1 = r \sum_O e_O = \sum_O e_Or,
    \end{equation*}
    so $\Psi$ is injective: if $\Psi(r) = 0$, then $e_Or = 0$ for all $O$, hence $r = \sum_O e_Or = 0$. It is surjective: given $(x_O)_O$ with $x_O \in e_O R$, write $x_O = e_O r_O$ and set $r = \sum_O x_O \in R$. Then $e_O r = e_O x_O = x_O$ (and $e_O x_{O^\prime} = 0$ for $O \neq O^\prime$), so $\Psi(r) = (x_O)_O$. Finally, $\Psi$ is a ring homomorphism since the $e_O$ are central idempotents.
\end{proof}

\begin{proposition}
\label{prop:field-factor-degree}
    For each orbit $O \in \Omega$, the ring $e_O R$ is a finite field and
    \begin{equation*}
        |e_O R| = q^{|O|} \qquad \text{(equivalently, } \dim_{\F_q}(e_O R) = |O|\text{)}.
    \end{equation*}
    Consequently,
    \begin{equation*}
        R \cong \prod_{O \in \Omega} \F_{q^{|O|}} \qquad \text{as $\F_q$-algebras.}
    \end{equation*}
\end{proposition}
\begin{proof}
    Fix an orbit $O = \{(i_0, j_0), (i_1, j_1), \dots, (i_{r - 1}, j_{r - 1})\}$ of size $r = |O|$, ordered so that $(i_{t + 1}, j_{t + 1}) \equiv (q i_t, q j_t)$ and indices are taken modulo $r$. Via $\ev$, the ideal $e_O R_K$ identifies with the coordinate subproduct
    \begin{equation*}
        e_O R_K \cong K^{O} \cong K^r,
    \end{equation*}
    consisting of tuples supported on the coordinates in $O$. Under this identification, the Frobenius $\sigma$ acts by applying $a \mapsto a^q$ entrywise and cyclically permuting the coordinates along the orbit:
    \begin{equation*}
        \sigma(a_0, a_1, \dots, a_{r - 1}) = (a_{r - 1}^q, a_0^q, \dots, a_{r - 2}^q).
    \end{equation*}
    It is easy to see that $e_O R$ is precisely the $\sigma$-fixed subring of $e_O R_K$ (since $R = (R_K)^\sigma$ and $e_O \in R$). We claim that the fixed points are exactly the tuples of the form
    \begin{equation*}
        (a, a^q, a^{q^2}, \dots, a^{q^{r - 1}}) \quad \text{with} \quad a^{q^r} = a.
    \end{equation*}
    Indeed, a tuple $(a_0, \dots, a_{r - 1})$ is fixed by $\sigma$ if and only if
    \begin{equation*}
        (a_0, \dots, a_{r - 1}) = (a_{r - 1}^q, a_0^q, \dots, a_{r - 2}^q),
    \end{equation*}
    equivalently if and only if $a_1 = a_0^q$, $a_2 = a_1^q = a_0^{q^2}$, \dots, $a_{r - 1} = a_0^{q^{r - 1}}$ and also $a_0 = a_{r - 1}^q = a_0^{q^r}$. Hence, the claim follows.
    
    The condition $a^{q^r} = a$ means $a$ lies in the finite field $\F_{q^r}\subseteq K$. Therefore the map
    \begin{equation*}
        \F_{q^r} \longrightarrow (e_O R_K)^\sigma, \qquad a \longmapsto (a, a^q, \dots, a^{q^{r - 1}})
    \end{equation*}
    is a ring isomorphism (it is injective and multiplicative/additive componentwise, and every fixed tuple arises uniquely from its first coordinate). Hence $e_O R \cong \F_{q^r}$ is a field of size
    $q^r = q^{|O|}$, proving the claim.
\end{proof}

\begin{example}
    Applying Proposition~\ref{prop:field-factor-degree} to Example~\ref{example:5x7grid} gives
        $R \cong (\F_{4^6})^4 \times (\F_{4^2} )^2 \times (\F_{4^3})^2 \times \F_{4}$.
\end{example}

Each orbit $O$ contributes one simple component $e_O R$ in the product decomposition (Proposition~\ref{prop:product-decomp}), and Proposition~\ref{prop:field-factor-degree} shows that this component is a field extension of $\F_q$ of degree $|O|$:
\begin{equation*}
    \dim_{\F_q}(e_O R) = |O|.
\end{equation*}
Thus the multiset of orbit sizes $\{|O|\}$ is exactly the multiset of extension degrees in the Wedderburn decomposition of $R$.

The previous result shows that $R$ decomposes as an $\F_q$-algebra into direct sums of fields indexed with the orbit idempotents (Proposition~\ref{prop:field-factor-degree}). The next result shows that this is also an $R$-submodule decomposition.
\begin{lemma}\label{lem:orbit-decomp}
    As an $R$-module, $R$ decomposes as
    \begin{equation*}
        R = \bigoplus_{O \in \Omega} e_O R.
    \end{equation*}
    Moreover, each summand $e_O R$ has $\F_q$-dimension
    \begin{equation*}
        \dim_{\F_q}(e_O R) = |O|.
    \end{equation*}
\end{lemma}
\begin{proof}
This is essentially a restatement of Proposition~\ref{prop:field-factor-degree}. For completeness, we summarize the proof. Over $K$, the evaluation isomorphism identifies $R_K$ with $K^{\ell m}$ and the primitive idempotents with coordinate projections. Summing primitive idempotents over an orbit $O$ produces an idempotent $e_O$ fixed by Frobenius, hence lying in $R$. Orthogonality and the identity decomposition follow from the coordinate decomposition of $K^{\ell m}$.

The $R$-submodule $e_O R$ corresponds (after scalar extension to $K$) to the $K$-subspace of $K^{\ell m}$ supported on the coordinates in $O$, whose $K$-dimension is $|O|$. Frobenius descent implies the same dimension over $\F_q$, giving $\dim_{\F_q}(e_O R) = |O|$. Finally, the sum over all $O$ is direct since the idempotents are orthogonal.
\end{proof}

For any subset $\mathcal{T} \subseteq \Omega$, define the idempotent and corresponding ideal
\begin{equation}
    \label{eq:e_T&I_T}
    e_{\mathcal{T}} = \sum_{O \in \mathcal{T}} e_O \in R, \qquad I_{\mathcal{T}} = e_{\mathcal{T}} R \subseteq R.
\end{equation}
Equivalently, $I_{\mathcal{T}} = \bigoplus_{O \in \mathcal{T}} e_O R$.
\begin{corollary}\label{cor:orbit-additive-dim}
    For every $\mathcal{T} \subseteq \Omega$ (hence for every ideal in $R$), we have
    \begin{equation*}
        \dim_{\F_q}(I_{\mathcal{T}}) = \sum_{O \in \mathcal{T}} |O|.
    \end{equation*}
\end{corollary}
\begin{proof}
    By Lemma~\ref{lem:orbit-decomp}, $I_{\mathcal{T}} = \bigoplus_{O \in \mathcal{T}} e_O R$ as an $\F_q$-direct sum, so dimensions add:
    \begin{equation*}
        \dim_{\F_q}(I_{\mathcal{T}}) = \sum_{O \in \mathcal{T}} \dim_{\F_q}(e_O R) = \sum_{O \in \mathcal{T}} |O|.
    \qedhere
    \end{equation*}
\end{proof}

\subsection{Lattice Theory}
\begin{definition}
    For $\mathcal{S} \subseteq \Omega$, define the \emph{$\mathcal{S}$-selected idempotent} by $e_{\mathcal{S}} = \sum_{O \in \mathcal{S}} e_O \in R$.    The corresponding \emph{$\mathcal{S}$-selected ideal} is $J_{\mathcal{S}} = e_{\mathcal{S}} R \lhd R$.
\end{definition}

\begin{theorem}
\label{thm:ideals-by-orbits}
    Let $X$ be a set and $\mathcal{P}(X)$ its power set. Then the map
    \begin{equation*}
    \mathcal{P}(\Omega) \to {\text{ ideals of } R},
    \qquad
    \mathcal{S} \mapsto J_{\mathcal{S}},
    \end{equation*}
    is a bijection. More precisely, we have 
    \begin{enumerate}
        \item Every ideal $J \lhd R$ is of the form $J = J_{\mathcal{S}}$ for a unique subset $\mathcal{S} \subseteq \Omega$.
        \item Under this correspondence,
        \begin{itemize}
            \item $J_{\mathcal{S} \cap \mathcal{T}} = J_{\mathcal{S}} \cap J_{\mathcal{T}}$,
            \item $J_{\mathcal{S} \cup \mathcal{T}} = J_{\mathcal{S}} + J_{\mathcal{T}}$,
            \item $\ann_R(J_{\mathcal{S}}) = J_{\mathcal{S}^c}$.
        \end{itemize}
    \end{enumerate}
    In particular, the number of ideals of $R$ is $2^{|\Omega|}$.
\end{theorem}

\begin{proof}
Recall that $\ev$ is a ring isomorphism between $R_K$ and $K^{\ell m}$. 
The ideals of the latter ring are exactly
products of coordinate ideals. 
We claim that 
\[
I_T=\{f\in R_K:\ f(\alpha^i,\beta^j)=0\ \text{ for every } (i,j)\in T\}
=\ann_{R_K}\!\left(\sum_{(i,j)\in T}e_{i,j}\right).
\]
Indeed, under the evaluation isomorphism $\ev:R_K\to K^{\ell m }$, the primitive idempotent $e_{i,j}$
corresponds to the coordinate vector $\delta_{(i,j)}\in K^{\ell m }$ which is $1$ in the $(i,j)$
component and $0$ elsewhere. 
Hence the evaluation of the sum of $e_{i,j}$'s over $T$ gives 
\begin{equation*}
\ev\left( \sum_{(i,j)\in T} e_{i,j} \right)=\mathbf{1}_T,
\end{equation*}
the indicator vector of the subset $T\subseteq \Z_{ell}\times\Z_m$.
Since multiplication in $K^{\ell m }$ is component-wise, for any $f\in R_K$ we have
\begin{equation*}
\ev(f)\cdot \mathbf{1}_T=0
\quad\Longleftrightarrow\quad
\ev(f)_{(i,j)}=0\ \text{for all }(i,j)\in T
\quad\Longleftrightarrow\quad
f(\alpha^i,\beta^j)=0\ \text{for all }(i,j)\in T.
\end{equation*}
Transporting this statement back along $\ev$ gives
\begin{equation*}
\ann_{R_K}\!\left(\sum_{(i,j)\in T}e_{i,j}\right)
=\{f\in R_K:\ f(\alpha^i,\beta^j)=0\ \text{for all }(i,j)\in T\}
= I_T,
\end{equation*}
as claimed.
Hence, we established a natural bijection between the subsets of $\Z_\ell \times \Z_m$ and the ideals of $R_K$ (via $\ev$).

Now let $J\lhd R$ be any ideal. By extending scalars, we obtain an ideal $J_K=K\otimes_{\F_q}J\lhd R_K$.
The zero set of $J$, that is, 
\begin{equation*}
Z(J)=\{(i,j)\in \Z_\ell \times \Z_m:\ f(\alpha^i,\beta^j)=0\ \text{ for every }f\in J\}
\end{equation*}
is Frobenius-closed since $J$ is defined over $\F_q$.
Indeed, if $(i,j)\in Z(J)$ then applying $\sigma$ to $f(\alpha^i,\beta^j)=0$ gives
$f(\alpha^{qi},\beta^{qj})=0$, so $(qi,qj)\in Z(J)$, and hence $Z(J)$ is a union of orbits.

Recall that $\Omega$ denotes the set of all $q$-Frobenius orbits in $\Z_\ell \times \Z_m$. 
Let $\mathcal{S}$ be the set of orbits \emph{not} contained in $Z(J)$:
\begin{equation*}
\mathcal{S}=\{O\in \Omega:\ O\not\subseteq Z(J)\}.
\end{equation*}
Equivalently, $Z(J)$ is the union of the complementary orbits $\mathcal{S}^c$.
Then the indicator idempotent of $Z(J)$ is
\begin{equation*}
c_{Z(J)}=\sum_{(i,j)\in Z(J)}e_{i,j}=\sum_{O\in\mathcal{S}^c} e_O\in R,
\end{equation*}
and (over $R_K$) the ideal of functions vanishing on $Z(J)$ is exactly
\begin{equation*}
I\bigl(Z(J)\bigr)=\ann_{R}\!\bigl(c_{Z(J)}\bigr)=\ann_R\!\Bigl(\sum_{O\in\mathcal{S}^c}e_O\Bigr)
=\Bigl(\sum_{O\in\mathcal{S}}e_O\Bigr)R
=e_{\mathcal{S}}R.
\end{equation*}
But by definition $J\subseteq I(Z(J))$, and since $I(Z(J))$ is the \emph{largest} ideal with zero set
containing $Z(J)$, we must have $J=I(Z(J))$.  Hence $J=e_{\mathcal{S}}R$.
Hence, we associated a set of orbits $\mathcal{S}$ to the ideal $J$.
We claim that the uniqueness of $\mathcal{S}$ follows from orthogonality of the idempotents $\{e_O\}$. 
Indeed, 
if $e_{\mathcal{S}}R=e_{\mathcal{T}}R$, then multiplying by $e_O$ isolates each orbit component and
forces $\mathcal{S}=\mathcal{T}$.
This finishes the proof of our claim, so (1) follows.

\medskip 
We now prove the three identities in (2).
First note that since the idempotents
$\{e_O\}_{O\in\Omega}$ are pairwise orthogonal and sum to $1$, we have for all
$\mathcal{S},\mathcal{T}$:
\begin{equation}
\label{eq:eSeT}
e_{\mathcal{S}}e_{\mathcal{T}}
=
\Bigl(\sum_{O\in\mathcal{S}}e_O\Bigr)\Bigl(\sum_{O^\prime\in\mathcal{T}}e_{O^\prime}\Bigr)
=
\sum_{O\in\mathcal{S}\cap\mathcal{T}} e_O
=
e_{\mathcal{S}\cap\mathcal{T}}.
\end{equation}
Now let $x\in J_{\mathcal{S}}\cap J_{\mathcal{T}}$. Then $x=e_{\mathcal{S}}r=e_{\mathcal{T}}s$ for some
$r,s\in R$. Multiplying the equality $x=e_{\mathcal{S}}r$ by $e_{\mathcal{T}}$ and using
\eqref{eq:eSeT} yields
\begin{equation*}
x=e_{\mathcal{T}}x=e_{\mathcal{T}}e_{\mathcal{S}}r=e_{\mathcal{S}\cap\mathcal{T}}r\in J_{\mathcal{S}\cap\mathcal{T}}.
\end{equation*}
Hence $J_{\mathcal{S}}\cap J_{\mathcal{T}}\subseteq J_{\mathcal{S}\cap\mathcal{T}}$.
Conversely, if $x=e_{\mathcal{S}\cap\mathcal{T}}r$, then
$x=e_{\mathcal{S}}(e_{\mathcal{T}}r)\in J_{\mathcal{S}}$ and similarly $x\in J_{\mathcal{T}}$,
so $J_{\mathcal{S}\cap\mathcal{T}}\subseteq J_{\mathcal{S}}\cap J_{\mathcal{T}}$.
Therefore
\begin{equation*}
J_{\mathcal{S}\cap\mathcal{T}}=J_{\mathcal{S}}\cap J_{\mathcal{T}}.
\end{equation*}

We claim that
\begin{equation}
\label{eq:sum-generated}
J_{\mathcal{S}}+J_{\mathcal{T}}=(e_{\mathcal{S}}+e_{\mathcal{T}}-e_{\mathcal{S}}e_{\mathcal{T}})R.
\end{equation}
Indeed, since $e_{\mathcal{S}}e_{\mathcal{T}}=e_{\mathcal{S}\cap\mathcal{T}}$, we can rewrite
\begin{equation*}
e_{\mathcal{S}}+e_{\mathcal{T}}-e_{\mathcal{S}}e_{\mathcal{T}}
=
\sum_{O\in\mathcal{S}}e_O+\sum_{O\in\mathcal{T}}e_O-\sum_{O\in\mathcal{S}\cap\mathcal{T}}e_O
=
\sum_{O\in\mathcal{S}\cup\mathcal{T}}e_O
=
e_{\mathcal{S}\cup\mathcal{T}}.
\end{equation*}
To see \eqref{eq:sum-generated} directly, note that $e_{\mathcal{S}}$ and $e_{\mathcal{T}}$ both lie
in the ideal $(e_{\mathcal{S}}+e_{\mathcal{T}}-e_{\mathcal{S}}e_{\mathcal{T}})R$, since
\begin{equation*}
e_{\mathcal{S}}=(e_{\mathcal{S}}+e_{\mathcal{T}}-e_{\mathcal{S}}e_{\mathcal{T}})\,e_{\mathcal{S}},
\qquad
e_{\mathcal{T}}=(e_{\mathcal{S}}+e_{\mathcal{T}}-e_{\mathcal{S}}e_{\mathcal{T}})\,e_{\mathcal{T}}.
\end{equation*}
Hence $J_{\mathcal{S}}+J_{\mathcal{T}}\subseteq e_{\mathcal{S}\cup\mathcal{T}}R$.
Conversely, if $x=e_{\mathcal{S}\cup\mathcal{T}}r$, then using
$e_{\mathcal{S}\cup\mathcal{T}}=e_{\mathcal{S}}+e_{\mathcal{T}}-e_{\mathcal{S}\cap\mathcal{T}}$ we get
\begin{equation*}
x=e_{\mathcal{S}}r+e_{\mathcal{T}}r-e_{\mathcal{S}\cap\mathcal{T}}r.
\end{equation*}
But $e_{\mathcal{S}\cap\mathcal{T}}r\in J_{\mathcal{S}}\cap J_{\mathcal{T}}\subseteq J_{\mathcal{S}}$.
Similarly, $e_{\mathcal{S}\cap\mathcal{T}}r\in J_{\mathcal{T}}$, so the right-hand side lies in $J_{\mathcal{S}}+J_{\mathcal{T}}$.
Thus $e_{\mathcal{S}\cup\mathcal{T}}R\subseteq J_{\mathcal{S}}+J_{\mathcal{T}}$, proving
\begin{equation*}
J_{\mathcal{S}\cup\mathcal{T}}=J_{\mathcal{S}}+J_{\mathcal{T}}.
\end{equation*}

It remains to show that $\ann_R(J_{\mathcal{S}})=J_{\mathcal{S}^c}$.
Let $r\in R$. Then $r\in \ann_R(J_{\mathcal{S}})$ if and only if $r(e_{\mathcal{S}}R)=0$, or, equivalently, if and only if 
$re_{\mathcal{S}}=0$. Now decompose $r$ using $1=\sum_O e_O$:
\begin{equation*}
r=1 \cdot r = \left(\sum_{O} e_O\right) r=\sum_{O} (e_O r).
\end{equation*}
If $e_{\mathcal{S}}r=0$, then for each $O\in\mathcal{S}$ we have
\begin{equation*}
e_O r=(e_{\mathcal{S}} r) e_O=0,
\end{equation*}
so $r=\sum_{O\notin\mathcal{S}} e_O r \in \sum_{O\in\mathcal{S}^c} e_OR = e_{\mathcal{S}^c}R=J_{\mathcal{S}^c}$.
Conversely, if $r\in e_{\mathcal{S}^c}R$, say $r=e_{\mathcal{S}^c}t$, then
\begin{equation*}
e_{\mathcal{S}} r = e_{\mathcal{S}^c} t \, e_{\mathcal{S}}=(e_{\mathcal{S}^c}e_{\mathcal{S}})t=0
\end{equation*}
since $e_{\mathcal{S}^c}e_{\mathcal{S}}=0$ by orthogonality. Hence $r\in\ann_R(J_{\mathcal{S}})$.
Therefore $\ann_R(J_{\mathcal{S}})=J_{\mathcal{S}^c}$.
This finishes the proof of our theorem.
\end{proof}

Consider the set function
\begin{equation*}
    f:\mathcal{P}(\Omega) \to \Z_{\ge 0}, \qquad f(\mathcal{T}) = \sum_{O \in \mathcal{T}} |O|.
\end{equation*}
Equivalently, by Corollary~\ref{cor:orbit-additive-dim}, $f(\mathcal{T}) = \dim_{\F_q} (I_{\mathcal{T}})$,
where $I_\mathcal{T}$ is the ideal corresponding to the zero set $\mathcal{T}$. We proceed to show the (poset/lattice theoretic) modularity property of $f$~\cite[Chapter 3]{StanleyV1}.

\begin{proposition}
\label{prop:modularity}
    For all $\mathcal{T}, \mathcal{U} \subseteq \Omega$ we have
    \begin{equation*}
        f(\mathcal{T}) + f(\mathcal{U}) = f(\mathcal{T} \cup \mathcal{U}) + f(\mathcal{T} \cap \mathcal{U}).
    \end{equation*}
    Equivalently, $\dim_{\F_q}(I_{\mathcal{T}}) + \dim_{\F_q}(I_{\mathcal{U}}) = \dim_{\F_q}(I_{\mathcal{T} \cup \mathcal{U}}) + \dim_{\F_q}(I_{\mathcal{T} \cap \mathcal{U}})$.
\end{proposition}
\begin{proof}
    Write the disjoint decomposition
    \begin{equation*}
        \mathcal{T} = (\mathcal{T} \setminus \mathcal{U}) \dot \cup (\mathcal{T} \cap \mathcal{U}), \qquad \mathcal{U} = (\mathcal{U} \setminus \mathcal{T}) \dot \cup (\mathcal{T} \cap \mathcal{U}),
    \end{equation*}
    and similarly
    \begin{equation*}
        \mathcal{T} \cup \mathcal{U} =(\mathcal{T} \setminus \mathcal{U}) \dot \cup (\mathcal{U} \setminus \mathcal{T}) \dot \cup (\mathcal{T} \cap \mathcal{U}).
    \end{equation*}
    Using additivity of sums over disjoint unions, we obtain 
    \begin{align*}
        f(\mathcal{T}) + f(\mathcal{U}) &= \sum_{O \in \mathcal{T} \setminus \mathcal{U}} |O| + \sum_{O \in \mathcal{T} \cap \mathcal{U}} |O| + \sum_{O \in \mathcal{U} \setminus \mathcal{T}} |O| + \sum_{O \in \mathcal{T} \cap \mathcal{U}} |O|\\
        &= \sum_{O \in \mathcal{T} \setminus \mathcal{U}} |O| + \sum_{O \in \mathcal{U} \setminus \mathcal{T}} |O| + \sum_{O \in \mathcal{T} \cap \mathcal{U}} |O| + \sum_{O \in \mathcal{T} \cap \mathcal{U}} |O|\\
        &= \sum_{O \in \mathcal{T} \cup \mathcal{U}} |O| + \sum_{O \in \mathcal{T} \cap \mathcal{U}} |O|\\
        &= f(\mathcal{T} \cup \mathcal{U}) + f(\mathcal{T} \cap \mathcal{U}),
    \end{align*}
    as claimed.
\end{proof}

There is a natural extension of multiplier automorphisms to the 2D cyclic code setting. Let $a \in [\ell]$ and $b \in [m]$ be integers such that $\gcd(a, \ell) = 1$ and $\gcd(b, m) = 1$. 
Then the mapping $\mu_{a, b} : [\ell] \times [m] \to [\ell] \times [m]$ defined by 
\begin{equation*}
    \mu_{a, b}(i, j) = (a \cdot i \!\!\!\! \mod \ell, \ b \cdot j \!\!\!\! \mod m)
\end{equation*}
is a bijection. Its action on the bivariate polynomials is given by 
\begin{equation*}
  \mu_{a, b}(x^i y^j) = x^{ai} y^{bj} \!\!\!\! \mod{\langle x^\ell - 1, y^m - 1 \rangle}.
\end{equation*}

\begin{remark}
In the semisimple case (where $\gcd(q, \ell m) = 1$), the $q$-Frobenius mapping defined by $\mu:(i, j) \mapsto (q \cdot i \!\!\! \mod \ell, \ q \cdot j \!\!\! \mod m)$ is a specific diagonal multiplier mapping.
\end{remark}

Let $\mathcal{C}$ be a 2D cyclic code in $R$. The mapping $\mu_{a, b}$ is a \emph{multiplier automorphism of $\mathcal{C}$} if $\mu_{a, b}(\mathcal{C}) = \mathcal{C}$. As in the 1D case, this can be stated in terms of defining sets. Let $\mathcal{T} = O_{(i_1, j_1)} \cup \cdots \cup O_{(i_r, j_r)}$ be the union of Frobenius orbits that make up the defining set for $\mathcal{C}$. The mapping $\mu_{a, b}$ is a multiplier automorphism of $\mathcal{C}$ if and only if $\mu_{a, b}(\mathcal{T}) = \mathcal{T}$.

\begin{lemma}
    Let $\mu_{a, b}$ be a multiplier mapping on $R$ and let $\Omega$ denote the set of $q$-Frobenius orbits in $\Z_\ell \times \Z_m$. Then $\mu_{a, b}$ induces a permutation of $\Omega$ and satisfies $\mu_{a, b}(e_O) = e_{\mu_{a, b}^{-1}(O)}$ for every $O \in \Omega$.
\end{lemma}
\begin{proof}
    Multiplication by $a$ modulo $\ell$ and by $b$ modulo $m$ commutes with multiplication by $q$. Thus $\mu_{a,b}(q^t i, q^t j) = (q^t a i, q^t b j)$ for every $t \ge 0$.
    This implies $\mu_{a, b}$ maps each $q$-Frobenius orbit to a $q$-Frobenius orbit, so it permutes $\Omega$.
    
    Now write
    \begin{equation*}
        e_O = \sum_{(i, j) \in O} e_{i, j}.
    \end{equation*}
    The multiplier $\mu_{a, b}$ sends the primitive spectral idempotent $e_{i, j}$ to $e_{\mu_{a, b}^{-1}(i, j)}$. We obtain
    \begin{equation*}
        \mu_{a, b}(e_O) = \sum_{(i, j) \in O} e_{\mu_{a, b}^{-1}(i, j)} = \sum_{(r, s) \in \mu_{a, b}^{-1}(O)} e_{r, s} = e_{\mu_{a, b}^{-1}(O)}.
    \end{equation*}
\end{proof}

\subsection{Distance Bounds}
In the broader algebraic context, ideals in $R$ are often referred to as abelian codes since they correspond to ideals in an abelian group ring. Unlike 1D cyclic codes, bounding the minimum distance of 2D abelian codes requires handling multiple indeterminates. The two primary methods for establishing distance lower bounds are Camion's generalized abelian BCH bound~\cite{Camion1971} and Jensen's concatenated approach~\cite{Jensen1985}. A good review of these works can be found in~\cite{Sabin1992}.
Sakata later provided alternative bounds and locator decoding algorithms using independent point sets and Gröbner bases for doubly periodic arrays~\cite{Sakata1978, Sakata1981}. 
Other relevant work in this area was done by Imai~\cite{Imai1977}.

The foundational lower bound for 2D cyclic codes relies on the concept of \emph{apparent distance}, originally introduced by Camion (1971) for abelian codes. 
Unlike the 1D BCH bound that looks for a single linear sequence of consecutive roots, this bound requires the code's defining set to contain specific sets of complete rows and columns, referred to as zero hypercolumns. In 2D, these collections of zero hypercolumns correspond to unions of consecutive index strips.

Building on Camion's work~\cite{Camion1971}, Bernal, Bueno-Carre{\~{n}}o, and Sim{\'{o}}n (2016) formalized in~\cite{Bernal2016} the notion of a \emph{strong apparent distance}, denoted $\mathrm{sd}(\mathcal{C})$, to establish a ``Multivariate BCH Bound'' (see~\cite[Theorem 30]{Bernal2016}). While we omit the rather involved algorithmic definition of $\mathrm{sd}(\mathcal{C})$ here for the sake of brevity, the interested reader is referred to the original paper for its full formulation. 

For practical code design, the 2D BCH bound is most easily expressed in terms of index sets. The following formulation of the bound is from~\cite[Corollary 26]{Bernal2016}. 

\begin{theorem}[2D BCH Bound]
\label{thm:2D-BCH-bound}
Let $q$ be a power of the prime $p$, let $\mathcal{C}$ be a nonzero 2D cyclic (abelian) code of length $\ell\times m$ over $\F_q$, and assume that $\gcd(p,\ell m)=1$. Write $r_1=\ell$, $r_2=m$, and $I=\Z_\ell\times\Z_m$. Let $\emptyset \neq\Gamma\subseteq\{1,2\}$. For each $k\in\Gamma$, choose a starting shift $b_k\geq 0$ and a designed distance $2\leq\delta_k\leq r_k$, and define the consecutive index set $J_k=\{\overline{b_k},\overline{b_k+1},\dots,\overline{b_k+\delta_k-2}\}\subseteq\Z_{r_k}$. Let $A_k=\{(i_1,i_2)\in I\mid i_k\in J_k\}$, so that $A_k$ is a union of $\delta_k-1$ consecutive full zero hypercolumns in the $k$th coordinate direction.

If the defining set of $\mathcal{C}$ contains $\bigcup_{k\in\Gamma}A_k$, then $d(\mathcal{C})\geq\prod_{k\in\Gamma}\delta_k$. In particular, when $\Gamma=\{k\}$, a run of $\delta_k-1$ consecutive full zero hypercolumns gives the one-sided bound $d(\mathcal{C})\geq\delta_k$.
\end{theorem}

The following example illustrates both the one-sided and two-sided forms of Theorem~\ref{thm:2D-BCH-bound}. We work with two ideals in the same ring, in the same manner as the classical annihilator and colon ideals that occur in our BB distance bounds.

\begin{example}
\label{ex:2D-BCH-one-two-sided}
Let $R=\F_2[x,y]/\langle x^3-1,\ y^3-1\rangle$, and fix primitive third roots of unity $\alpha,\beta\in\F_4$. We consider the ideals $\mathcal{C}_1=\langle g_1\rangle$ and $\mathcal{C}_2=\langle g_2\rangle$ generated by $g_1=1+x+x^2$ and $g_2=(1+x)(1+y)=1+x+y+xy$, respectively. Since each ideal is principal, its defining set is the zero set of its generator.

For the first ideal, $g_1(\alpha^i,\beta^j)=0$ exactly when $i\in\{1,2\}$, independently of $j$. Hence $\mathcal{Z}(\mathcal{C}_1)=\{1,2\}\times\Z_3$. Taking $\Gamma=\{1\}$, $b_1=1$, and $\delta_1=3$ gives $J_1=\{1,2\}$ and $A_1=\{1,2\}\times\Z_3$. The one-sided form of Theorem~\ref{thm:2D-BCH-bound} therefore gives $d(\mathcal{C}_1)\geq 3$. 

For the second ideal, $g_2(\alpha^i,\beta^j)=(1+\alpha^i)(1+\beta^j)$, which vanishes exactly when $i=0$ or $j=0$. Thus $\mathcal{Z}(\mathcal{C}_2)=(\{0\}\times\Z_3)\cup(\Z_3\times\{0\})$. Taking $\Gamma=\{1,2\}$, $b_1=b_2=0$, and $\delta_1=\delta_2=2$ gives $J_1=J_2=\{0\}$, $A_1=\{0\}\times\Z_3$, and $A_2=\Z_3\times\{0\}$. The two-sided form of Theorem~\ref{thm:2D-BCH-bound} therefore gives $d(\mathcal{C}_2)\geq 2\cdot2=4$. 
\end{example}

\begin{remark}
\label{rem:2D-BCH-rectangle-warning}
The two-sided hypothesis in Theorem~\ref{thm:2D-BCH-bound} concerns a union of full zero hypercolumns, not merely a rectangular block of zeros. For $\Gamma=\{1,2\}$, the theorem requires the defining set to contain $A_1\cup A_2=(J_1\times\Z_m)\cup(\Z_\ell\times J_2)$. This union necessarily contains the rectangular block $A_1\cap A_2=J_1\times J_2$, which has $(\delta_1-1)(\delta_2-1)$ points, but the converse is false: containing the rectangle $J_1\times J_2$ does not imply that the defining set contains either family of full strips.

The first ideal in Example~\ref{ex:2D-BCH-one-two-sided} gives a concrete counterexample. Its defining set $\mathcal{Z}(\mathcal{C}_1)=\{1,2\}\times\Z_3$ contains the consecutive rectangle $\{1,2\}\times\{0\}$ associated with $\delta_1=3$ and $\delta_2=2$. However, it does not contain the required union $(\{1,2\}\times\Z_3)\cup(\Z_3\times\{0\})$, since $(0,0)$ is not a zero of $g_1$. Treating the rectangle as sufficient would incorrectly predict the product bound $d(\mathcal{C}_1)\geq 3\cdot2=6$, whereas the true distance is $3$.

Thus, when applying the theorem, one must check that every point in each required hypercolumn belongs to the defining set. It is not enough to find a consecutive rectangular block. We also recall that consecutiveness is cyclic: the index sets $J_k$ are taken modulo $r_k$ and may wrap around from $r_k-1$ to $0$.
\end{remark}

For consistency with the source, we use the term hypercolumn in Theorem~\ref{thm:2D-BCH-bound}. In two dimensions, collections of zero hypercolumns correspond to unions of full strips (parallel to the coordinate axes). This description lends itself to visualization on a 2D grid, so we use the language of strips going forward. In the setting considered here, $\prod_{k\in\Gamma}\delta_k$ equals $\delta_k$ in the one-sided case $\Gamma=\{k\}$ and equals $\delta_1\delta_2$ in the two-sided case $\Gamma=\{1,2\}$. We call this value a \emph{BCH distance certificate}.

Much like the 1D case, this bound can be strengthened. The Hartmann-Tzeng and Roos bounds have natural multivariate generalizations for abelian codes, typically leveraging product sets of independent sequences across both the $x$ and $y$ dimensions to establish higher minimum distances. 

The shift bound was formally generalized to multivariate polynomial rings by Pellikaan~\cite{Pellikaan1996}, making it directly applicable to 2D cyclic (abelian) codes.
\begin{theorem}[Multivariate Shift Bound]
    \label{thm:2D-shift-bound}
    Let $0 \neq f(x, y) \in K[x, y]$, where $K$ is a field, and let
    \begin{equation*}
        S = \{ (\theta_1, \theta_2) \in K \times K \mid f(\theta_1, \theta_2) = 0 \}.
    \end{equation*}
    If $A \subset K \times K$ is an independent set with respect to $S$ in the sense of the generalized shift bound, then $\wt(f) \geq |A|$. Consequently, if every nonzero codeword $c(x, y) \in \mathcal{C}$ has a zero set containing a set $S$ for which one can construct an independent set $A$ of size $D$, then $d(\mathcal{C}) \geq D$.
\end{theorem}

Bounds for the duals of 2D cyclic codes also exist, but they are considerably more complex to evaluate than the classical Carlitz-Uchiyama theorem. While the 1D Carlitz-Uchiyama bound relies on Hasse-Weil estimates for curves, bounding the exponential sums associated with bivariate dual codes requires multi-dimensional character sum estimates.

Camion generalized the classical BCH bound to the abelian case~\cite{Camion1971}. An abelian codeword can be associated with a polynomial representing its discrete Fourier transform. An abelian code is then identified with a specific polynomial called its support. The generalized bound is determined by viewing this support polynomial successively as a polynomial in each single indeterminate:
\begin{itemize}
    \item By isolating one indeterminate, the coefficients of the support polynomial can be viewed as polynomials in the remaining variables.
    \item The standard 1D cyclic BCH bound is then applied to these coefficients to estimate their minimum number of nonzeros.
    \item Maximizing these nonzero estimates across the different variables establishes a lower bound on the minimum distance of the entire abelian code, known as the abelian BCH bound.
\end{itemize}

An alternative approach to establishing a lower bound was presented by Jensen~\cite{Jensen1985}. As Sabin notes in~\cite{Sabin1992}, this method avoids the complex mathematical analysis required for a direct generalization of a cyclic bound to the multivariate case. Instead, it views the 2D abelian code as a direct sum of concatenations of cyclic outer and inner codes.
\begin{itemize}
    \item An abelian code in the group ring corresponding to our grid $\mathbb{Z}_\ell \times \mathbb{Z}_m$ can be decomposed into minimal cyclic inner codes concatenated with cyclic outer codes.
    \item Once the 2D code is decomposed in this manner, the bound on generalized concatenated codes can be applied.
    \item A minimum distance bound is given by applying the standard 1D cyclic BCH bound to both the inner and outer cyclic codes of each concatenation.
\end{itemize}

In recent decades the classical bounds for multidimensional cyclic and abelian codes have seen significant modern refinements. Jensen's concatenated approach has been heavily optimized. By utilizing trace representations and the Chinese Remainder Theorem, G\"uneri and \"Ozbudak provided an explicit improvement over Jensen's bound for quasi-cyclic codes~\cite{GuneriOzbudak2012}, which they showed are linked to 2D cyclic codes~\cite{GuneriOzbudak2012}. This concatenated framework was also recently extended to quasi-abelian codes, allowing general minimum distance bounds to be transferred to the abelian setting~\cite{Guneri2021}. Furthermore, they estimated the weights of codewords in multidimensional cyclic codes by expressing them via multi-dimensional character sums and bounding them using the Hasse-Weil-Igusa bounds on Artin--Schreier type hypersurfaces~\cite{GuneriOzbudak2008_Artin}. This approach provides a powerful algebraic alternative to the standard 2D BCH and concatenated bounds.

\section{Bivariate Bicycle Codes}
\label{sec:BBpast}
In this section, we assume that $q$ is an arbitrary prime power unless otherwise noted. Let $\ell$ and $m$ be positive integers and consider the ring $R = \F_q[x, y] / \idlGens{x^\ell - 1, y^m - 1}$.

\begin{definition} 
\label{def:BBcodes}
    A \emph{bivariate bicycle code}, denoted $\mathrm{BB}(a, b)$, is the CSS code defined by two polynomials $a(x, y)$ and $b(x, y)$ in $R$. Let $A$ and $B$ be the linear transformations on $R^{\ell m}$ corresponding to multiplication by $a(x, y)$ and $b(x, y)$, respectively. The code has parity-check matrices
    \begin{equation}
        H_X = \begin{pmatrix} A & B \end{pmatrix} \quad , \quad H_Z = \begin{pmatrix} B^T & -A^T \end{pmatrix}.
    \end{equation}
\end{definition}

\begin{proposition}[\cite{LinPryadko2024} for $q=2$]
    \label{prop:equivalentcodes}
    Let $P\in \F_q^{\ell m \times \ell m}$ denote the permutation matrix implementing inversion on $\mathbb{Z}_\ell \times \mathbb{Z}_m$. Consider the CSS code defined by the (BB-type) check matrices
    \begin{equation*}
        H_X = \begin{pmatrix} A & B\end{pmatrix}, \qquad H_Z = \begin{pmatrix}B^{\top} & -A^{\top}\end{pmatrix}.
    \end{equation*}
    Then, using qudit permutations, stabilizer relabeling, global Hadamard,
    and monomial column scaling,
    the following are equivalent:
    \begin{enumerate}
        \item $H_X = \begin{pmatrix}A & -B\end{pmatrix}$, \quad $H_Z = \begin{pmatrix}B^\top & A^\top\end{pmatrix}$;
        \item $H_X = \begin{pmatrix}A^\top & -B^\top\end{pmatrix}$, \quad $H_Z = \begin{pmatrix}B & A\end{pmatrix}$;
        \item $H_X = \begin{pmatrix}-B^\top & A^\top\end{pmatrix}$, \quad $H_Z = \begin{pmatrix}A & B\end{pmatrix}$;
        \item $H_X = \begin{pmatrix}A & B\end{pmatrix}$, \quad $H_Z = \begin{pmatrix}-B^\top & A^\top\end{pmatrix}$.
    \end{enumerate}
    In particular, up to code equivalence, all transposes in the BB presentation
    may be removed and replaced by conjugation with $P$.
\end{proposition}
\begin{proof}
    We use the following standard equivalence operations for CSS codes:
    \begin{itemize}
        \item[\textbf{(E1)}] Simultaneous permutation of columns of $H_X$ and $H_Z$.
        \item[\textbf{(E2)}] Permutation of rows of $H_X$ and $H_Z$.
        \item[\textbf{(E3)}] Exchange of $H_X$ and $H_Z$.
        \item[\textbf{(E4)}] Multiplication of columns by $c \in \F_q^*$ where $c^2 = 1$. (If $\operatorname{char}\F_q \neq 2$.) 
    \end{itemize}
    By Proposition~\ref{prop:algebraic-properties}(2), transposition corresponds to inversion in the group algebra, which at the matrix level is implemented by conjugation with the
    permutation matrix $P$ which is given by $P=M_\iota$ (in the notation of Proposition~\ref{prop:algebraic-properties}). Hence we see that 
    \begin{equation*}
        A^{\top} = M_\iota A M_\iota, \qquad B^{\top} = M_\iota B M_\iota.
    \end{equation*}
    Applying the column scaling $\operatorname{diag}(I_N, -I_N)$ (which is a particular case of operation \textbf{(E4)}) to the second block removes the minus sign in $H_Z$, but introduces it in $H_X$. This gives the equivalent
    presentation
    \begin{equation*}
        H_X = \begin{pmatrix}A & -B\end{pmatrix}, \qquad H_Z = \begin{pmatrix}M_\iota B M_\iota & M_\iota A M_\iota\end{pmatrix},
    \end{equation*}
    which is exactly presentation (1).
    
    Next, we analyze the effect of conjugation by $M_\iota$. Applying the column permutation $\operatorname{diag}(M_\iota, M_\iota)$ followed by the row permutation
    $M_\iota$ (operations \textbf{(E1)} and \textbf{(E2)}) to presentation (1) yields 
    \begin{equation*}
        H_X = \begin{pmatrix}M_\iota A M_\iota & -M_\iota B M_\iota\end{pmatrix}, \qquad H_Z = \begin{pmatrix}M_\iota M_\iota B M_\iota M_\iota & M_\iota M_\iota A M_\iota M_\iota\end{pmatrix}.
    \end{equation*}
    Since $M_\iota^2 = I_N$, this simplifies to presentation (2).
    
    Next, we interchange the two blocks. This is done by permuting the columns. Hence, applying operation \textbf{(E1)} transforms presentation (2) into presentation (3).
    
    Finally, we interchange $X$ and $Z$ checks by applying \textbf{(E3)} to presentation (3). This simply swaps the two matrices, yielding presentation (4).
\end{proof}

\begin{remark}
    It is important to note that while the equivalent presentations in Proposition~\ref{prop:equivalentcodes} share the same code parameters $[\![N,K,D]\!]$, their physical implementations could vary widely. 
    The equivalence operations yield isomorphic abstract Tanner graphs but may alter the chosen geometric embedding or hardware layout.
\end{remark}

The BB code literature is dominated by large-scale numerical searches over code polynomials $a(x, y)$ and $b(x, y)$ leading to an abundance of interesting codes with good parameters
~\cite{wang_coprime_2025},~\cite{PostemaKokkelmans2025},~\cite{EberhardtSteffan2025},~\cite{symons2025sequences}.
Several properties of these codes can be determined from their polynomials without needing to explicitly construct the stabilizers.

Consider the coprime case where $\gcd(\ell, m) = 1$ and $\gcd(\ell m, q) = 1$ and define $R_z = \F_2[z]/\langle z^{\ell m} - 1 \rangle$. 
\begin{proposition}[\cite{wang_coprime_2025}]
\label{prop:wang_coprime}
Let $q=2$. 
    The code dimension of a coprime BB code is given by $k = 2 \deg g(z)$, where $g(z) = \gcd(a(z), b(z), z^{\ell m} - 1)$ is computed as an element of $\F_2[z]$.
\end{proposition}
\noindent Note that this formula is identical to that for generalized bicycle codes, as the coprime construction can be viewed as a special case.

By characterizing polynomials which divide the trinomials
\begin{equation}\label{BBansatz}
    x^{i_1} + y^{i_2} + y^{i_3} \quad \textrm{ or } \quad y^{i_1} + x^{i_2} + x^{i_3},
\end{equation}
\cite{PostemaKokkelmans2025} classified the possible lengths of coprime BB codes.
\begin{theorem}[\cite{PostemaKokkelmans2025}]
Let $q=2$. 
    Any coprime BB code of length $2 \ell m$ with polynomials of the form~\eqref{BBansatz} and $g(z) = \gcd(a(z), b(z)$, $z^{\ell m} - 1)$ is non-trivial ($k \geq 2$) if and only if
    \begin{itemize}
        \item $g(z) \neq 1$
        \item $\ell m$ is divisible by a Mersenne prime or an outlier prime (73, 121369, 178481, 262657, 599479).
    \end{itemize}
\end{theorem}

Now consider the case $\ell$ is odd and $m$ is even: $m = 2^r m^\prime$, $r \geq 1$, $2 \nmid m^\prime$. Then, in characteristic two, $(y^m - 1) = (y^{m^\prime} - 1)^{2^r} = \prod [g_i(y)]^{2^r}$, where $g_i$ is an irreducible factor of $y^{m^\prime} - 1$. Let $x^\ell - 1 = \prod_{j = 1}^\eta f_j(x)$ be the factorization into monic irreducible polynomials in $\F_2[x]$ and define $S_j = \F_2[x, y]/\idlGens{f_j(x), y^m - 1}$. Then
\begin{equation}
    \psi = (\psi_1, \dots, \psi_\eta) \, : \, \F_2[x, y]/\idlGens{x^\ell - 1, y^m - 1} \to \bigoplus^\eta_{j = 1} S_j,
\end{equation}
where $\psi_j(p(x, y)) = p(x, y) \, \mathrm{mod} \, f_j(x)$, is a valid ring isomorphism.

\begin{lemma}
    Codes defined over the ring $\F_2[x, y]/\idlGens{x^\ell - 1, y^m - 1}$ with parity-check matrix $H = \begin{pmatrix} A & B\end{pmatrix}$ are generated by the ideal
    \begin{equation}
        H = \langle g_1(x, y) \prod_{j \neq 1} f_j(x), \dots, g_\eta(x, y) \prod_{j \neq \eta} f_j(x) \rangle,
    \end{equation}
    where $g_j(x, y) = \gcd(\psi_j(a(x, y)), \psi_j(b(x, y)), y^m - 1)$ in $(\F_2[x]/\idlGens{f_j(x)}[y]$.
\end{lemma}
\begin{corollary}[\cite{PostemaKokkelmans2025}]
Let $q=2$.
    The dimension of the resulting BB code is given by
    \begin{equation}
        k = 2 \sum_{j = 1}^\eta \deg_y g_j(x, y) \deg_x f_j(x).
    \end{equation}
\end{corollary}
\noindent A useful example computation is provided in \cite{PostemaKokkelmans2025}.

Next, consider the case $\ell$ and $m$ are any positive even integers. Analysis of codes of this kind of ring typically requires Gr\"{o}bner bases. Assume the lexicographical ordering throughout and denote the leading monomial of a polynomial $f$ by $\mathrm{LM}(f)$.
\begin{theorem}\label{thm:gendims}
    Let $G$ be a Gr\"{o}bner basis of the ideal $\langle a(x,y), b(x, y), x^\ell - 1, y^m - 1 \rangle$. If $\mathfrak{m}$ denotes the set of standard monomials not contained in the leading-term ideal of $G$, then the dimension of the corresponding code is given by $k = 2 |\mathfrak{m}| = 2 \dim_{\F_q} \operatorname{span}_{\F_q}(\mathfrak{m})$. 
\end{theorem}

Assuming $q=2$, 
Eberhardt and Steffan \cite{EberhardtSteffan2025} applied techniques from commutative algebra to show how logical operators and operations are determined by $a(x, y)$ and $b(x, y)$. The first homology group of the chain complex of $\mathrm{BB}(a,b)$ is
\begin{equation}\label{R-homology}
    H_1(\mathrm{BB}(a, b)) = \ker \begin{pmatrix} a & b \end{pmatrix} / \im \begin{pmatrix} b \\ a \end{pmatrix} = \{(f, g) \in R \times R \mid af = bg\} / \{(rb, ra) \mid r \in R\}.
\end{equation}
Denote the homology classes by $[(f, g)] \in H_1$.
\begin{definition}[Def.~1 \cite{EberhardtSteffan2025}]
    \label{def:purity}
    An element $h \in H_1$ is \emph{horizontally pure} if $h = [(f, 0)]$ and \emph{vertically pure} if $h = [(0, g)]$. 
    We denote the subspaces of $H$ generated by horizontally and vertically pure classes by $H_h$,$H_v$ and say that $H$ is pure if $H = H_h + H_v$.
\end{definition}
\begin{remark}
The purity of $H$ does not imply that every element of $H$ can be written uniquely as a sum $h=h_h + h_v.$ 
This is related to the fact that pure horizontal and vertical classes can be equal in the homology group $[f,0] = [0,g]$ even when the class they generate isn't the 0 class of $H$.
A unique decomposition of $h$ into horizontally and vertically pure classes corresponds to the special case $H=H_h \oplus H_v.$
\end{remark}

Recalling that $R \times R$ arises from having two disjoint copies of each monomial (edge label), purity reflects whether or not the logicals remain disjoint, being supported on either one copy or the other. It is easy to see that this would be the case if $a(x)$ and $b(y)$ are univariate polynomials.

\begin{proposition}[Cor.~2 \cite{EberhardtSteffan2025}]
    The homology is pure if and only if $\idlGens{a} \cap \idlGens{b} = \idlGens{ab}$.
\end{proposition}
\noindent Note that if $\idlGens{a} + \idlGens{b} = R$, then $H_1 = 0$.

\begin{definition}[Def.~2 \cite{EberhardtSteffan2025}]
    The code is \emph{principal} if it is pure and $\mathrm{ann}(a)$ and $\mathrm{ann}(b)$ are principal.
\end{definition}
\noindent If the code is principal, all logical operators can be generated from the two principal generators of the ideals.
\begin{proposition}[Prop.~2, Cor.~8 \cite{EberhardtSteffan2025}]
    \label{prop:princ_is_pure}
    The code is principal for $\ell$ and $m$ odd ($R$ is semisimple).
\end{proposition}

\begin{definition}[Def.~3 \cite{EberhardtSteffan2025}]
    A polynomial $f \in R$ is \emph{semi-periodic} if $f(x, y) = x^k + \zeta(y)$ and $k k^\prime = \ell$.
\end{definition}
\noindent The definition holds switching $x$ and $y$ and $\ell$ and $m$.
\begin{theorem}[Thm.~5 \cite{EberhardtSteffan2025}]
    If a BB code with~\eqref{BBansatz} is pure and $a, b$ are both either univariate or semi-periodic, then it is principal.
\end{theorem}

In the general case, finding a principal generator of a principal ideal is difficult and is referred to as the \emph{principal ideal problem}. In the semi-periodic case above, we can actually write it down explicitly. Let $a(x, y) = x^k + \zeta(y)$ with $k k^\prime = \ell$, and define $\chi(y) = \zeta(y)^{k^\prime} - 1$ and $\bar{\chi}(y) = \gcd(\chi(y), y^m - 1)$. Then $g(y)$, defined by $g(y) \bar{\chi}(y) = y^m - 1$, is the generator polynomial of a cyclic code over $\F_2[y]/(y^m - 1)$. (The parity-check polynomial for the code is $\bar{\chi}(y)$.) Let
\begin{equation}
    P(x, y) = \sum^{k^\prime - 1}_{i = 0} x^{\ell - i k} \zeta(y)^i g(y).
\end{equation}
\begin{theorem}[Thm.~6 \cite{EberhardtSteffan2025}]
    With the previous assumptions, $\mathrm{ann}(a) = \idlGens{P(x, y)}$.
\end{theorem}
\noindent The argument can be repeated to get $\mathrm{ann}(b) = \idlGens{Q(x, y)}$.

Much less has been said about the minimum distance of BB codes. Lin and Pryadko \cite{LinPryadko2024} use group algebras to provide an upper bound by looking at the logical operators supported on only half of the qubits. Let $q = 2$. We strengthen some of their conclusions in upcoming sections. Consider a BB code with stabilizers $H_X = \begin{pmatrix} A & B \end{pmatrix}$ and $H_Z = \begin{pmatrix} B^T & A^T \end{pmatrix}$ and choose a set of central orthogonal idempotents $\{E_1, E_2, \hdots, E_r\}$ for the group algebra. The $E_i$ serve as projection operators onto independent subspaces and give the block-diagonalization $A = \sum_i A E_i$ and $B = \sum_i B E_i$. Let $E_A = \sum_{A E_i \neq 0} E_i$ and similarly for $E_B$. Now consider the modified stabilizers
\begin{equation*}
    H^\prime_L = \begin{pmatrix}
        A & B\\
        0 & 1 - E_A
    \end{pmatrix}, \quad
    H^\prime_R = \begin{pmatrix}
        A & B\\
        1 - E_B & 0
    \end{pmatrix}.
\end{equation*}
Solutions of $H^\prime_{L/R} \begin{pmatrix} u \\ v \end{pmatrix} = 0$ contain the original logical operators of the code, $Au + Bv = 0$, while enforcing that $v$ and $u$ are restricted strictly to the supports of $A$ and $B$ for $H^\prime_L$ and $H^\prime_R$, respectively. Intuitively, this is allowing $v$ ($u$) to be nonzero but only inside the support of $A$ ($B$). Since this restriction creates a strict subset of the valid logical operators of the full code, the minimum weight of this subset is bounded from below by the true quantum distance, yielding $d(H) \leq \min\{d(H^\prime_L), d(H^\prime_R)\}$. However, the distance of $H^\prime$ is just as hard to compute as the distance of the original code. Restricting further to the case that $v = 0$ (the shortened code) provides an even looser, but now computable upper bound. Considering just solutions $Au = 0$ also includes $Z$ stabilizers which must be excluded from the distance computation. 
Since $Z$ stabilizers are elements of
$\mathrm{rowspace}(H_Z)=\im(H_Z^T)$, they are of the form
\[
H_Z^T r=
\begin{pmatrix}
Br\\
Ar
\end{pmatrix},
\qquad r\in\F_2^{\ell m}.
\]
Thus, a $Z$ stabilizer supported only on the left block has the form
$(Br,0)$ with $Ar=0$; in particular, its left component lies in
$\im(B)$.

The authors of \cite{LinPryadko2024} remove the support on $\im(B)$ by considering the (now classical) code $H_L = \begin{pmatrix} A \\ E_B \end{pmatrix}$ which gives $Au = 0$ and $E_B u = 0$, enforcing that $u \in \ker(A)$ while explicitly projecting out the stabilizer space via $E_B u = 0$, guaranteeing $u \notin \im(B)$. Likewise, define $H_R = \begin{pmatrix} B \\ E_A \end{pmatrix}$. Now, $d(H_{X/Z}) \leq \min\{d(H^\prime_L), d(H^\prime_R)\} \leq \min\{d(H_L), d(H_R)\}$.\\

\noindent {\textbf{Remark:}} Since $A$ and $B$ are classical 2D cyclic codes, the results of the previous section allow us to precisely tailor $A$ and $B$ and control their intersections. The idempotents $E_i$ correspond to the orbit idempotents $e_O$.\\

\noindent {\textbf{Remark:}} Combining the upper bound with the dimension formula $k = 2 \dim(\ker A \cap \ker B)$ \cite{bravyi2024} gives the impression that a BB code is completely determined by the intersection. However, maximizing the intersection by choosing $A = B$ trivially drops the distance to two. We will explain the role of the orbits outside the intersection in the next section.\\

Also noteworthy, \cite{symons2025sequences} describes the change in distance that occurs when using small ``seed'' BB codes to construct sequences of larger BB codes via graph lifting.

\section{BB Codes From $q$-Frobenius Orbits}
\label{sec:BBnew}
We continue under the assumption that $q$ is a prime power. Let $Q = \mathrm{BB}(a, b)$ be the bivariate bicycle code defined by the polynomials $a, b \in R = \F_q[x, y] / \idlGens{x^\ell - 1, y^m - 1}$.
After identifying $\F_q^{\ell m}$ with $R$ in the standard monomial basis, let $A$ and $B$ denote the corresponding matrix representations of the multiplication operators by $a$ and $b$, respectively. The $X$ stabilizers can be identified with the classical code $\mathcal{C}_X$ with parity checks of the form $\begin{pmatrix} A & B \end{pmatrix}$. Equivalently, $\mathcal{C}_X = \{(u, v) \in R^2 \mid au + bv = 0\}$. Defining
\begin{equation}
    \Psi_X \,:\, R^2 \longrightarrow R, \qquad (u, v) \longmapsto au + bv,
\end{equation}
we have $\mathcal{C}_X = \ker \Psi_X$ and $\im \Psi_X = \idlGens{a, b}$. Similarly, $\ker A \cong \ann\idlGens{a}$ and $\ker B \cong \ann\idlGens{b}$ are ideals in $R$. 
Hence, BB codes are two-block (index-2 quasi-cyclic) codes whose constituents are 2D cyclic codes.

Projecting the defining constraint of $C_X$, $a u+ b v = 0$ onto a single Frobenius orbit $O \in \Omega$ gives the local constituent equation
\begin{equation}
\label{eq:local_coupling}
    a_O u_O + b_O v_O = 0 \qquad (\text{over } \F_{q^{|O|}}).
\end{equation}
Let
\begin{equation}
    S_a = \{O \in \Omega \mid a_O \neq 0\}, \qquad S_b = \{O \in \Omega \mid b_O \neq 0\}.
\end{equation}
be the active supports of $a$ and $b$, respectively. We record the resulting orbit partition in both active-support and zero-set form.
Throughout this section, a hat denotes regions defined from the nonzero Wedderburn, equivalently Fourier, supports $S_a$ and $S_b$. Unhatted regions are defined from the zero sets $\mathcal{Z}_a$ and $\mathcal{Z}_b$ on the same root grid.
Also, when convenient, we identify any Frobenius-closed set of root-grid points (such as $\mathcal{Z}_a$) with the subset of $\Omega$ consisting of the orbits it contains. 
Whenever the cardinality operator $\vert{}\cdot\vert{}$ is applied to these regions, it counts the total number of root-grid points, not the number of orbits.

\begin{definition}[Active-Support Version]\leavevmode
    \begin{enumerate}
        \item \textbf{Coupled active-support region}
        \begin{equation*}
            \suppreg{\mathcal{T}}_{a, b} = S_a \cap S_b.
        \end{equation*}
        
        \item \textbf{Uncoupled active-support regions}
        \begin{equation*}
            \suppreg{\mathcal{U}}_{a, b} = S_a \setminus S_b \quad , \quad \suppreg{\mathcal{U}}_{b, a} = S_b \setminus S_a
        \end{equation*}
        
        \item \textbf{Free active-support region}
        \begin{equation*}
            \suppreg{\mathcal{F}}_{a, b} = \Omega \setminus (S_a \cup S_b)
        \end{equation*}
    \end{enumerate}
\end{definition}

\begin{definition}[Zero-Set Version]\leavevmode
    \begin{enumerate}
        \item \textbf{Coupled zero-set region}
        \begin{equation*}
            \mathcal{T}_{a, b} = \mathcal{Z}_a \cap \mathcal{Z}_b
        \end{equation*}

        \item \textbf{Uncoupled zero-set regions}
        \begin{equation*}
            \mathcal{U}_{a, b} = \mathcal{Z}_a \setminus \mathcal{Z}_b \quad , \quad \mathcal{U}_{b, a} = \mathcal{Z}_b \setminus \mathcal{Z}_a
        \end{equation*}
        
        \item \textbf{Free zero-set region}
        \begin{equation*}
            \mathcal{F}_{a,b} = \Omega \setminus (\mathcal{Z}_a \cup \mathcal{Z}_b)
        \end{equation*}
    \end{enumerate}
\end{definition}

\begin{lemma}
    \label{lem:Fouriercorr}
    The active-support and zero-set regions satisfy
    \begin{gather*}
        \suppreg{\mathcal{T}}_{a, b} \longleftrightarrow \mathcal{F}_{a, b}, \quad \suppreg{\mathcal{F}}_{a, b} \longleftrightarrow \mathcal{T}_{a, b},  \quad \suppreg{\mathcal{U}}_{b, a} \longleftrightarrow \mathcal{U}_{a, b}.
    \end{gather*}
\end{lemma}
\begin{proof}
    The spectral support of a polynomial is the set-theoretic complement of its zero set in the ambient root space $\Omega$. Thus, $S_a = \Omega \setminus \mathcal{Z}_a$ and $S_b = \Omega \setminus \mathcal{Z}_b$. Substituting these into the formulas and applying De Morgan's laws gives the desired results.
\end{proof}

The same statements can be made for the $Z$ stabilizers. By Proposition~\ref{prop:algebraic-properties},
\begin{equation*}
    \Psi_Z : R^2 \to R, \qquad (u, v) \mapsto \iota(b) u - \iota(a) v
\end{equation*}
and $\im \Psi_Z = \idlGens{\iota(a), \iota(b)} = \iota(\idlGens{a, b}).$

\subsection{Dimension}
\begin{theorem}
    \label{thm:BB-dimension}
    The dimension of $\mathrm{BB}(a, b)$ is 
    \begin{equation}
    \label{dimorbform}
        k = 2 \sum_{O \in \suppreg{\mathcal{F}}_{a, b}} |O| = 2 |\mathcal{T}_{a,b}| = 2 |\mathcal{Z}_a \cap \mathcal{Z}_b|.
    \end{equation}
\end{theorem}
\begin{proof}
    The ambient space of the code is $R^2$, so $n = 2 \dim_{\F_q}(R)$.
    The $X$ stabilizers are defined by the image of $\Psi_X$. Thus, $\rank(H_X) = \dim_{\F_q}(\langle a, b \rangle)$.
    Since the involution $\iota$ is a linear automorphism, it preserves dimensions: $\rank(H_Z) = \dim_{\F_q}(\iota(\langle a, b \rangle)) = \dim_{\F_q}(\langle a, b \rangle)$.
    
    Using the standard CSS dimension formula, if $\mathcal{C}_X = \ker(H_X)$ and $\mathcal{C}_Z = \ker(H_Z)$, then $k = \dim \mathcal{C}_X + \dim \mathcal{C}_Z - n$. Since
    \begin{equation*}
        \dim \mathcal{C}_X = n -\rank(H_X), \qquad \dim \mathcal{C}_Z = n - \rank(H_Z),
    \end{equation*}
    we obtain
    \begin{equation*}
        k = n - \rank(H_X)-\rank(H_Z) = 2 \dim_{\F_q}(R) - 2 \dim_{\F_q}(\langle a, b \rangle).
    \end{equation*}
    By Proposition~\ref{prop:algebraic-properties}(1), $\dim_{\F_q}(R) - \dim_{\F_q}(\langle a, b \rangle) = \dim_{\F_q}(\ann\langle a, b \rangle)$. Therefore, $k = 2 \dim_{\F_q}(\ann\langle a, b \rangle)$.
    Since $R$ is semisimple by assumption, the Wedderburn decomposition gives $R \cong \prod_{O \in \Omega} K_O$, where $K_O = \F_{q^{|O|}}$. Under this isomorphism, an element $r \in R$ is represented by its spectral components $(r_O)_{O \in \Omega}$. 
    
    The ideal $\ann\langle a, b \rangle$ is the intersection $\ann(a) \cap \ann(b)$. Thus, an element $r = (r_O)_{O \in \Omega}$ lies in $\ann\langle a, b \rangle$ if and only if
    \begin{equation*}
        r_O a_O = 0 \qquad \text{and} \qquad r_O b_O = 0
    \end{equation*}
    for every $O \in \Omega$. Since each $K_O$ is a field, $r_O$ is arbitrary exactly on those orbits where $a_O = b_O = 0$, and must vanish on all other orbits. 
    
    The condition $a_O = b_O = 0$ is precisely the definition of the free active-support region, $O \in \suppreg{\mathcal{F}}_{a, b}$. Thus, the annihilator is supported exclusively on this region, yielding
    \begin{equation*}
        \dim_{\F_q}(\ann\langle a, b \rangle) = \sum_{O \in \suppreg{\mathcal{F}}_{a, b}} \dim_{\F_q} K_O = \sum_{O \in \suppreg{\mathcal{F}}_{a, b}} |O|.
    \end{equation*}
    Finally, Lemma~\ref{lem:Fouriercorr} identifies the free active-support region $\suppreg{\mathcal{F}}_{a,b}$ with the common-zero region $\mathcal{T}_{a,b}=\mathcal{Z}_a \cap \mathcal{Z}_b$, yielding the desired formula.
\end{proof}

\begin{remark}
    Since $\ann \langle a, b \rangle = \ann \langle a \rangle \cap \ann \langle b \rangle$, Equation~\eqref{dimorbform} is the Frobenius-orbit version of the standard two-block dimension count. The intermediate formula $k = 2 (\dim_{\F_q}(R) - \dim_{\F_q}(\langle a, b \rangle))$ also appears in \cite{EberhardtSteffan2025}.
\end{remark}

Theorem~\ref{thm:BB-dimension} shows that $\mathcal{T}_{a, b}$ plays an analogous role to the defining set for cyclic codes. In order to construct a BB code with high dimension, one must carefully choose the polynomials $a$ and $b$ to have a large ``core'' of common roots. As previously mentioned, a similar interpretation can be deduced from the minimum distance upper bounds of \cite{LinPryadko2024}. However, there is always a tension between dimension and distance. If we attempt to maximize both formulas by setting $a = b$, the bound increases but the true minimum distance drops to $d \leq 2$. The uncoupled regions therefore cannot be ignored. These orbits make non-trivial contributions to the stabilizers, which can significantly lower the distance compared to the upper bound.

We can take the previous theorem further by showing that the reason the uncoupled orbits are not required in the dimension count is because they correspond to homologically trivial logical operators. 

\begin{remark}
    In the homological formulation of stabilizer codes, the logical $Z$ operators are identified with the homology group $H_1 = \ker(H_X) / \im(H_Z^\top)$, while the logical $X$ operators correspond to the cohomology group $H^1 = \ker(H_Z) / \im(H_X^\top)$. For topological codes, the isomorphism $H_1 \cong H^1$ is traditionally established via Poincaré duality. In our framework, this geometric duality is replaced by an exact algebraic isomorphism defined by the canonical involution $\iota(x,y) = (x^{-1}, y^{-1})$ coupled with a block exchange. The non-degenerate bilinear form of the symmetric Frobenius ring guarantees that orthogonal complements are isomorphic to ideal annihilators. Consequently, $H_Z = \begin{pmatrix} \iota(B) & -\iota(A) \end{pmatrix}$ ensures that the map $(u, v) \mapsto (\iota(v), -\iota(u))$ is a bijective dimension-preserving automorphism between the logical spaces,  establishing the duality $H_1 \cong H^1$.
\end{remark}

\begin{lemma}
    \label{lem:orbitwise_complex_split}
    The chain complex
    \begin{equation*}
        R \xrightarrow{H_Z^\top} R^2 \xrightarrow{H_X} R
    \end{equation*}
    splits as a direct sum of subcomplexes indexed by $O\in\Omega$ as 
    \begin{equation*}
        e_O R \xrightarrow{H_Z^\top|_O} (e_O R)^2 \xrightarrow{H_X|_O} e_O R,
    \end{equation*}
    where $H_X|_O : (e_O R)^2 \to e_O R$ is defined by the localized (restricted) action
    \begin{equation*}
        H_X|_O(u_O, v_O) = H_X(u_O, v_O) = a_O u_O + b_O v_O,
    \end{equation*}
    $(u_O, v_O) \in (e_O R)^2$, $a_O = a e_O$, and $b_O = b e_O$ are the projections of the check polynomials into the local field component. The map $H_Z^\top|_O$ is defined analogously. Consequently,
    \begin{equation*}
        H_1 \cong \bigoplus_{O \in \Omega} H_1(O), \qquad H_1(O) = \ker(H_X|_O) \big/ \im(H_Z^\top|_O).
    \end{equation*}
\end{lemma}
\begin{proof}
    Since each $e_O$ is central and $e_O^2 = e_O$, multiplication by $a$ and $b$ preserves each summand $e_O R$. In other words,
    \begin{equation*}
        e_O(a u) = (e_O a)u = a_O (e_O u) \in e_O R, \qquad e_O(b v)= b_O (e_O v) \in e_O R,
    \end{equation*}
    and similarly for the action in $H_Z^\top$. It follows that both maps $H_X$ and $H_Z^\top$ restrict to $(e_O R)^2$ and $e_O R$, and that the global complex is the direct sum of these orbitwise subcomplexes. Taking homology commutes with finite direct sums, yielding the claimed decomposition of $H_1$.
\end{proof}

\begin{theorem}
    \label{thm:trivial_homology}
    For every orbit $O \notin \suppreg{\mathcal{F}}_{a,b}$, the local homology group $H_1(O)$ is trivial.
\end{theorem}
\begin{proof}
    The local complex over $e_O R$ is exact whenever $(a_O, b_O) \neq (0, 0)$. We can equate the kernel and the image. We find
    $$\ker \begin{pmatrix} a_O & b_O \end{pmatrix} = e_O R \begin{pmatrix} b_O \\ -a_O \end{pmatrix} = \operatorname{im} \begin{pmatrix} b_O \\ -a_O \end{pmatrix}$$
    Thus $H_1(O) = 0$ unless $a_O = b_O = 0$.
\end{proof}

\begin{corollary}
    \label{cor:halo_free_representatives}
    Every logical class $[(u, v)]$ admits a representative $(u^\prime, v^\prime)$ such that
    $$e_O u^\prime = e_O v^\prime = 0 \qquad \text{for all } O \notin \suppreg{\mathcal{F}}_{a,b}.$$
    We may choose a representative whose spectral support is contained strictly in $\suppreg{\mathcal{F}}_{a,b}$.
\end{corollary}
\begin{proof}
    By Lemma~\ref{lem:orbitwise_complex_split} the homology class $[(u, v)]$ decomposes orbitwise. Theorem~\ref{thm:trivial_homology} implies that the local component $(e_O u, e_O v) \in \ker(H_X|_O)$ lies in $\mathrm{im}(H_Z^\top|_O)$ for every $O \notin \suppreg{\mathcal{F}}_{a,b}$. Thus there exists $w_O \in e_O R$ with $(e_O u, e_O v) = H_Z^\top|_O(w_O)$. Let $w = \sum_{O \notin \suppreg{\mathcal{F}}_{a,b}} w_O \in R$. The sum is finite and orthogonal since the idempotents $e_O$ are orthogonal. Then $(u^\prime, v^\prime) = (u, v) - H_Z^\top(w)$ is homologous to $(u, v)$. This representative has vanishing components. This means $e_O u^\prime = 0$ and $e_O v^\prime = 0$ for all $O \notin \suppreg{\mathcal{F}}_{a,b}$.
\end{proof}

\subsection{Minimum Distance}
There are multiple ways to try to bound the minimum distance of BB codes. The most natural way is to try to bound the quantum distance through the classical distances of the two 2D cyclic codes. This has the potential to miss low-weight mixed-block logical operators.

The $Z$-logical space of $Q=\mathrm{BB}(a,b)$ corresponds to the homology group
\[
H_1 = \ker H_X/\im H_Z^\top = \{(u,v) \in R^2 \mid au + bv = 0\}/\{(br, -ar) \mid r \in R\}.
\]
Similarly, the $X$-logical space corresponds to the cohomology group $H^1$ obtained from the dual complex.

\begin{remark} 
    Since the inversion map is a coordinate permutation, the $X$- and $Z$-side distance formulas are identical after replacing $a$, $b$ by $\iota(a)$, and $\iota(b)$.
\end{remark}

For $c \in R$ we set 
\begin{equation*}
    d_R(\ann\idlGens{c}) = \min \{\wt_R(u) \mid 0 \neq u \in R , c u = 0\}.
\end{equation*}
If $\ann\idlGens{c} = 0$, we set $d_R(\ann\idlGens{c}) = \infty$. For the polynomials $a, b$, we write 
\begin{equation*}
    d_a = d_R(\ann\idlGens{a}), \quad d_b = d_R(\ann\idlGens{b}).
\end{equation*}
For a pair of $\F_q$-subspaces $U$ and $V$ of $R$ such that $V\subseteq U$, we write 
\begin{equation*}
    d_R(U) = \min \{\wt(u) \mid 0 \neq u \in U \}, \quad d_R(U \setminus V) = \min \{\wt(u) \mid  u \in U \setminus V \}.
\end{equation*}
Finally, for $a$ and $b$ in $R$ we define the \emph{colon ideals}
\begin{equation*}
    \idlGens{b : a} = \{u \in R \mid a u \in \langle b \rangle \}, \quad \idlGens{a : b} = \{u\in R \mid b u \in \langle a \rangle \}.
\end{equation*}

\begin{theorem}
\label{thm:colon-lower-bound}
Define
\[
E_a
=
d_R\bigl(
\ann\idlGens{a}
\setminus
b\ann\idlGens{a}
\bigr),
\qquad
E_b
=
d_R\bigl(
\ann\idlGens{b}
\setminus
a\ann\idlGens{b}
\bigr),
\]
and
\[
N_{a,b}
=
d_R(\idlGens{b:a})
+
d_R(\idlGens{a:b}).
\]
Then the $X$- and $Z$-distances of $Q=\mathrm{BB}(a,b)$ satisfy
\[
d_X,d_Z\ge\min\{E_a,E_b,N_{a,b}\}.
\]
Consequently,
\[
d(Q)\ge\min\{E_a,E_b,N_{a,b}\}.
\]
\end{theorem}

\begin{proof}
    We prove the $Z$-distance bound. The $X$-distance bound follows by applying the same argument to the transposed checks.
    
    A $Z$-type operator is represented by a pair $(u, v) \in R^2$. It commutes with all $X$-checks precisely when $a u + b v = 0$. Thus, the $Z$-kernel is $K_Z = \{(u, v) \in R^2 \mid a u + b v = 0 \}$. The $Z$-stabilizer subspace is the image of $H_Z^\top$. Since $H_Z = \begin{pmatrix} M_{\iota(b)} & -M_{\iota(a)}\end{pmatrix}$, its transpose is $H_Z^\top = \begin{pmatrix} M_b \\ -M_a \end{pmatrix}$. Thus, we have $\im(H_Z^\top) = \{(b r, -a r) \mid r \in R\}$. Therefore, a nontrivial $Z$-logical operator is an element of $K_Z \setminus \{(b r, -a r) \mid r \in R\}$.
    
    Let $(u, v) \in K_Z$ be a nontrivial $Z$-logical operator. We split our analysis into three cases. First, suppose that $v = 0$ and $u \neq 0$. Then the kernel equation gives $a u = 0$, so $u \in \ann\idlGens{a}$. Moreover, $(u, 0)$ is a $Z$-stabilizer if and only if there exists $r \in R$ such that $(u, 0) = (b r, -a r)$. This is equivalent to $-a r = 0 \implies a r = 0$, meaning $r \in \ann\idlGens{a}$, and $u = b r$. Equivalently, $u \in b \ann\idlGens{a}$. Since $(u, 0)$ is nontrivial, we must have $u \in \ann\idlGens{a} \setminus b \ann\idlGens{a}$. Hence, $\wt_R(u, v) = \wt_R(u) \geq E_a$.
    
    Second, suppose that $u = 0$ and $v \neq 0$. Then the kernel equation gives $b v = 0$, so $v \in \ann\idlGens{b}$. Moreover, $(0, v)$ is a $Z$-stabilizer if and only if there exists $r \in R$ such that $(0, v) = (b r, -a r)$. This forces $b r = 0$, meaning $r \in \ann\idlGens{b}$, and $v = -a r$. Since $-1$ is a unit in $\F_q$, the ideal $a \ann\idlGens{b}$ is invariant under scalar multiplication by $-1$, meaning $\{-ar \mid r \in \ann\idlGens{b}\} = a \ann\idlGens{b}$. Thus, $v \in a \ann\idlGens{b}$. Since $(0, v)$ is nontrivial, we must have $v \in \ann\idlGens{b} \setminus a \ann\idlGens{b}$. Hence, $\wt(u, v) = \wt(v) \geq E_b$.
    
    Finally, suppose that both $u$ and $v$ are nonzero. In this case, $a u + b v = 0$ implies that $a u = -b v \in \langle b \rangle$ and $b v = -a u \in \langle a \rangle$. Once again, since ideals absorb units, this is equivalent in terms of colon ideals to $u \in \idlGens{b : a}$ and $v \in \idlGens{a : b}$. It follows that $\wt(u, v) = \wt(u) + \wt(v) \geq d_R(\idlGens{b : a}) + d_R(\idlGens{a : b}) = N_{a, b}$.
    
    We conclude from these observations that every nontrivial $Z$-logical operator falls into one of the three cases we considered above. The result follows.

   For the $X$ side, the transposed checks correspond to the pair
\[
(c,d)=(\iota(b),-\iota(a)).
\]
The inversion map is a Hamming-weight-preserving ring automorphism and
satisfies
\[
\iota(\ann\idlGens{a})
=\ann\idlGens{\iota(a)},
\qquad
\iota(\idlGens{b:a})
=\idlGens{\iota(b):\iota(a)}.
\]
It also maps
$b\ann\idlGens{a}$ onto
$\iota(b)\ann\idlGens{\iota(a)}$.
Since multiplication by $-1$ does not change an ideal or its minimum
distance, the three terms for $(c,d)$ are
\[
E_c=E_b,\qquad E_d=E_a,\qquad N_{c,d}=N_{a,b}.
\]
Applying the preceding argument to $(c,d)$ therefore gives
\[
d_X\ge\min\{E_a,E_b,N_{a,b}\}.
\] 
\end{proof}

Theorem~\ref{thm:colon-lower-bound} highlights our design principle for bivariate bicycle codes: Ensuring large minimum distances for the individual annihilator ideals $\ann\idlGens{a}$ and $\ann\idlGens{b}$ is insufficient to guarantee a large quantum code distance. Specifically, the mixed-block logical operators governed by the colon ideals $\idlGens{b : a}$ and $\idlGens{a : b}$ must be controlled. We explore this further in the appendix.
In this regard, to make Theorem~\ref{thm:colon-lower-bound} explicitly computable, in our next result, we will translate these abstract ideal operations into the spectral domain. 
When $R$ is semisimple, every ideal is uniquely determined by its zero set of Frobenius orbits, allowing us to express the required annihilator and colon ideals as classical 2D cyclic codes.
\begin{remark}
    It is worth pointing out that the previous theorem does not assume semisimplicity of $R$.
\end{remark}
For an orbit set $S \subseteq \Omega$, define the $q$-ary abelian code
\begin{equation*}
    \mathcal{C}(S) = \{f \in R \mid f(P) = 0\ \text{ for every}\ P \in S\}.
\end{equation*}
While Theorem~\ref{thm:colon-lower-bound} provides a general lower bound, the un-excluded mixed-block term $N_{a, b} = d_R(\idlGens{b:a}) + d_R(\idlGens{a:b})$ has a practical limitation when applied to sparse check polynomials. Since $ab \in \idlGens{b}$, the stabilizer vector $(b, -a) \in \im(H_Z^\top)$ always lies in $\idlGens{b:a} \times \idlGens{a:b}$, capping $N_{a, b} \le \wt(a) + \wt(b)$. To prevent sparse check polynomials from capping the mixed-block distance, we refine $N_{a, b}$ by excluding stabilizers from one block at a time.

\begin{theorem}
    \label{thm:refined-colon-bound}
    Assume that $R = \F_q[x, y]/\idlGens{x^\ell - 1, y^m - 1}$ is semisimple. For a BB code $Q = \mathrm{BB}(a, b)$, define the alternating stabilizer-excluded colon distance by
    \begin{equation}
        N^*_{a, b} = \min \Big\{ d_R(\idlGens{b:a} \setminus \idlGens{b}) + d_R(\idlGens{a:b}), \; d_R(\idlGens{b:a}) + d_R(\idlGens{a:b} \setminus \idlGens{a}) \Big\}.
    \end{equation}
    Then the minimum $Z$-distance of $Q$ satisfies $d_Z \ge \min \{ E_a, \, E_b, \, N^*_{a, b} \}$, where $E_a = d_R(\ann\idlGens{a} \setminus b\ann\idlGens{a})$ and $E_b = d_R(\ann\idlGens{b} \setminus a\ann\idlGens{b})$.
\end{theorem}
\begin{proof}
    Let $(u, v) \in \ker(H_X) \setminus \im(H_Z^\top)$ be a nontrivial $Z$-logical operator with both $u \neq 0$ and $v \neq 0$. Rearranging the syndrome condition $au + bv = 0$ gives $au = -bv \in \idlGens{b}$ and $bv = -au \in \idlGens{a}$, which implies $u \in \idlGens{b:a}$ and $v \in \idlGens{a:b}$ by definition of the colon ideal. A pair $(u, v)$ is a stabilizer $(br, -ar) \in \im(H_Z^\top)$ if and only if $u \in \idlGens{b}$ and $v \in \idlGens{a}$ simultaneously. Since $(u, v)$ is nontrivial, at least one component must lie outside the stabilizer ideal, either $u \in \idlGens{b:a} \setminus \idlGens{b}$ or $v \in \idlGens{a:b} \setminus \idlGens{a}$.
    
    To see why membership in both principal ideals forces $(u, v)$ to be a stabilizer generated by a single element $r \in R$, suppose $(u, v) \in \ker(H_X)$ satisfies $u \in \idlGens{b}$ and $v \in \idlGens{a}$. We can write $u = b r_1$ and $v = -a r_2$ for some $r_1, r_2 \in R$. Substituting these into the syndrome equation $au + bv = 0$ yields $ab(r_1 - r_2) = 0$, meaning $r_1 - r_2 \in \ann\idlGens{ab}$. When $\gcd(q, \ell m) = 1$, the Wedderburn decomposition of the semisimple ring $R$ gives $\ann\idlGens{ab} = \ann\idlGens{a} + \ann\idlGens{b}$. We may therefore decompose the difference as $r_1 - r_2 = s_a + s_b$, where $s_a \in \ann\idlGens{a}$ and $s_b \in \ann\idlGens{b}$. Setting $r = r_1 - s_b = r_2 + s_a$ yields $br = b(r_1 - s_b) = b r_1 = u$ and $-ar = -a(r_2 + s_a) = -a r_2 = v$. This confirms that $(u, v) = (br, -ar) \in \im(H_Z^\top)$ is generated by a single element $r \in R$. Thus, any nontrivial logical operator $(u, v)$ must have at least one component outside these principal ideals. Taking the minimum Hamming weight over the two alternating exclusions yields $\wt(u, v) \ge N^*_{a, b}$.
\end{proof}

\begin{theorem}
\label{thm:orbit-partition-colon-bound}
    Using the notation above, 
    \begin{equation*}
        \ann\idlGens{a} = \mathcal{C}(\mathcal{U}_{b, a} \sqcup \mathcal{F}_{a, b}), \quad \ann\idlGens{b} = \mathcal{C}(\mathcal{U}_{a, b} \sqcup \mathcal{F}_{a, b}), \quad \idlGens{b : a} = \mathcal{C}(\mathcal{U}_{b, a}), \quad \idlGens{a : b} = \mathcal{C}(\mathcal{U}_{a, b}).
    \end{equation*}
    Moreover, 
    \begin{equation*}
        b\ann\idlGens{a} = \mathcal{C}(\mathcal{T}_{a, b} \cup \mathcal{U}_{b, a} \cup \mathcal{F}_{a, b}), \quad a\ann\idlGens{b} = \mathcal{C}(\mathcal{T}_{a, b} \cup \mathcal{U}_{a, b} \cup \mathcal{F}_{a, b}).
    \end{equation*}
    Consequently, the one-block terms in the BB distance bound are 
    \begin{equation*}
        E_a = d_R(\mathcal{C}(\mathcal{U}_{b, a} \cup \mathcal{F}_{a, b}) \setminus \mathcal{C}(\mathcal{T}_{a,b}\cup \mathcal{U}_{b, a}\cup \mathcal{F}_{a, b})), \quad 
        E_b = d_R(\mathcal{C}(\mathcal{U}_{a, b} \cup \mathcal{F}_{a, b}) \setminus \mathcal{C}(\mathcal{T}_{a, b} \cup \mathcal{U}_{a, b} \cup \mathcal{F}_{a, b})),
    \end{equation*}
    and the mixed-block term is $N_{a, b} = d_R(\mathcal{C}(\mathcal{U}_{b, a})) + d_R(\mathcal{C}(\mathcal{U}_{a, b}))$.
    In particular, the $Z$-distance satisfies 
    \begin{equation*}
        d_Z \geq \min \{d_R(\mathcal{C}(\mathcal{U}_{b, a} \cup \mathcal{F}_{a, b})), \ d_R(\mathcal{C}(\mathcal{U}_{a, b} \cup \mathcal{F}_{a, b})), \ d_R(\mathcal{C}(\mathcal{U}_{b, a})) + d_R(\mathcal{C}(\mathcal{U}_{a, b}))\}.
    \end{equation*}
\end{theorem}

\begin{remark}\label{rem:conservative_bound}
    Note that the final lower bound on $d_Z$ in Theorem~\ref{thm:orbit-partition-colon-bound} is strictly weaker than the bound in Theorem~\ref{thm:colon-lower-bound} since it drops the set differences removing stabilizers (for instance, bounding $d_R(\mathcal{C}(\mathcal{U}_{b,a} \cup \mathcal{F}_{a,b}))$ rather than $E_a$). Since $d_R(U \setminus V) \ge d_R(U)$ for any nested subspaces $V \subset U$, ignoring stabilizer subtraction yields a conservative but readily computable lower bound that can be evaluated directly from the minimum distances of classical 2D cyclic codes.
\end{remark}

\begin{proof}
    Since $R$ is semisimple, after extending scalars to a splitting field the ring decomposes coordinatewise over the root grid $\mu_\ell \times \mu_m$.
    Let $P$ be a root-grid point. The ideal $\langle a \rangle$ has arbitrary coordinate at $P$ when $a(P) \neq 0$, and has zero coordinate at $P$ when $a(P) = 0$. Thus,
    \begin{equation*}
        \ann\idlGens{a} = \{f \in R \mid f(P) = 0 \text{ whenever } a(P) \neq 0\}.
    \end{equation*}
    Equivalently, $\ann\idlGens{a}$ has defining zero-orbit set $\Omega \setminus \mathcal{Z}_a$. Since $\mathcal{Z}_a = \mathcal{T}_{a, b} \sqcup \mathcal{U}_{a, b}$, we get $\Omega \setminus \mathcal{Z}_a = \mathcal{U}_{b, a} \sqcup \mathcal{F}_{a, b}$. Therefore,
    \begin{equation*}
        \ann\idlGens{a} = \mathcal{C}(\mathcal{U}_{b, a} \cup \mathcal{F}_{a, b}).
    \end{equation*}
    Similarly,
    \begin{equation*}
        \ann\idlGens{b} = \mathcal{C}(\mathcal{U}_{a, b} \cup \mathcal{F}_{a, b}).
    \end{equation*}
    Now consider $\idlGens{b : a} = \{u \in R \mid a u \in \langle b \rangle \}$. At a root-grid point $P$, if $b(P) \neq 0$, then the ideal $\langle b \rangle$ imposes no coordinate restriction at $P$. If $b(P) = 0$, then membership $a u \in \langle b \rangle$ requires $a(P) u(P) = 0$. This imposes the condition $u(P) = 0$ precisely when $b(P) = 0$ and $a(P) \neq 0$. In orbit notation, this translates to
    $\mathcal{Z}_b \setminus \mathcal{Z}_a  = \mathcal{U}_{b, a}$. Therefore,
    \begin{equation*}
        \idlGens{b : a} = \mathcal{C}(\mathcal{U}_{b, a}).
    \end{equation*}
    A similar argument gives 
    \begin{equation*}
        \idlGens{a : b} = \mathcal{C}(\mathcal{U}_{a, b}). 
    \end{equation*}
    Next, $b \ann\idlGens{a}$ is obtained by multiplying elements spectrally supported on $\mathcal{Z}_a = \mathcal{T}_{a, b} \sqcup \mathcal{U}_{a, b}$ by $b$. On $\mathcal{T}_{a, b}$, we have $b = 0$, while on $\mathcal{U}_{a, b}$, we have $b \neq 0$. Thus $b \ann\idlGens{a}$ is supported exactly on $\mathcal{U}_{a, b}$. Its defining zero-orbit set is therefore, 
    \begin{equation*}
        \Omega \setminus \mathcal{U}_{a, b} = \mathcal{T}_{a, b} \sqcup \mathcal{U}_{b, a} \sqcup \mathcal{F}_{a, b}, 
    \end{equation*}
    implying that 
    \begin{equation*}
        b \ann\idlGens{a} = \mathcal{C}(\mathcal{T}_{a, b}\cup \mathcal{U}_{b, a} \cup \mathcal{F}_{a, b}).
    \end{equation*}
    Similarly, we have 
    \begin{equation*}
        a \ann\idlGens{b} = \mathcal{C}(\mathcal{T}_{a, b}\cup \mathcal{U}_{a, b} \cup \mathcal{F}_{a, b}).
    \end{equation*}
    Substituting these identities into our previous Theorem~\ref{thm:colon-lower-bound}, gives the asserted formulas for $E_a, E_b$, and $N_{a,b}$. The final lower bound follows since 
    \begin{equation*}
        d_R(\mathcal{C}(\mathcal{U}_{b, a} \cup \mathcal{F}_{a, b}) \setminus \mathcal{C}(\mathcal{T}_{a, b}\cup \mathcal{U}_{b, a} \cup \mathcal{F}_{a, b})) \geq d_R(\mathcal{C}(\mathcal{U}_{b, a} \cup \mathcal{F}_{a, b})), 
    \end{equation*}
    and similarly for the $b$ side. The $X$-distance statement follows by applying the same argument to the inverted pair $(\iota(b), \iota(a))$.
\end{proof}

\begin{corollary}
    \label{cor:spectral-one-sided-colon}
    Under the assumptions of Theorem~\ref{thm:orbit-partition-colon-bound}, the alternating stabilizer exclusions of Theorem~\ref{thm:refined-colon-bound} correspond to set-difference distances of classical 2D cyclic codes:
    \begin{equation}
        d_R(\idlGens{b:a} \setminus \idlGens{b}) = d_R\Big( \mathcal{C}(\mathcal{U}_{b, a}) \setminus \mathcal{C}(\mathcal{U}_{b, a} \cup \mathcal{T}_{a, b}) \Big),
    \end{equation}
    and symmetrically,
    \begin{equation}
        d_R(\idlGens{a:b} \setminus \idlGens{a}) = d_R\Big( \mathcal{C}(\mathcal{U}_{a, b}) \setminus \mathcal{C}(\mathcal{U}_{a, b} \cup \mathcal{T}_{a, b}) \Big).
    \end{equation}
    Consequently, the refined mixed-block distance bound $N^*_{a, b}$ is computable from the minimum Hamming weights of classical 2D cyclic codes:
    \begin{equation}
        N^*_{a, b} = \min \left\{ 
        \begin{aligned}
            &d_R\Big( \mathcal{C}(\mathcal{U}_{b, a}) \setminus \mathcal{C}(\mathcal{U}_{b, a} \cup \mathcal{T}_{a, b}) \Big) + d_R\big( \mathcal{C}(\mathcal{U}_{a, b}) \big), \\
            &d_R\big( \mathcal{C}(\mathcal{U}_{b, a}) \big) + d_R\Big( \mathcal{C}(\mathcal{U}_{a, b}) \setminus \mathcal{C}(\mathcal{U}_{a, b} \cup \mathcal{T}_{a, b}) \Big)
        \end{aligned}
        \right\}.
    \end{equation}
\end{corollary}

\begin{proof}
    Since $R$ is semisimple, every principal ideal is generated by an idempotent and corresponds uniquely to the classical 2D cyclic code defined by its Frobenius-orbit zero set. By Theorem~\ref{thm:orbit-partition-colon-bound}, $\idlGens{b:a} = \mathcal{C}(\mathcal{U}_{b, a})$ and $\idlGens{a:b} = \mathcal{C}(\mathcal{U}_{a, b})$. Furthermore, since $\mathcal{Z}_b = \mathcal{T}_{a, b} \sqcup \mathcal{U}_{b, a}$, the principal stabilizer ideal is $\idlGens{b} = \mathcal{C}(\mathcal{T}_{a, b} \cup \mathcal{U}_{b, a})$, which forms a subcode of $\mathcal{C}(\mathcal{U}_{b, a})$. Substituting these classical code identities into the alternating exclusion formula for $N^*_{a, b}$ yields the stated result.
\end{proof}

\begin{remark}
    \label{rem:when-nstar-matters}    
In Example~\ref{subsec:bb510} below as well as in the BCH-based constructions the check polynomials $a$ and $b$ are dense polynomials. Consequently the stabilizer ceiling is large. The un-excluded colon sum $N_{a,b}$ is already sufficient to certify the code distance.
    
  By contrast, for standard BB codes in the literature defined by sparse weight-$3$ trinomials, the un-excluded bound can never exceed $N_{a, b} \le 3 + 3 = 6$. When evaluating such sparse presentations on grids where the true quantum distance is $d > 6$, the subcode subtraction in Corollary~\ref{cor:spectral-one-sided-colon} is essential. Indeed, removing the subcodes $\mathcal{C}(\mathcal{U}_{b, a} \cup \mathcal{T}_{a, b}) = \idlGens{b}$ and $\mathcal{C}(\mathcal{U}_{a, b} \cup \mathcal{T}_{a, b}) = \idlGens{a}$ eliminates the weight-$3$ check polynomials from their respective colon ideals, breaking the $\wt(a) + \wt(b) = 6$ stabilizer ceiling and allowing $N^*_{a, b}$ to certify $d > 6$.
\end{remark}

\begin{example}
    \label{subsec:bb510}
    To demonstrate how to use Theorem~\ref{thm:orbit-partition-colon-bound} in practice, let $\ell = 15$, $m = 17$, and consider a coprime BB code over $\F_2[z] / \idlGens{z^{255} - 1}$. The polynomial $z^{255} - 1$ splits over $\F_{2^8}$, and $\F_{2^8}$ contains a primitive $255$-th root of unity, which we denote by $\alpha$. The multiplicative order of $2$ modulo $255$ is $8$, so the $2$-cyclotomic cosets modulo $255$ have sizes dividing $8$. There are $35$ such cosets in total.
    
We begin by assigning each $2$-cyclotomic coset to one of four classes as recorded in Table~\ref{tab:bb510-partition}. Each entry $C_{s, r}$ denotes the full $2$-cyclotomic coset of size $r$.
    \begin{table}[ht]
        \centering
        \renewcommand{\arraystretch}{1.3}
        \begin{tabular}{c|l|c}
            \textbf{Class} & \textbf{2-Cyclotomic Cosets With Cardinality} & $\sum |C_{s, r}|$ \\
            \hline
            $\mathcal{T}_{a, b}$ & $C_{9, 8}, C_{13, 8}, C_{17, 4}, C_{21, 8}, C_{25, 8}, C_{29, 8}, C_{37, 8}$ & $52$ \\
            $\mathcal{U}_{a, b}$ & $C_{11, 8}, C_{23, 8}, C_{27, 8}, C_{43, 8}, C_{45, 8}, C_{51, 4}, C_{53, 8}, C_{85, 2}, C_{87, 8}, C_{91, 8}$ & $70$ \\
            $\mathcal{U}_{b, a}$ & $C_{0, 1}, C_{1, 8}, C_{3, 8}, C_{5, 8}, C_{7, 8}, C_{15, 8}, C_{39, 8}, C_{63, 8}, C_{127, 8}$ & $65$ \\
            $\mathcal{F}_{a, b}$ & $C_{19, 8}, C_{31, 8}, C_{47, 8}, C_{55, 8}, C_{59, 8}, C_{61, 8}, C_{95, 8}, C_{111, 8}, C_{119, 4}$ & $68$
        \end{tabular}
        \caption{2-cyclotomic cosets $C_{s, r}$ modulo $255$ with cardinality $r$.}
        \label{tab:bb510-partition}
    \end{table}
    
    With the choices in the table, $|\mathcal{Z}_a| = |\mathcal{T}_{a, b}| + |\mathcal{U}_{a, b}| = 122$ and $|\mathcal{Z}_b| = |\mathcal{T}_{a, b}| + |\mathcal{U}_{b, a}| = 117$. 
   We take $a, b$ to be the canonical cyclic-ideal generators satisfying $\langle a\rangle=\mathcal{C}(\mathcal{Z}_{a})$ and $\langle b\rangle=\mathcal{C}(\mathcal{Z}_{b})$, which yields check weights $\wt(a)=53$ and $\wt(b)=54$.
To apply Theorem~\ref{thm:orbit-partition-colon-bound}, we determine the relevant annihilator and colon ideals.

    Since $\idlGens{a, b} = \idlGens{a} + \idlGens{b} = \mathcal{C}(\mathcal{T}_{a, b})$, we have $\dim_{\F_2} \ann\idlGens{a, b} = |\mathcal{T}_{a, b}| = 52$. Hence the dimension of the code is $k = 104$.
    
    Following Remark~\ref{rem:conservative_bound}, we evaluate the conservative lower bounds of Theorem~\ref{thm:orbit-partition-colon-bound} by applying the BCH bound directly to the classical 2D cyclic codes. Each term is bounded by an explicit consecutive run in the corresponding defining set:
    \begin{enumerate}
        \item First, $\ann\idlGens{a} = \mathcal{C}(\mathcal{U}_{b, a} \cup \mathcal{F}_{a, b})$. The defining set $\mathcal{U}_{b, a} \cup \mathcal{F}_{a, b}$ contains a cyclic consecutive run of length $29$. Hence, $E_a \ge d_R(\ann\idlGens{a}) = d_R(\mathcal{C}(\mathcal{U}_{b, a} \cup \mathcal{F}_{a, b})) \ge 30$.
        \item Second, $\ann\idlGens{b} = \mathcal{C}(\mathcal{U}_{a, b} \cup \mathcal{F}_{a, b})$. The defining set $\mathcal{U}_{a, b} \cup \mathcal{F}_{a, b}$ contains a consecutive run of length $22$. Hence, $E_b \ge d_R(\ann\idlGens{b}) = d_R(\mathcal{C}(\mathcal{U}_{a, b} \cup \mathcal{F}_{a, b})) \ge 23$.
        \item Third, $\idlGens{a : b} = \mathcal{C}(\mathcal{U}_{a, b})$ and $\idlGens{b : a} = \mathcal{C}(\mathcal{U}_{b, a})$. The set $\mathcal{U}_{a, b}$ contains a consecutive run of length $9$, so $d_R(\idlGens{a : b}) = d_R(\mathcal{C}(\mathcal{U}_{a, b})) \ge 10$. The set $\mathcal{U}_{b, a}$ contains a cyclic consecutive run of length $13$, so $d_R(\idlGens{b : a}) = d_R(\mathcal{C}(\mathcal{U}_{b, a})) \ge 14$. Therefore, $N_{a,b} = d_R(\idlGens{a : b}) + d_R(\idlGens{b : a}) \ge 24$.
    \end{enumerate}
    Combining these estimates gives $d_Z \ge \min\{30, 23, 24\} = 23$.
    
    The inversion map $\iota: z \mapsto z^{-1}$ sends exponents $i \mapsto -i \pmod{255}$. It sends every cyclic block of consecutive exponents to another cyclic block of the same length. Hence, the BCH-designed lower bounds above are preserved after applying $\iota$. Thus, $d_X \ge 23$ and $d = \min\{d_X, d_Z\} \ge 23$. The final parameters are $[\![510, 104, d \ge 23]\!]$. The rate and estimated $k d^2 / n$ for our code are 
    \begin{equation*}
        r = \frac{104}{510} \approx 0.204, \quad \frac{k d^2}{n} \ge \frac{104\cdot 23^2}{510} \approx 108.
    \end{equation*}
\end{example}

\begin{example}
    \label{ex:nstar-improves-certificate}
    We present a code where the alternating stabilizer exclusion of
    Theorem~\ref{thm:refined-colon-bound} (via Corollary~\ref{cor:spectral-one-sided-colon}) \emph{strictly improves} the
    certified $Z$-distance over the un-excluded colon sum. Consider the
    $[\! [42, 2]\! ]$ BB code over the semisimple ring
    $R = \F_2[x, y] / \idlGens{x^3 - 1, y^7 - 1}$,
    defined by the weight-$4$ check polynomials
    \begin{equation}
        a(x, y) = xy + xy^4 + x^2 y + x^2 y^3
        \quad \text{and} \quad
        b(x, y) = y^2 + x^2 y^3 + x^2 y^4 + x^2 y^6 .
    \end{equation}
    A direct enumeration of the classical 2D cyclic codes gives the
    un-excluded colon distances
    \begin{equation}
        d_R\big(\mathcal{C}(\mathcal{U}_{b, a})\big) = d_R(\idlGens{b:a}) = 2
        \quad \text{and} \quad
        d_R\big(\mathcal{C}(\mathcal{U}_{a, b})\big) = d_R(\idlGens{a:b}) = 2,
    \end{equation}
    so the un-excluded mixed-block term is $N_{a, b} = 2 + 2 = 4$.

    Excluding the stabilizer subcodes $\idlGens{b} = \mathcal{C}(\mathcal{U}_{b, a} \cup \mathcal{T}_{a, b})$
    and $\idlGens{a} = \mathcal{C}(\mathcal{U}_{a, b} \cup \mathcal{T}_{a, b})$ lifts each
    escaping minimum weight from $2$ to $3$:
    \begin{equation}
        d_R\big(\idlGens{b:a} \setminus \idlGens{b}\big) = 3
        \quad \text{and} \quad
        d_R\big(\idlGens{a:b} \setminus \idlGens{a}\big) = 3.
    \end{equation}
    The refined mixed-block term of Theorem~\ref{thm:refined-colon-bound} is therefore
    \begin{equation}
        N^*_{a, b} = \min \{ 3 + 2, \; 2 + 3 \} = 5 > 4 = N_{a, b}.
    \end{equation}
    The single-block terms evaluate to $E_a = 7$ and $E_b = 9$. Since the
    mixed-block term is the binding one ($N_{a, b} < \min\{E_a, E_b\}$), the
    improvement propagates directly to the distance certificate:
    \begin{equation}
        d_Z \ge \min \{ E_a, \, E_b, \, N_{a, b} \} = \min \{ 7, 9, 4 \} = 4,
        \qquad
        d_Z \ge \min \{ E_a, \, E_b, \, N^*_{a, b} \} = \min \{ 7, 9, 5 \} = 5.
    \end{equation}
    The alternating stabilizer exclusion thus tightens the guaranteed
    $Z$-distance from $4$ to $5$. Both bounds are valid: the true minimum
    distance of this code is $d_Z = 7$.
\end{example}

\begin{remark}
    Echoing our previous statements, we emphasize that we have chosen to use the BCH bound in the previous example to simplify the discussion. It is, however, often loose, and numerically computing some of the other bounds may provide better estimates.
\end{remark}

\subsection{Metachecks}
We define metachecks $M_X$ and $M_Z$ in the usual way:
\begin{equation}\label{eq:chain}
    \begin{tikzcd}[column sep=large]
        C_3 \arrow[r, "M_Z^T"] & C_2 \arrow[r, "H_Z^T"] & C_1 \arrow[r, "H_X"] & C_0 \arrow[r, "M_X"] & C_{-1}.
    \end{tikzcd}
\end{equation}
Given (unlifted) stabilizers $H_X = \begin{pmatrix} a & b\end{pmatrix}$, the space of valid syndromes is $\idlGens{a, b} = a u + b v$ for $u, v \in R$. This is an ideal in $R$. The classical code encoding this information provides a set of metachecks for the $X$ stabilizers. The parity checks $\ann \idlGens{a, b}$ completely specify the space. Since $R$ is semisimple, $\ann \idlGens{a, b}$ is principal and there exists a $m \in R$ such that $\idlGens{m} = \ann \idlGens{a, b}$. 

If $m s \neq 0$, then the measured syndrome $s$ is not a valid syndrome for the BB code and a measurement error has occurred. Ideally, we want $m$ to be sparse but there is no mathematical guarantee that such an element exists. (On the contrary, $m$ is often a dense central idempotent.) We can ease this burden by instead requiring two (or more) sparse generators which \emph{together} generate the space: $\idlGens{m_1, m_2} = \ann \idlGens{a, b}$. The metachecks are now of the form $\begin{pmatrix} m_1 s & m_2 s \end{pmatrix} = \begin{pmatrix} 0 & 0 \end{pmatrix}$, however, the code and its properties are still determined by the 2D cyclic code $\ann \idlGens{a, b} \subset R$.

Similar to $m$, there are no guarantees that such sparse $m_1, m_2$ exist. Instead of trying to prove restricted theoretical results along these lines, we treat this as a design problem.
\begin{example}
    \label{ex:metacheck1}
    Start with a (1D) cyclic code with known good distance and generator polynomial $g$. Choose two coprime codewords $u$ and $v$ and define $a = g u$ and $b = g v$. Then, $\idlGens{a, b} = \idlGens{g u, g v} = g \idlGens{u, v} = \idlGens{g}$. We search for sparse metachecks in $\ann\idlGens{g}$. This is equivalent to searching for low-weight codewords of the dual code.

    This idea may be extended to 2D, however, the fact that $R$ is no longer a Euclidean domain makes the search for low-weight codewords more difficult.
\end{example}

\begin{remark}
    Roughly speaking, for an exact two-block group algebra metacheck chain complex, 
    restricting the metachecks to weights no more than $w$ is equivalent to saying that the code in the previous example is a classical locally-recoverable code (LRC) with locality $r = w - 1$. It may be possible to exploit results from the extensive LRC literature (e.g.~\cite{tamo_cyclic_2016}) to design metachecks for BB codes from first principles.
\end{remark}

Before presenting a concrete realization of the previous example, we record a useful result.

\begin{theorem}
    \label{thm:exactness-equivalence}
    Let $M_X = \begin{pmatrix} m_{1, X} \\ m_{2, X} \end{pmatrix}$ and $M_Z = \begin{pmatrix} m_{1, Z} \\ m_{2, Z} \end{pmatrix}$ be metachecks for a dimension $k$ code $\mathrm{BB}(a, b)$ over $R$ with $\langle m_{1, X}, m_{2, X} \rangle = \operatorname{ann}\langle a, b \rangle$ and $\langle m_{1, Z}, m_{2, Z} \rangle = \operatorname{ann}\langle \iota(a), \iota(b) \rangle$ respectively. Let $d_{M_X}$ and $d_{M_Z}$ denote the minimum distances of the codes defined by $\ker(M_X)$ and $\ker(M_Z)$ respectively.

    \begin{enumerate}
        \item[(1)] $\operatorname{im}(M_Z^\top) = \ker(H_Z^\top)$ and $\operatorname{im}(M_X^\top) = \ker(H_X^\top)$. Furthermore $\ker(M_X) = \operatorname{im}(H_X)$ and $\ker(M_Z) = \operatorname{im}(H_Z)$.
        
        \item[(2)] $\operatorname{rank} M_X = \operatorname{rank} M_Z = k/2$. In particular $M_X$ captures all $X$-stabilizer redundancies and $M_Z$ captures all $Z$-stabilizer redundancies.
        
        \item[(3)] $d_{M_X} \leq \min \{ \operatorname{wt}(a), \operatorname{wt}(b) \}$ and $d_{M_Z} \leq \min \{ \operatorname{wt}(\iota(a)), \operatorname{wt}(\iota(b)) \}$.
    \end{enumerate}
\end{theorem}
\begin{proof}
    We prove the results for $X$ since the $Z$ side follows identical arguments after applying $\iota$.

    (1) The image of $H_X$ is the ideal $\langle a, b \rangle$. A vector $\sigma$ is in the kernel of $M_X$ when $m_{1, X} \sigma = 0$ and $m_{2, X} \sigma = 0$. This means $\ker(M_X) = \operatorname{ann}\langle m_{1, X}, m_{2, X} \rangle$. The hypothesis states $\langle m_{1, X}, m_{2, X} \rangle = \operatorname{ann}\langle a, b \rangle$. Taking the annihilator of both sides yields $\ker(M_X) = \operatorname{ann}(\operatorname{ann}\langle a, b \rangle)$. The double annihilator property gives $\operatorname{ann}(\operatorname{ann}\langle a, b \rangle) = \langle a, b \rangle$. This proves $\ker(M_X) = \operatorname{im}(H_X)$. The chain condition $M_X H_X = 0$ transposes to $H_X^\top M_X^\top = 0$. This gives $\operatorname{im}(M_X^\top) \subseteq \ker(H_X^\top)$. The dimension of $\langle a, b \rangle$ is $\ell m - k/2$. The rank-nullity theorem gives $\operatorname{rank} M_X = \ell m - (\ell m - k/2) = k/2$. The dimension of $\operatorname{im}(M_X^\top)$ is $k/2$. The dimension of $\ker(H_X^\top)$ is $\ell m - \operatorname{rank} H_X = k/2$. Equal dimensions with an inclusion force equality. This yields $\operatorname{im}(M_X^\top) = \ker(H_X^\top)$.

    (2) Part 1 proves $\operatorname{rank} M_X = k/2$. The identical argument proves $\operatorname{rank} M_Z = k/2$. The matrices capture all stabilizer redundancies.

    (3) If $a \neq 0$ let $u=1$ and $w=0$. Write $a = u a + w b$. This is an element of $\operatorname{im}(H_X)$. Part 1 dictates $\operatorname{im}(H_X) = \ker(M_X)$. We see that $d_{M_X} \leq \operatorname{wt}(a)$. If $b \neq 0$ use $u=0$ and $w=1$. This gives $d_{M_X} \leq \operatorname{wt}(b)$. Taking the minimum over $\{ \operatorname{wt}(a), \operatorname{wt}(b) \}$ finishes the proof. The $Z$-side is identical noting that $\operatorname{wt}(a) = \operatorname{wt}(\iota(a))$ since the involution $\iota$ is a permutation of the monomial basis.
\end{proof}

\begin{remark}
    Here, we are explicitly choosing the metachecks to form an exact sequence. This is true if and only if the metacheck operators completely generate the space. 
\end{remark}

\begin{example}
    \label{ex:metacheck-common-factor}
    We illustrate Theorem~\ref{thm:exactness-equivalence} by constructing an explicit metacheck complex over the semisimple ring $R = \F_2[z]/\langle z^{15}-1 \rangle$.
    
    Let $\mathcal{C} \subset R$ be the cyclic binary $[15, 11, 3]$ Hamming code. Its generator polynomial is the primitive quartic $g(z) = 1 + z + z^4$, and its check polynomial is $h(z) = (z^{15}-1)/g(z)$, which has degree 11. Choose two coprime polynomials $u, v \in R$ of low weight, for instance, $u(z) = 1 + z^2$ and $v(z) = 1 + z + z^3$, with $\gcd(u, v, z^{15}-1) = 1$, and define the BB constructor polynomials by
    \begin{align*}
        a(z) &= u(z) h(z) = 1 + z + z^4 + z^8 + z^9 + z^{10} + z^{11} + z^{13}\\
        b(z) &= v(z) h(z) = 1 + z^3 + z^7 + z^8 + z^9 + z^{10} + z^{12} + z^{14}.
    \end{align*}
    The $X$-syndrome code has dimension $\dim \langle h \rangle = 15 - \deg(h) = 4$ and the BB code has dimension 22.

For the metachecks, we simply take $m_X$ to be the generator polynomial of $\mathcal{C}$, which generates the entire ideal on its own. The dual of the Hamming code is the $[15, 4, 8]$ Simplex code. Hence, $d_{M_X} = 8$, matching the upper bound $d_{M_X} \le \min\{\wt(a), \wt(b)\}$. 
Repeating for $Z$, we take $m_Z = 1+z^3+z^4$ with $d_{M_Z} = 8$.
\end{example}

\section{Automorphisms}
\label{sec:BBauts}
First established in~\cite{CRSS1998}, computational tools such as~\cite{Sayginel2024} identify fault-tolerant logical Clifford gates of a stabilizer code by mapping its symplectic check matrix to an enlarged linear code and extracting global permutation automorphisms, including dualities that interchange $X$- and $Z$-checks. In a complementary direction, the fold-transversal construction of Breuckmann and Burton~\cite{BreuckmannBurton2024} builds Clifford gates from $ZX$-dualities of a CSS code. Eberhardt and Steffan~\cite{EberhardtSteffan2025} specialize this to two-block group-algebra and bivariate bicycle codes using the homological and ideal-theoretic structure of the group algebra. The original BB codes paper~\cite{bravyi2024} already exploits the translation symmetry of these codes in its syndrome circuit and uses the fold-transversal machinery of~\cite{BreuckmannBurton2024} to extend the reach of its logical measurements.

For semisimple BB codes, and for two-block group algebra codes in general, the zero sets alone do not determine the coordinate-permutation automorphism group of the stabilizers. To recover it, we view $H_X$ and $H_Z$ as index-$2$ quasi-cyclic modules inside $R^2$. 

Before defining the symmetry group we fix our conventions and terminology, since several inequivalent notions of ``automorphism group of a code'' coexist in the literature. 
First, in this section, $q$ denotes a prime number. 
For a linear code $C \subseteq \F_q^n$, the \emph{permutation automorphism group} $\Aut_{\mathrm{perm}}(C)$ is the group of coordinate permutations preserving $C$, and the \emph{monomial automorphism group} enlarges it by nonzero scalar multiplication on individual coordinates. Over $\F_2$ the only nonzero scalar is $1$, so the two notions coincide; for $q>2$ they differ, and the monomial group is the relevant one, since coordinatewise scaling by $\F_p^\times$ is a nontrivial symmetry. For cyclic codes the subgroup generated by the cyclic shifts is always present. A multiplier map must satisfy stabilizer and slope conditions to survive as an automorphism. Multiplier maps are not automatically present.

For qudit stabilizer codes over a prime field $\F_q$, the classical object is the automorphism group of the associated additive code, defined inside the wreath product $\Gamma \wr S_n\cong \Gamma^n \rtimes S_n$ of coordinate permutations and local symmetries $\Gamma$ of the single-qudit Pauli group. Modulo phases, the single-qudit Pauli group is $\F_q \times \F_q$, and since $p$ is prime, the local Clifford group acting on it modulo Paulis is strictly linear:
\begin{equation*}
    \mathrm{Cl}_1 / \mathcal P_1 \cong \mathrm{Sp}_2(\F_q) = \mathrm{SL}_2(\F_q).
\end{equation*}
This replaces the binary coincidence $\mathrm{Cl}_1 / \mathcal{P}_1 \cong \GL_2(\F_2) \cong S_3$. Note that $\Sp_2 = \SL_2$ since the symplectic condition on a $2 \times 2$ matrix is exactly $\det = 1$.

\begin{definition}
    \label{def:autlc}
    Let $\mathrm{Cl}_1 \wr S_n$ denote the group of local-Clifford--permutation operators on $n$ qudits. The \emph{local Clifford automorphism group} of a stabilizer code $Q$ with stabilizer group $\mathcal{S}$ is
    \begin{equation*}
        \Aut_{\mathrm{LC}}(Q) = \{U\in \mathrm{Cl}_1 \wr S_n \mid U \mathcal{S} U^\dagger = \mathcal{S}\}.
    \end{equation*}
    Equivalently, if $H$ is the full symplectic check matrix of $Q$ over $\F_q$, then $U \in \Aut_{\mathrm{LC}}(Q)$ if and only if there is an invertible check-mixing matrix $M_C$ over $\F_q$ satisfying the linear condition
    \begin{equation}
        \label{eq:autocond}
        M_C H = H\mathcal{M}_V,
    \end{equation}
    where $\mathcal{M}_V \in \mathrm{Sp}_{2n}(\F_q)$ is the symplectic matrix induced by $U$, so that the physical transformation $v \mapsto v \mathcal M_V$ is an $\F_p$-linear symplectic map on $\F_q^{2n}$.
\end{definition}

Every $U \in \Aut_{\mathrm{LC}}(Q)$ normalizes $\mathcal{S}$ and hence induces a linear symplectic automorphism $M_{\mathrm{log}} \in \Sp_{2k}(\F_q)$ of the logical Pauli group $\mathcal L \cong \mathcal{C}(\mathcal{S}) / \mathcal{S}$. The physical automorphism implements a nontrivial logical Clifford gate exactly when $M_{\mathrm{log}} \neq \Id$. We use the right-action convention $v \mapsto v \mathcal{M}_V$ for symplectic row vectors $v \in \F_q^{2n}$, with standard form $J_n = \begin{psmallmatrix} 0 & I_n \\ -I_n & 0 \end{psmallmatrix}$, so that $\mathcal{M}_V J_n \mathcal{M}_V^\top = J_n$. 
Let $L$ be the $2k \times 2n$ matrix whose rows form a normalized symplectic logical basis $\bar{X}_1, \dots, \bar{X}_k, \bar{Z}_1, \dots, \bar{Z}_k$ with $L J_n L^\top = J_k$. Then $L \mathcal{M}_V = M_{\mathrm{log}} L + S_0$ with every row of $S_0$ in the stabilizer row space. Since the stabilizer space is symplectic orthogonal to the logical space, right-multiplication by $J_n L^\top$ removes the $S_0$-term, giving
\begin{equation}
    \label{eq:logaction}
    M_{\mathrm{log}} = L \mathcal{M}_V J_n L^\top J_k^{-1}.
\end{equation}
This converts a physical stabilizer automorphism into its explicit logical Clifford action.

\begin{remark}
The permutation automorphism group of a CSS code, defined below, is the subgroup of $\Aut_{\mathrm{LC}}(Q)$ coming from pure coordinate permutations that preserve the $X$- and $Z$-check row spaces separately. For $p>2$, this group is enlarged by coordinate scalings in $\F_q^\times$, which act on each side without mixing them. The larger local Clifford group may also include Hadamard-type operations exchanging the two sides. Entangling Cliffords are not local Clifford operations, though they may still be fault-tolerant Clifford automorphisms. We treat them separately below.
\end{remark}

The ideal $\idlGens{a, b} \subseteq R$ records which $X$-syndromes can occur. The object controlling coordinate symmetries is the rowspace of $H_X$ over $\F_p$, equivalently the image of $H_X^\top$ inside $R^2$.
\begin{definition}
    The \emph{$X$-check module} and the \emph{$Z$-check module} are the $R$-submodules
    \begin{equation*}
        \mathcal{R}_X = \im(H_X^\top) = R \cdot (\iota(a), \iota(b)), \qquad \mathcal{R}_Z = \im(H_Z^\top) = R\cdot (b,\,-a),
    \end{equation*}
    both contained in $R^2$. While the syndrome ideals are cyclic codes of length $\ell m$, the check modules are quasi-cyclic codes of index two and length $2\ell m$.
\end{definition}

Since $R$ is semisimple, $R = \bigoplus_{O \in \Omega} e_O R$, where $e_O$ is the primitive idempotent of the $p$-Frobenius orbit $O$ and $\Omega$ is the orbit set. Every $R$-submodule $\mathcal{M}\subseteq R^2$ decomposes constituent-wise as $\mathcal M = \bigoplus_O e_O \mathcal{M}$ with $e_O \mathcal{M} \subseteq (e_O R)^2$, and its spectral support is $\{O \mid e_O \mathcal{M} \neq 0\}$.

For $\mathcal{R}_X$, the constituent at $O$ is the line\footnote{The line going through the point $(x, y)$ in the standard Euclidean plane is given by all scalar multiples $(r x, r y)$ for $r \in \mathbb{R}$. This goes through the origin ($r = 0$) and has slope $r y / r x = y / x$.} $e_O \mathcal{R}_X = e_O R \cdot (e_O\iota(a), e_O\iota(b)) \subseteq (e_O R)^2$. The orbit-region language of the earlier sections describes these constituents exactly:
\begin{itemize}
    \item on the free active-support region, both coordinates vanish and the constituent is zero;
    \item on the uncoupled active-support regions, the constituents are the coordinate lines $e_O R \cdot (1, 0)$ and $e_O R \cdot(0, 1)$;
    \item on the coupled active-support region, both coordinates are nonzero and the constituent is the graph of the \emph{slope}
        \begin{equation*}
            \lambda_O = e_O\iota(b) \big/ e_O\iota(a) \in (e_O R)^\times, \qquad e_O\mathcal R_X = e_O R \cdot(1, \lambda_O).
        \end{equation*}
\end{itemize}
Thus, the spectral support of $\mathcal R_X$ is the union of the coupled active-support region and the two uncoupled active-support regions of $(\iota(a), \iota(b))$ plus additional information from the slope on the coupled active-support region.\footnote{Each slope $\lambda_O$ ranges over the full multiplicative group $(e_O R)^\times$. This group has size $p^{\vert{}O\vert{}}-1$. The slope data carries structural information for all fields, including $\F_2$, as constituents form field extensions.} A similar discussion applies to $\mathcal{R}_Z$.

\subsection{Permutation Automorphisms}
A \emph{permutation automorphism} of $Q$ is a permutation $\pi \in S_{2\ell m}$ of the physical qudit coordinates such that
\begin{equation*}
    \pi(\mathcal{R}_X) = \mathcal{R}_X, \qquad \pi(\mathcal{R}_Z) = \mathcal{R}_Z.
\end{equation*}
Equivalently, $\pi$ is a pure coordinate permutation satisfying the stabilizer automorphism condition \eqref{eq:autocond} with no $X/Z$-basis mixing. We denote the resulting group by $\Aut_{\mathrm{perm}}(Q)$. When applied to a classical code $\mathcal C\subseteq\F_q^{2N}$, the same notation $\Aut_{\mathrm{perm}}(\mathcal C)$ denotes its coordinate permutation automorphism group, so that by definition
\begin{equation*}
    \Aut_{\mathrm{perm}}(Q) = \Aut_{\mathrm{perm}}(\mathcal{R}_X) \cap \Aut_{\mathrm{perm}}(\mathcal{R}_Z).
\end{equation*}
For $q > 2$, the multiplicative group $\F_q^\times$ contributes a symmetry through the local Clifford multiplier $M_c = \mathrm{diag}(c, c^{-1}) \in \mathrm{Sp}_2(\F_q)$ applied simultaneously to every qudit coordinate pair. This map scales the $X$-check module by $c$ and the $Z$-check module by $c^{-1}$, sending $\mathcal{R}_X \mapsto c\mathcal{R}_X = \mathcal{R}_X$ and $\mathcal{R}_Z \mapsto c^{-1}\mathcal{R}_Z = \mathcal{R}_Z$.
Since $\det M_c = 1$, it is strictly symplectic for all $c \in \F_q^\times$. By contrast, uniform scalar multiplication $c \cdot \mathrm{Id}$ acts as $\mathrm{diag}(c, c)$ with determinant $c^2$, which rescales the symplectic form and is therefore a Clifford automorphism only when $c = \pm 1$. Independent per-block scalars are likewise not automorphisms in general: scaling only the first block by $c \ne 1$ sends $a u + b v = 0$ to $c\,au + bv = (c-1) au$, which is nonzero on the coupled active-support region.

\begin{remark}
    \label{rem:tanner-decoding}
    A permutation automorphism whose induced check-mixing matrix $M_C$ in \eqref{eq:autocond} is itself a permutation matrix is precisely an automorphism of the Tanner graphs of $H_X$ and $H_Z$. It permutes qudit nodes and check nodes while preserving adjacency. All of the structured symmetries constructed in this section, namely translations, multipliers, block swaps, and the $ZX$-dualities below, are of this form. Over $\F_q$, the edges of the Tanner graph carry the nonzero entries of $H_X$ and $H_Z$ as labels, and a graph automorphism is required to preserve those labels up to the allowed $\F_q^\times$ rescaling of each check.
    
    In the automorphism-ensemble decoding (AED) strategy~\cite{geiselhart2021automorphism, geiselhart2022automorphism, mandelbaum2023generalized, mandelbaum2024improved, koutsioumpas2025automorphism}, each of an ensemble of belief propagation (BP) decoders is fed the received vector or syndrome transformed by a different automorphism, and its correction is pulled back through the inverse automorphism. The ensemble helps only when its members produce distinct corrections. This requires care. If BP is exactly equivariant under an automorphism $\pi$, then $\pi^{-1}(\mathrm{BP}(\pi(s))) = \mathrm{BP}(s)$ and that ensemble member would be redundant.
\end{remark}

Diagonal translations by elements of $\Z_\ell \times \Z_m$ always give permutation automorphisms. The mapping is $(u,v) \mapsto (tu, tv)$ where $t$ is a monomial group element. Multiplication by $t$ commutes with multiplication by every element of the commutative group algebra $R$.  
Thus $\mathbb{Z}_\ell \times \mathbb{Z}_m \leq \Aut_{\mathrm{perm}}(Q)$.

Multiplier symmetries require more care. Over a prime field, a group automorphism $\mu_{\mathbf{s}} \colon g \mapsto g^{\mathbf{s}}$ for $\mathbf{s} = (s_x, s_y) \in (\Z_\ell)^\times \times (\Z_m)^\times$ extends linearly to a ring automorphism of $R$, sending $R \cdot (\iota(a), \iota(b))$ to $R \cdot (\mu_{\mathbf{s}}(\iota(a)), \mu_{\mathbf{s}}(\iota(b)))$.

\begin{proposition}
    \label{prop:kernel-suffices}
    Let $\Theta \colon R^2 \to R^2$ be a monomial Euclidean isometry satisfying $\Theta^\top \Theta = \Id$ (a coordinate permutation possibly multiplied by scalars $c = \pm 1$). If $\Theta(\ker H_X) = \ker H_X$ and $\Theta(\ker H_Z) = \ker H_Z$, then $\Theta$ preserves the CSS stabilizer of $Q$, that is $\Theta(\mathcal{R}_X) = \mathcal{R}_X$ and $\Theta(\mathcal{R}_Z) = \mathcal{R}_Z$.
\end{proposition}

\begin{proof}
    As a coordinate isometry, $\Theta$ commutes with orthogonal complement, $\Theta(U^\perp) = \Theta(U)^\perp$ for every $\F_q$-subspace $U \subseteq R^2$. Applying this to $\ker H_X$ and $\ker H_Z$ and using $\Theta(\ker H_X) = \ker H_X$, $\Theta(\ker H_Z) = \ker H_Z$ gives $\Theta((\ker H_X)^\perp) = (\ker H_X)^\perp$ and $\Theta((\ker H_Z)^\perp) = (\ker H_Z)^\perp$. For any parity-check matrix $H$, its row space is $\mathrm{rowspace}(H) = \ker(H)^\perp$, so $\mathcal{R}_X = \mathrm{rowspace}(H_X) = (\ker H_X)^\perp$ and $\mathcal{R}_Z = \mathrm{rowspace}(H_Z) = (\ker H_Z)^\perp$. Hence, $\Theta(\mathcal{R}_X) = \mathcal{R}_X$ and $\Theta(\mathcal{R}_Z) = \mathcal{R}_Z$, and $\Theta$ preserves the full stabilizer.
\end{proof}

On the spectral side, $\mu_{\mathbf{s}}$ permutes the $p$-Frobenius orbits by $O \mapsto \mathbf{s}^{-1}O$. Two nonzero constituent pairs span the same line in $(e_O R)^2$ when their coordinates and slopes match. This is shown by the condition
\begin{equation*}
    e_O R \cdot (\alpha, \beta) = e_O R \cdot (\gamma, \delta) \iff \alpha \delta = \beta \gamma.
\end{equation*}
Therefore, on the coupled active-support region, we require $\mu_s(\lambda_{\mathbf{s}O}) = \lambda_{O}$ for every coupled active-support orbit $O$, since the slope $\lambda_{\mathbf{s}O}$ resides in the field component $e_{\mathbf{s}O}R$ rather than $e_{O}R$. When incorporating relative shifts $t_1, t_2$, this condition becomes $\frac{(t_2)_{O}}{(t_1)_{O}} \mu_s(\lambda_{\mathbf{s}O}) = \lambda_{O}$.

Preserving the two zero sets $\mathcal{Z}_a$ and $\mathcal{Z}_b$ is necessary but not sufficient. One must match the slope $\lambda_O$ exactly. Even over $\mathbb{F}_2$, an orbit constituent is $e_O R \cong \mathbb{F}_{2^{|O|}}$. The slope $\lambda_O$ can take $2^{|O|} - 1$ possible nonzero values. This makes the slope test a strict analytical requirement for any field. The zero sets determine dimensions and BCH-type information. The automorphism group of the quasi-cyclic check module depends on the actual elements $a$ and $b$ through their slopes.

\begin{theorem}
    \label{thm:x-kernel-symmetries}
    Let $Q = \mathrm{BB}(a, b)$. For $\mathbf{s} \in (\mathbb{Z}_\ell)^\times \times (\mathbb{Z}_m)^\times$ and shifts $t_1 = (t_{1x}, t_{1y}), t_2 = (t_{2x}, t_{2y}) \in \mathbb{Z}_\ell \times \mathbb{Z}_m$, define the block-preserving and block-swapping linear maps
    \begin{align*}
        \Phi_{\mathbf{s}, t_1, t_2}(u, v) &= (x^{t_{1x}} y^{t_{1y}} \mu_{\mathbf{s}}(u), x^{t_{2x}} y^{t_{2y}} \mu_{\mathbf{s}}(v)), \\
        \Psi_{\mathbf{s}, t_1, t_2}(u, v) &= (x^{t_{1x}} y^{t_{1y}} \mu_{\mathbf{s}}(v), x^{t_{2x}} y^{t_{2y}} \mu_{\mathbf{s}}(u)).
    \end{align*}
    Then:
    \begin{enumerate}
        \item $\Phi_{\mathbf{s}, t_1, t_2}(\ker H_X) = \ker H_X$ if and only if $\mu_{\mathbf{s}}$ stabilizes the zero sets, $\mathcal{Z}_{\mu_{\mathbf{s}}(a)} = \mathcal{Z}_a$ and $\mathcal{Z}_{\mu_{\mathbf{s}}(b)} = \mathcal{Z}_b$, and
        \begin{equation*}
            \mu_{\mathbf{s}}(a) b = x^{t_{1x}-t_{2x}} y^{t_{1y}-t_{2y}} a \mu_{\mathbf{s}}(b)
        \end{equation*}
        holds in $R$.
        \item $\Psi_{\mathbf{s}, t_1, t_2}(\ker H_X) = \ker H_X$ if and only if $\mu_{\mathbf{s}}$ swaps the zero sets, $\mathcal{Z}_{\mu_{\mathbf{s}}(a)} = \mathcal{Z}_b$ and $\mathcal{Z}_{\mu_{\mathbf{s}}(b)} = \mathcal{Z}_a$, and
        \begin{equation*}
            \mu_{\mathbf{s}}(a) a = x^{t_{2x}-t_{1x}} y^{t_{2y}-t_{1y}} b \mu_{\mathbf{s}}(b)
        \end{equation*}
        holds in $R$.
    \end{enumerate}
    A map $\Phi$ belongs to the full code permutation automorphism group $\Aut_{\mathrm{perm}}(Q)$ if and only if $\Phi(\ker H_X) = \ker H_X$ and $\Phi(\ker H_Z) = \ker H_Z$. Consequently, the block-preserving symmetries of $Q$ form the intersection subgroup of solution pairs satisfying condition (1) simultaneously for $(a, b)$ and for the $Z$-check generators $(\iota(b), -\iota(a))$. On the coupled active-support region, the identity in (1) corresponds to equality of the slopes after orbit relabeling, $\lambda_{\mathbf{s}O} = \lambda_O$.
\end{theorem}

\begin{proof}   
    We prove (1); case (2) is identical, with the block-swapping map $\Psi$ in place of $\Phi$ (equivalently, the roles of $a$ and $b$ interchanged). The analogous condition for $\ker H_Z$, required for a full permutation automorphism by Proposition~\ref{prop:kernel-suffices}, follows by the same computation with the blocks of $H_Z$ inserted.
    
    Since $R$ is semisimple, it is a product of the residue fields $e_O R$ over the $p$-Frobenius orbits $O \in \Omega$, and an identity in $R$ holds if and only if it holds on every orbit.
    
    Assume $\Phi_{\mathbf{s}, \mathbf{t}_1, \mathbf{t}_2}(\ker H_X) = \ker H_X$. The map is a coordinate permutation multiplied by scalars, so it preserves the marginal dimensions of the kernel, which forces $\mu_{\mathbf{s}}$ to stabilize $\mathcal{Z}_a$ and $\mathcal{Z}_b$. For $(u, v) \in \ker H_X$, applying $\mu_{\mathbf{s}}$ to $a u + b v = 0$ gives
    \begin{equation*}
        \mu_{\mathbf{s}}(a) \,\mu_{\mathbf{s}}(u) + \mu_{\mathbf{s}}(b) \, \mu_{\mathbf{s}}(v) = 0,
    \end{equation*}
    using that $\mu_{\mathbf{s}}$ is a ring automorphism of $R$. The transformed pair lies in $\ker H_X$ exactly when $a x^{t_{1x}} y^{t_{1y}} \mu_{\mathbf{s}}(u) + b x^{t_{2x}} y^{t_{2y}} \mu_{\mathbf{s}}(v) = 0$. We check the claimed identity orbit by orbit.
    
    On the uncoupled and free active-support regions, $O \in \mathcal{Z}_a \cup \mathcal{Z}_b$, at least one of $a, b$ vanishes on $O$; zero-set stability makes the corresponding $\mu_{\mathbf{s}}(a)$ or $\mu_{\mathbf{s}}(b)$ vanish too, so both sides of the identity are $0$ on $O$.
    
    On the coupled active-support region, $O \notin \mathcal{Z}_a \cup \mathcal{Z}_b$, all four of $a, b, \mu_{\mathbf{s}}(a), \mu_{\mathbf{s}}(b)$ are invertible in $e_O R$. Solving the kernel relation for $\mu_{\mathbf{s}}(u) = -(\mu_{\mathbf{s}}(a))^{-1} \mu_{\mathbf{s}}(b)\mu_{\mathbf{s}}(v)$ and substituting into the membership requirement, then choosing a kernel element with $\mu_{\mathbf{s}}(v) \ne 0$ on $O$ and dividing it out, gives
    \begin{equation*}
        a x^{t_{1x}} y^{t_{1y}}(\mu_{\mathbf{s}}(a))^{-1} \mu_{\mathbf{s}}(b) = b x^{t_{2x}} y^{t_{2y}}.
    \end{equation*}
    Multiplying by $x^{-t_{2x}} y^{-t_{2y}} \mu_{\mathbf{s}}(a)$ yields $\mu_{\mathbf{s}}(a) b = x^{t_{1x} - t_{2x}} y^{t_{1y} - t_{2y}} a \, \mu_{\mathbf{s}}(b)$. Since this holds on every orbit, it holds in $R$. The converse follows by retracing the substitutions, and the slope reading is the coupled-region form of the same identity.
\end{proof}

\begin{corollary}
    \label{cor:strict-multiplier}
    The strict linear multiplier $\widetilde{\mu}_{\mathbf{s}}(u, v) = (\mu_{\mathbf{s}}(u), \mu_{\mathbf{s}}(v))$, with no relative shift and no block swap, is an automorphism of $Q = \mathrm{BB}(a, b)$ if and only if $\mu_{\mathbf{s}}$ stabilizes the individual zero sets and $\mu_{\mathbf{s}}(a) \, b = a \, \mu_{\mathbf{s}}(b)$ 
    in $R$. This is precisely the slope multiplier stabilizer $\lambda_{\mathbf{s}O} = \lambda_O$ on the coupled active-support region.
\end{corollary}
\begin{proof}
    Set $\mathbf{t}_1 = \mathbf{t}_2 = \mathbf{0}$ in Theorem~\ref{thm:x-kernel-symmetries} (1). The prefactor $x^{t_{1x} - t_{2x}} y^{t_{1y} - t_{2y}}$ becomes $1$, leaving $\mu_{\mathbf{s}}(a) b = a \, \mu_{\mathbf{s}}(b)$.
\end{proof}

\begin{remark}
    \label{rem:single-block}
    It is instructive to compare with the automorphism groups of the two single-block cyclic codes $\langle a \rangle$ and $\langle b \rangle$ taken separately. Each contains the shift group $\mathbb{Z}_{\ell m}$ together with its multiplier stabilizer for $\langle a \rangle$ and likewise for $\langle b\rangle$, and for many lengths this exhausts the full permutation automorphism group of a cyclic code~\cite{BienertKlopsch2010, HuffmanPless}. The coupled pair is far more rigid: a linear multiplier must stabilize \emph{both} zero sets simultaneously and, beyond that, satisfy the exact identity of Theorem~\ref{thm:x-kernel-symmetries}, which pins the relative slope on the coupled active-support region. The symmetry group of the pair is therefore in general a small subgroup of the intersection of the two single-block groups, and small compared with either factor. Example~\ref{ex:aut-127} below realizes the strict drop: the Frobenius multiplier group $\langle p \rangle$ stabilizes both zero sets, yet no nontrivial multiplier survives the slope test.
\end{remark}

\subsection{$ZX$-Dualities \& Fourier-Type Gates}
Over $\F_2$ the duality that exchanges the two CSS sides is a transversal Hadamard, an involution. Over a general prime field $\F_q$, its role is played by the qudit Fourier gate $F$, which acts on single-qudit Paulis by $FXF^\dagger=Z$ and $FZF^\dagger=X^{-1}$ and therefore has order four rather than two. The sign in $FZF^\dagger=X^{-1}$ is the same minus sign carried by our $H_Z$, and it forces the block-swap duality to be a signed monomial map. The systematic use of such dualities to build fault-tolerant Clifford gates, the \emph{fold-transversal} framework, is due to Breuckmann and Burton~\cite{BreuckmannBurton2024}. The specialization to two-block group-algebra codes that we follow here was developed by Eberhardt and Steffan~\cite{EberhardtSteffan2025}.

\begin{definition}
    A \emph{$ZX$-duality} of a CSS code $(H_X,H_Z)$ over $\F_q$ is a monomial coordinate map $D$ on the $2\ell m$ physical qudits, that is a permutation dressed by nonzero $\F_q^\times$ scalars, such that
    \begin{equation*}
        \mathrm{rowspace}(H_X D) = \mathrm{rowspace}(H_Z), \qquad \mathrm{rowspace}(H_Z D^{-\top}) = \mathrm{rowspace}(H_X).
    \end{equation*}
    To such a duality we associate the Fourier-type symplectic matrix
    \begin{equation*}
        F_D =\begin{pmatrix}
            0 & D \\
            -D^{-\top} & 0
        \end{pmatrix},
    \end{equation*}
    representing transversal physical Fourier gates together with the coordinate map $D$. Note the two independent doublings. The map $D$ acts on the physical two-block coordinate space $R^2\cong\F_q^{2 \ell m}$, while $F_D$ acts on the symplectic Pauli space $\F_q^{4 \ell m}$ by exchanging the $X$- and $Z$-components with a sign. Over $\F_2$ the sign is invisible, $F_D$ becomes the symmetric Hadamard block, and $D$ is an ordinary permutation.
\end{definition}

Let
\begin{equation*}
    D_0 = \begin{pmatrix}
        0 & \sigma \\
        -\sigma & 0
    \end{pmatrix},
\end{equation*}
where $\sigma : g \mapsto g^{-1}$ is the coordinate inversion on one copy of $R$.

\begin{theorem}
    \label{thm:zx-always}
    For every BB code $Q = \mathrm{BB}(a, b)$ over $\F_q$, the signed block-swapping inversion $D_0$ satisfies
    \begin{equation*}
        H_X D_0 = -\sigma H_Z, \qquad H_Z D_0 = \sigma H_X,
    \end{equation*}
    where $\sigma$ on the right-hand side denotes the corresponding check-index row permutation. Consequently $D_0$ is a $ZX$-duality, and $F_{D_0}$ is a Fourier-type Clifford automorphism of $Q$. Moreover $D_0$ is antisymmetric ($D_0^\top = -D_0$), satisfies $D_0^{-\top} = D_0$ and $D_0^2 = -\Id$ (so $D_0$ has order four for $p > 2$), while its symplectic lift $F_{D_0}$ is an involution ($F_{D_0}^2 = \Id$).
\end{theorem}
\begin{proof}
    Using $M_e^\top = \sigma M_e \sigma = M_{\iota(e)}$ and $H_Z = \begin{pmatrix} M_{\iota(b)} & -M_{\iota(a)} \end{pmatrix}$, we compute
    \begin{equation*}
        H_X D_0 = \begin{pmatrix} M_a & M_b \end{pmatrix}
            \begin{pmatrix}
                0 & \sigma \\
                -\sigma & 0
            \end{pmatrix}
            =
            \begin{pmatrix}
            -M_b \sigma & M_a \sigma
            \end{pmatrix}
            = \sigma
            \begin{pmatrix}
                -M_{\iota(b)} & M_{\iota(a)}
            \end{pmatrix}
            =
            -\sigma H_Z,
    \end{equation*}
    since $M_b \sigma = \sigma M_{\iota(b)}$ and $M_a \sigma = \sigma M_{\iota(a)}$. The computation of $H_Z D_0 = \sigma H_X$ is identical with the blocks of $H_Z$ inserted. 

    Since $\sigma^2 = \Id$, we have $D_0^\top = -D_0$ and $D_0^2 = \begin{psmallmatrix} -\sigma^2 & 0 \\ 0 & -\sigma^2 \end{psmallmatrix} = -\Id$, which implies $D_0^{-1} = -D_0$ and $D_0^{-\top} = (-D_0)^{-1} = D_0$. Thus, $H_Z D_0^{-\top} = H_Z D_0 = \sigma H_X$. Since $\sigma$ is an invertible row permutation and $-1 \in \F_q^\times$, we obtain $\mathrm{rowspace}(H_X D_0) = \mathrm{rowspace}(H_Z)$ and $\mathrm{rowspace}(H_Z D_0^{-\top}) = \mathrm{rowspace}(H_X)$, confirming that $D_0$ is a $ZX$-duality and $F_{D_0}$ preserves the stabilizer group.

    Finally, $D_0^4 = \Id$, yielding order four for $D_0$ when $p > 2$ (and order two over $\F_2$). For the symplectic lift, using $D_0^{-\top} = D_0$, we calculate:
    \begin{equation*}
        F_{D_0}^2 = \begin{pmatrix} 0 & D_0 \\ -D_0 & 0 \end{pmatrix}^2 = \begin{pmatrix} -D_0^2 & 0 \\ 0 & -D_0^2 \end{pmatrix} = \begin{pmatrix} \Id & 0 \\ 0 & \Id \end{pmatrix} = \Id,
    \end{equation*}
    so $F_{D_0}$ is an involution.
\end{proof}

\begin{remark}
    \label{rem:zx-duality-structure}
    Over $\F_2$, the duality $D_0$ is the \emph{$ZX$-duality} $\tau_0$ of Eberhardt and Steffan~\cite[\S III-B]{EberhardtSteffan2025}, defined through $(f, g) \mapsto (\iota(g), \iota(f))$ on $R \times R$, and its Hadamard lift is the fold-transversal gate $H_{\tau_0}$ of~\cite[Thm.~6]{BreuckmannBurton2024}, \cite[Thm.~3]{EberhardtSteffan2025}. Over $\F_q$, the same map acquires the sign forced by $H_Z$. The underlying coordinate map $D_0$ given by $(f,g) \mapsto (\iota(g), -\iota(f))$ has order four for $p>2$, while its combined symplectic lift $F_{D_0}$ is an involution of order two. By \cite[Lem.~1]{BreuckmannBurton2024}, the set of all $ZX$-dualities is the coset $\tau_0 \cdot \Aut$ of the automorphism group. Reference~\cite{bravyi2024} does not develop $ZX$-dualities itself, but its logical-measurement constructions invoke these fold-transversal techniques, and the SWAP-implemented automorphism searches of~\cite{Sayginel2024} recover the grid translations and $ZX$-dualities computationally.
    
    In our setting, Theorem~\ref{thm:zx-always} adds the observation that at the level of the quasi-cyclic check modules, $D_0$ exists unconditionally for every matched or unmatched pair $(a, b)$ over any abelian group algebra of order coprime to $p$. It realizes the block swap of Theorem~\ref{thm:x-kernel-symmetries}(2) by mapping $D_0(\mathcal{R}_X) = \mathcal{R}_Z$, so $\Aut_{\mathrm{perm}}(\mathcal{R}_Z) = D_0 \Aut_{\mathrm{perm}}(\mathcal{R}_X) D_0^{-1}$ and
    \begin{equation*}
        \Aut_{\mathrm{perm}}(Q) = \{\pi \in \Aut_{\mathrm{perm}}(\mathcal{R}_X) \mid D_0^{-1} \pi D_0 \in \Aut_{\mathrm{perm}}(\mathcal{R}_X)\}.
    \end{equation*}
    Since its Fourier lift $F_{D_0}$ gives a canonical Fourier-type automorphism, the subgroup $\langle \Aut_{\mathrm{perm}}(Q), \, F_{D_0} \rangle \leq \Aut_{\mathrm{LC}}(Q)$ is always present. The operative question is therefore never existence, but the induced logical action \eqref{eq:logaction}, which may be trivial or not depending on the code. Without a separate computation of all local Clifford symmetries, it is not clear if this canonical subgroup exhausts $\Aut_{\mathrm{LC}}(Q)$; indeed, the computations of~\cite{Sayginel2024} show that for several of the BB codes discovered in~\cite{bravyi2024}, the full group is strictly larger than the block-monomial part.
\end{remark}

\subsection{Phase-Type \& $CZ$-Type Clifford Gates}
The single-qudit phase gate satisfies $S X S^\dagger = \omega X Z$ for a phase $\omega$ and $S Z S^\dagger = Z$; in symplectic notation this is the transvection $(x \mid z) \mapsto (x \mid z + x)$. More generally, for a coordinate map $D$ on the $2\ell m$ physical coordinates, consider
\begin{equation*}
    S_D = \begin{pmatrix}
        \Id & D\\
        0 & \Id
        \end{pmatrix},
        \qquad (x \mid z) \longmapsto (x \mid z + x D).
\end{equation*}
If $D \neq \Id$, the operator $S_D$ is generally not a local Clifford operation. When $D$ is symmetric with entries in $\F_q$, $S_D$ is implemented by $CZ$-type gates on the off-diagonal support of $D$ together with single-qudit phase gates on the diagonal; this is the phase-type fold-transversal construction of~\cite[\S2.4.2]{BreuckmannBurton2024}, realized for symmetric BB codes in~\cite{EberhardtSteffan2025}.
\begin{proposition}
    \label{prop:phase}
    Let $D$ be a coordinate map on $\F_q^{2\ell m}$. With the right-action convention and standard symplectic form $J = \begin{psmallmatrix} 0 & I \\ -I & 0 \end{psmallmatrix}$,
    \begin{equation*}
        S_D J S_D^\top = \begin{pmatrix}
            D^\top - D & I\\
            -I & 0
        \end{pmatrix},
    \end{equation*}
    so $S_D$ is symplectic if and only if $D^\top = D$.
    Assuming $D^\top = D$, the gate $S_D$ preserves the CSS stabilizer group if and only if $\mathrm{rowspace}(H_X D) \subseteq \mathrm{rowspace}(H_Z)$. 
    Since $\rank(H_X) = \rank(H_Z)$ for the two-block codes considered here, this inclusion becomes the equality $\mathrm{rowspace}(H_X D) = \mathrm{rowspace}(H_Z)$.
\end{proposition}
\begin{proof}
    The displayed product follows by block multiplication, using $J^\top = -J$. Its off-diagonal blocks are always $\pm I$, so symplecticity is exactly the vanishing of the top-left block $D^\top - D$. An $X$-stabilizer row $(h_X \mid 0)$ maps to $(h_X \mid h_X D)$, which stays in the stabilizer group if and only if $h_X D \in \mathrm{rowspace}(H_Z)$ for every $h_X \in \mathrm{rowspace}(H_X)$. The $Z$-stabilizers are fixed by $S_D$. Equal ranks turn the inclusion into equality.
\end{proof}

\begin{remark}
    Over $\F_2$ the canonical duality $D_0$ is a symmetric involution, so it doubles as a valid $CZ$-type map and $S_{D_0}$ is automatically a phase-type automorphism. Over $\F_q$ this coincidence breaks in two ways. First, the duality $D_0$ is antisymmetric with $D_0^2 = -\Id$, so it fails the symmetry test of Proposition~\ref{prop:phase} and does not give a phase-type gate. Second, the unsigned swap $\begin{psmallmatrix} 0 & \sigma \\ \sigma & 0 \end{psmallmatrix}$ is symmetric, but a direct check shows $\mathrm{rowspace}(H_X D) \neq \mathrm{rowspace}(H_Z)$ once the sign in $H_Z$ is present, so it fails the stabilizer condition. A phase-type gate therefore requires a genuinely separate symmetric map $D$ solving $\mathrm{rowspace}(H_X D ) =\mathrm{rowspace}(H_Z)$, and its existence is a property of the individual code rather than an automatic consequence of the duality.
\end{remark}

The bare transversal phase gate $S^{\otimes 2 \ell m}$ corresponds to $D = \Id$, which is symmetric. It is valid only if $\mathrm{rowspace}(H_X) \subseteq \mathrm{rowspace}(H_Z)$, which in the equal-rank setting means $\mathrm{rowspace}(H_X) = \mathrm{rowspace}(H_Z)$. Together with CSS orthogonality this forces the $X$-stabilizer row space to be self-orthogonal. Thus bare transversal phase gates are highly constrained and should not be confused with the more general entangling $CZ$-type gates $S_D$.

\begin{example}
    \label{ex:aut-127}
 
We illustrate Theorem~\ref{thm:x-kernel-symmetries}, Corollary~\ref{cor:strict-multiplier}, and Theorem~\ref{thm:zx-always} on a BB code of length $N=2\cdot127=2\cdot(2^7-1)$ over the semisimple ring $R = \mathbb{F}_2[z]/\langle z^{127}-1 \rangle$. The 2-cyclotomic cosets modulo 127 consist of $C_0=\{0\}$ and eighteen cosets of size 7. Denoting each coset by its minimal representative, we partition the eighteen size-7 cosets into disjoint sets $\mathcal{T}, \mathcal{A}, \mathcal{B}, \mathcal{U}$:
\begin{align*}
    \mathcal{T} &= C_1 \cup C_3 \cup C_5 \cup C_7 \quad (|\mathcal{T}|=28), \\
    \mathcal{A} &= C_9 \cup C_{11} \quad (|\mathcal{A}|=14), \\
    \mathcal{B} &= C_{13} \cup C_{15} \cup C_{19} \cup C_{21} \quad (|\mathcal{B}|=28), \\
    \mathcal{U} &= C_{23} \cup C_{27} \cup C_{29} \cup C_{31} \cup C_{43} \cup C_{47} \cup C_{55} \cup C_{63} \quad (|\mathcal{U}|=56).
\end{align*}
We define the check polynomials $a(z)$ and $b(z)$ using the canonical monic cyclic-ideal generators satisfying $\langle a \rangle = \mathcal{C}(\mathcal{Z}_a)$ and $\langle b \rangle = \mathcal{C}(\mathcal{Z}_b)$, where $\mathcal{Z}_a = \mathcal{T} \sqcup \mathcal{A}$ and $\mathcal{Z}_b = \mathcal{T} \sqcup \mathcal{B}$. Since $|\mathcal{T}_{a,b}| = |\mathcal{T}| = 28$, Theorem~\ref{thm:BB-dimension} gives logical dimension $k = 2|\mathcal{T}| = 56$, yielding a $[\![254, 56, d]\!]$ code.

    \emph{(1) Permutation symmetries and the slope test:}
    The common zero-set stabilizer of $a$ and $b$ is the Frobenius multiplier group $\langle 2 \rangle = \{1, 2, 4, 8, 16, 32, 64\} \leq (\Z_{127})^\times$. 
    By Corollary~\ref{cor:strict-multiplier}, for a multiplier $\mu_s \in \langle 2 \rangle$ to be a check-module automorphism, it must satisfy the slope test on the coupled active-support region, $\mu_s(a)b = a\mu_s(b)$ in $R$. More generally, allowing a relative translation as in Theorem~\ref{thm:x-kernel-symmetries}(1), we test for the existence of a shift $j \in \Z_{127}$ such that
    \begin{equation*}
        z^j \mu_s(\iota(a)) \, \iota(b) = \mu_s(\iota(b)) \, \iota(a) \quad \text{in } R, 
    \end{equation*}
    where we apply the test to the check-module generators $(\iota(a), \iota(b))$ since $\mathcal{R}_X = \im(H_X^\top) = R \cdot (\iota(a), \iota(b))$.
   We compute the Hamming weights of both products across the Frobenius multipliers $s \in \langle 2 \rangle$:
    \begin{center}
        \begin{tabular}{c c c}
            \hline
            $s$ & $\wt(z^j \mu_s(\iota(a)) \, \iota(b))$ & $\wt(\mu_s(\iota(b)) \, \iota(a))$ \\
            \hline
            $1$ & $51$ & $51$ \\
            $2$ & $59$ & $75$ \\
            $4$ & $63$ & $63$ \\
            $8$ & $67$ & $55$ \\
            $16$ & $55$ & $67$ \\
            $32$ & $63$ & $63$ \\
            $64$ & $75$ & $59$ \\
            \hline
        \end{tabular}
    \end{center}
    Since cyclic shifts preserve Hamming weight, no shift $z^j$ can equate the polynomials whenever the two weights differ. For $s \in \{4, 32\}$, the weights coincide at $63$, but inspecting the gap structure between nonzero exponents confirms that no cyclic shift exists. Thus, while the multipliers $\mu_s$ stabilize the zero sets $\mathcal{Z}_a$ and $\mathcal{Z}_b$, no nonidentity multiplier survives the slope test.

    \emph{(2) Folded Hadamard-type duality:}
    By Theorem~\ref{thm:zx-always}, the signed block-swapping inversion $D_0$ is a $ZX$-duality unconditionally, so its lift $F_{D_0}$ (a folded Hadamard gate over $\F_2$) is a valid Clifford automorphism of $Q$ with logical action given by~\eqref{eq:logaction}. By contrast, a bare unpermuted transversal Hadamard requires $\mathrm{rowspace}(H_X) = \mathrm{rowspace}(H_Z)$. This is obstructed already at the level of spectral supports: because the common-zero region is not inversion-invariant ($-T \neq T$ modulo $127$), the supports of $\mathcal{R}_X$ and $\mathcal{R}_Z$ differ. Thus $Q$ supports the folded duality $F_{D_0}$ but not the unpermuted transversal Hadamard.

    \emph{(3) Phase-type gates:}
    Over $\F_2$, $D_0^2 = \Id$, so Proposition~\ref{prop:phase} implies that the entangling $CZ$-type phase operator $S_{D_0}$ is a valid Clifford automorphism of the code, acting as a $CZ$-layer along the two-cycles of $D_0$. The bare transversal phase gate $S^{\otimes 254}$ ($D = \Id$) is obstructed by $\mathrm{rowspace}(H_X) \neq \mathrm{rowspace}(H_Z)$. Furthermore, there is no symmetric block-swap phase gate arising from an order-two group automorphism $\omega \in \Aut(\Z_{127})$ that exchanges the zero sets: the only order-two multiplier modulo $127$ is $\mu_{-1}$, which cannot map $\mathcal{Z}_a$ to $\mathcal{Z}_b$ because $|\mathcal{A}| = 14 \neq 28 = |\mathcal{B}|$.
\end{example}

\section{Lifts \& Projections}
\label{sec:lifts-and-projections}
To close, we show how the idempotent framework is the natural language to describe lifts and projections of BB codes~\cite{symons2025sequences}. 
Throughout this section, $p$ denotes the prime characteristic of the underlying field $\F_q$; when $q=p$ is prime, the $p$-cyclotomic cosets analyzed below coincide exactly with the $q$-Frobenius orbits of the code.

Intuitively, lifting defines a map from a BB code on $R = \F_p[x, y] / \idlGens{x^\ell - 1, y^m - 1}$ to a BB code on a larger ring $R^\prime = \F_p[x, y] / \idlGens{x^{\ell^\prime} - 1, y^{m^\prime} - 1}$ and projection does the opposite. Such a map is constrained by the relationship of the orbits on $R$ and $R^\prime$. In order to define this mapping, we require $\ell^\prime = h \ell$ and $m^\prime = h m$. It suffices to study the behavior of the $p$-cyclotomic cosets modulo $n$ and $h n$. As before, we require $\gcd(p, h) = 1$.

Recall that the $p$-cyclotomic coset modulo $n$ containing a representative $s \in \Z_n$ is the orbit generated by multiplication by $p$:
\begin{equation*}
    C_s^{(n)} = \{ s \cdot p^t \pmod n \mid t \ge 0 \}.
\end{equation*}
Every element $v \in C_s^{(n)}$ has exactly $h$ distinct lifts in $\mathbb{Z}_{hn}$.

\begin{definition}
    The \emph{lifted set} of a $p$-cyclotomic coset modulo $n$, $C_s^{(n)}$, is
    \begin{equation*}
        L(C_s^{(n)}) = \{ u \in \mathbb{Z}_{hn} \mid u \equiv v \pmod n \text{ for some } v \in C_s^{(n)} \}.
    \end{equation*}
\end{definition}

\begin{theorem}
    \label{thm:lift}
    Retain the notation above and assume $\gcd(p, hn) = 1$. Fix a coset $C_s^{(n)}$ of size $\kappa = |C_s^{(n)}|$. Write $g = \gcd(s, n)$. We set $g = n$ when $s = 0$. Write $n' = n/g$. Let $c$ be the integer defined by $p^\kappa = 1 + c n'$. (This exists since the congruence $s p^\kappa \equiv s \pmod n$ holds true.) Then the following hold.

    (1) The lifted set $L(C_s^{(n)})$ is invariant under multiplication by $p$ modulo $hn$. It partitions into a disjoint union of $p$-cyclotomic cosets modulo $hn$.

    (2) Suppose $\gcd(h, n) = 1$. Then the sizes of the cosets in the partition of $L(C_s^{(n)})$ are governed by the $p$-cyclotomic cosets of $\mathbb{Z}_h$. A coset $C_r^{(h)} \subseteq \mathbb{Z}_h$ contributes cosets of $L(C_s^{(n)})$ of length $\mathrm{lcm}(|C_r^{(h)}|, \kappa)$. The number of such cosets of that length is $\gcd(|C_r^{(h)}|, \kappa)$.

    (3) Suppose $h$ divides $n$ and $h$ divides $n'$.
    Let $\delta = \gcd(c, h)$. Then $L(C_s^{(n)})$ partitions into exactly $\delta$ distinct $p$-cyclotomic cosets modulo $hn$. Each coset has size $\kappa h/\delta$.

    (4) Suppose $h$ divides $n$ with no further restriction on $g$. Define $\lambda \equiv p^\kappa \pmod h$ and $\tau \equiv c \cdot (s/g) \pmod h$. Let $T$ be the affine map from $\mathbb{Z}_h$ to $\mathbb{Z}_h$ defined by $T(j) = \lambda j + \tau$. Since $\gcd(p, h) = 1$ the value $\lambda$ is a unit modulo $h$ and $T$ is a bijection. The cosets of $L(C_s^{(n)})$ modulo $hn$ are in one-to-one correspondence with the cycles of $T$. A cycle of $T$ of length $d$ corresponds to a $p$-cyclotomic coset modulo $hn$ of size $\kappa d$. Concretely the cycle through $j \in \mathbb{Z}_h$ has length equal to the least $t > 0$ with $S_t \equiv 0 \pmod{m_j}$. Here $m_j = h / \gcd((\lambda - 1)j + \tau, h)$ and $S_t = 1 + \lambda + \dots + \lambda^{t-1}$.
\end{theorem}

Parts (2) and (3) and (4) provide a partial case analysis. Cases with $1 < \gcd(h,n) < h$ and $h$ not dividing $n$ are omitted. Part (2) is the coprime case $\gcd(h,n)=1$. Parts (3) and (4) assume $h$ divides $n$. Part (3) assumes $h$ divides $n'$. Part (4) gives the general law when $h$ divides $n$. When $h$ divides $n'$ the element $s/g$ is a unit modulo $h$. In this subcase $T$ is a pure translation and every cycle has the same length. When $\gcd(\lambda-1, h)=1$ the shift $\phi(j) = (\lambda-1)j+\tau$ is a bijection of $\mathbb{Z}_h$. Here $T$ is conjugate to pure multiplication $u \mapsto \lambda u$. The splitting of $L(C_s^{(n)})$ mirrors the $p^\kappa$ cyclotomic cosets of $\mathbb{Z}_h$ scaled by $\kappa$.

\begin{proof}
(1) Let $u \in L(C_s^{(n)})$. We have $u \equiv v \pmod n$ for some $v \in C_s^{(n)}$. Then $pu \equiv pv \pmod n$. We know $pv \bmod n \in C_s^{(n)}$ since $C_s^{(n)}$ is closed under multiplication by $p$. Hence $pu \bmod hn$ is again a lift of an element of $C_s^{(n)}$. This means $pu \bmod hn \in L(C_s^{(n)})$. Thus $L(C_s^{(n)})$ is invariant under multiplication by $p$ modulo $hn$. It partitions into $p$-cyclotomic cosets modulo $hn$.

(2) Since $\gcd(h, n) = 1$ the Chinese Remainder Theorem gives a ring isomorphism $\phi$ from $\mathbb{Z}_{hn}$ to $\mathbb{Z}_h \times \mathbb{Z}_n$ defined by $\phi(u) = (u \bmod h, u \bmod n)$. Multiplication by $p$ acts componentwise under this map. For a lift $u$ of $v \in C_s^{(n)}$ write $\phi(u) = (r, v)$ with $r \in \mathbb{Z}_h$. As $u$ ranges over $L(C_s^{(n)})$ the first coordinate ranges over all of $\mathbb{Z}_h$. The second coordinate ranges over $C_s^{(n)}$. The $p$-orbit of $(r, v)$ has length $\mathrm{lcm}(|C_r^{(h)}|, |C_v^{(n)}|) = \mathrm{lcm}(|C_r^{(h)}|, \kappa)$. The pairs whose first coordinate lies in a fixed $C_r^{(h)}$ split into $\gcd(|C_r^{(h)}|, \kappa)$ orbits of that common length. Summing over the cosets $C_r^{(h)}$ of $\mathbb{Z}_h$ gives the stated counts.

(4) We prove (4) first. Part (3) is the special case $\lambda \equiv 1$. The size of the coset $C_u^{(hn)}$ containing a lift $u$ is the least $M > 0$ with $u(p^M - 1) \equiv 0 \pmod{hn}$. Reducing modulo $n$ forces $s(p^M - 1) \equiv 0 \pmod n$. Since $p$ has order $\kappa$ on $C_s^{(n)}$ this requires $\kappa \mid M$. We write $M = \kappa t$. Thus $C_u^{(hn)}$ has size $\kappa$ times the number of steps of the induced action on the fiber of lifts.

Index the $h$ lifts of a fixed $v \in C_s^{(n)}$ by $j \in \mathbb{Z}_h$ via $u_j = v + jn$. Multiplication by $p^\kappa$ fixes each residue modulo $n$. It permutes the lifts. We compute this permutation on the level $v = s$. Using $p^\kappa = 1 + c n'$ together with $n'g = n$ and $s = g(s/g)$ we first observe $s p^\kappa = s + c n' s = s + c n' g (s/g) = s + c n (s/g)$. Therefore $p^\kappa u_j = p^\kappa (s + jn) = s p^\kappa + jn p^\kappa = s + c n (s/g) + jn p^\kappa = s + (p^\kappa j + c(s/g))n \pmod{hn}$.

The coefficient of $n$ appears only modulo $h$. Reducing the lift index modulo $h$ shows that multiplication by $p^\kappa$ sends the lift with index $j$ to the lift with index $j \mapsto \lambda j + \tau \pmod h$. Here $\lambda \equiv p^\kappa \pmod h$ and $\tau \equiv c(s/g) \pmod h$.

This computes the return map at the base level $s$. Starting from a lift of $s$ and applying multiplication by $p$ exactly $\kappa$ times returns to a lift of $s$. The lift index advances by $T(j) = \lambda j + \tau$. Consequently the full $p$-orbit modulo $hn$ of a lift of $s$ visits the $\kappa$ levels of $C_s^{(n)}$ once per pass. It closes after exactly $d$ passes through the base level. Thus the cosets of $L(C_s^{(n)})$ correspond to the cycles of the affine bijection $T(j) = \lambda j + \tau$ of $\mathbb{Z}_h$. A cycle of length $d$ yields one $p$-cyclotomic coset modulo $hn$ of total size $\kappa d$.

It remains to compute the cycle length through $j$. Iterating $T$ gives $T^t(j) = \lambda^t j + \tau S_t$ with $S_t = 1 + \lambda + \dots + \lambda^{t-1}$. This means $T^t(j) - j = (\lambda^t - 1)j + \tau S_t = S_t ((\lambda - 1)j + \tau) \pmod h$ using $\lambda^t - 1 = (\lambda - 1)S_t$. Define $\phi$ by setting $\phi(j) = (\lambda - 1)j + \tau$. Then $m_j = h / \gcd(\phi(j), h)$. The condition $T^t(j) = j$ is $S_t \phi(j) \equiv 0 \pmod h$. This means $S_t \equiv 0 \pmod{m_j}$. The cycle length is the least such $t$ as claimed.

(3) Assume $h \mid n'$. Since $s/g$ is coprime to $n' = n/g$ and every prime factor of $h$ divides $n'$ we have $\gcd(s/g, h) = 1$. This means $s/g$ is a unit modulo $h$. Moreover $\lambda \equiv p^\kappa = 1 + c n' \equiv 1 \pmod h$ since $h \mid n'$. Thus $T$ is the pure translation $j \mapsto j + \tau$ with $\tau \equiv c (s/g)$. Its cycles all have the common length $h / \gcd(\tau, h)$. There are $\gcd(\tau, h)$ of them. Since $s/g$ is a unit we have $\gcd(\tau, h) = \gcd(c(s/g), h) = \gcd(c, h) = \delta$. Hence $L(C_s^{(n)})$ splits into exactly $\delta$ cosets. Each has size $\kappa (h/\delta)$. Finally $\delta \kappa (h/\delta) = h\kappa = |L(C_s^{(n)})|$ confirming the count. When $\gcd(s, n) = 1$ we have $g = 1$ and $n' = n$. The hypothesis $h \mid n'$ holds automatically.
\end{proof}

\begin{remark}
The hypothesis $h \mid n^{\prime}$ in (3) is essential and cannot be derived from $h \mid n$ alone. 
For example, $p=2$, $n=15$, $h=5$, $s=5$ gives $\kappa=2$, $g=5$, $n^\prime=3$, so $h \nmid n^\prime$. 
Here, $L(C_5^{(15)})$ splits into three cosets of sizes $2, 4, 4$ modulo $75$, not one coset of size $10$. 
Part (4) covers these cases.
\end{remark}

The opposite operation to lifting is projection. Let
$\pi\colon \Z_{hn} \to \Z_n$ be the canonical reduction modulo $n$,
given by $\pi(u) = u \bmod n$.

\begin{definition}
    The \emph{projected set} of a $p$-cyclotomic coset modulo $hn$, $C_u^{(hn)}$, is
    \begin{equation*}
        \pi(C_u^{(hn)}) = \{ v \in \Z_n \mid v \equiv w \pmod n \text{ for some } w \in C_u^{(hn)} \}.
    \end{equation*}
\end{definition}

\begin{theorem}
\label{thm:project}
Retain the notation of Theorem~\ref{thm:lift} and assume $\gcd(p, hn) = 1$. Let $u \in \Z_{hn}$, and let $s = \pi(u) = u \bmod n$ be its base projection in $\Z_n$. Write $g = \gcd(s,n)$ (with $g = n$ when $s = 0$), $n^\prime = n/g$, and let $c$ be the integer defined by $p^{\kappa} = 1 + c\,n^\prime$. Let $K = \vert{}C_u^{(hn)}\vert{}$ be the size of the cover orbit and $\kappa = \vert{}C_s^{(n)}\vert{}$ the size of the base orbit. Then the following hold.
\begin{enumerate}
    \item[(1)] The projected set $\pi(C_u^{(hn)})$ is exactly the full $p$-cyclotomic coset $C_s^{(n)}$ modulo $n$. The size of the cover coset is an exact integer multiple of the base coset size, $K = w \cdot \kappa$ for some integer $w \ge 1$ called the \emph{wrapping multiplicity}. The restriction $\pi\vert{}_{C_u^{(hn)}}\colon C_u^{(hn)} \to C_s^{(n)}$ is a uniform $w$-to-$1$ covering map of orbits.

    \item[(2)] Suppose $\gcd(h, n) = 1$. Let $r = u \bmod h \in \Z_h$. Then the wrapping multiplicity is
    \begin{equation*}
        w = \frac{\operatorname{lcm}\!\big(|C_r^{(h)}|,\, \kappa\big)}{\kappa}.
    \end{equation*}

    \item[(3)] Suppose $h \mid n^\prime$. Let $\delta = \gcd(c, h)$. Then the wrapping multiplicity is the same for every cover coset in the fiber above $C_s^{(n)}$, namely the uniform constant $w = h / \delta$ 
    and $L(C_s^{(n)})$ splits into exactly $\delta$ such cover cosets. 

    \item[(4)] Suppose $h \mid n$, with no further restriction on $g$. Write the cover element as $u = s + j n$ for some $j \in \Z_h$. Then $w$ is exactly the cycle length of the lift index $j$ under the affine return map $T(j) = \lambda j + \tau \pmod h$, with $\lambda \equiv p^{\kappa}$ and $\tau \equiv c\,(s/g) \pmod h$, defined in Theorem~\ref{thm:lift}(4). Explicitly, $w$ is the least $t > 0$ with $S_t \equiv 0 \pmod{m_j}$, where $m_j = h/\gcd\!\big((\lambda-1)j + \tau,\, h\big)$ and $S_t = 1 + \lambda + \dots + \lambda^{t-1}$. 
\end{enumerate}
\end{theorem}

\begin{proof}
\emph{(1)} By definition, the elements of $C_u^{(hn)}$ are generated by $u p^t \bmod hn$ for $t \ge 0$. Projecting gives $\pi(u p^t \bmod hn) = u p^t \bmod n$. Since $u \equiv s \pmod n$, this reduces to $s p^t \bmod n$. As $t$ ranges over $\Z_{\ge 0}$, $s p^t \bmod n$ sweeps out exactly the base orbit $C_s^{(n)}$. Moreover $K$ is the least positive integer with $u p^K \equiv u \pmod{hn}$. Reducing modulo $n$ gives $s p^K \equiv s \pmod n$, so the fundamental period $\kappa$ of $s$ divides $K$; write $K = w \cdot \kappa$. Since multiplication by $p$ is invertible, the orbit traverses the base coset exactly $w$ times before closing, so exactly $w$ distinct elements of $C_u^{(hn)}$ map to each element of $C_s^{(n)}$.

\emph{(2)} By the Chinese Remainder Theorem, $u \leftrightarrow (r, s) \in \Z_h \times \Z_n$, and multiplication by $p$ acts componentwise. The length of the $p$-orbit of $(r,s)$ is the $\operatorname{lcm}$ of the component orbit lengths, so $K = \operatorname{lcm}\!\big(\vert{}C_r^{(h)}\vert{}, \kappa\big)$. Since $K = w \cdot \kappa$, solving for $w$ gives the stated fraction. This is the projection dual of Theorem~\ref{thm:lift}(2).

\emph{(3)} Since $C_u^{(hn)} \subseteq L(C_s^{(n)})$, Theorem~\ref{thm:lift}(3) applies: under $h \mid n^{\prime}$, the set $L(C_s^{(n)})$ partitions into exactly $\delta = \gcd(c,h)$ cover cosets, each of size $K = \kappa \cdot (h/\delta)$, independently of the lift chosen. Substituting $K = w \cdot \kappa$ yields $w = h/\delta$, uniformly over the fiber.

\emph{(4)} By Theorem~\ref{thm:lift}(4), applying $p^{\kappa}$ to the lift $u_j = s + jn$ returns an element in the same fiber above $s$ and updates the lift index via $j \mapsto T(j) = \lambda j + \tau \pmod h$. The $p$-orbit modulo $hn$ visits the $\kappa$ levels of $C_s^{(n)}$ once per pass and closes only when the lift index returns to itself, after exactly $w$ passes, where $w$ is the cycle length of $j$ under $T$. Hence $K = w \cdot \kappa$, and the closed form for $w$ via $m_j$ and $S_t$ is that of Theorem~\ref{thm:lift}(4).
\end{proof}

Parts (1)--(4) of Theorem~\ref{thm:project} transfer verbatim to the group algebra once we track orbit sums
rather than idempotents. For a $p$-cyclotomic coset $C_u^{(hn)}$ define its orbit sum
$\sigma_u = \sum_{z \in C_u^{(hn)}} x^{z} \in \widetilde{R}$, and let
$\rho\colon \widetilde{R} = \F_p[x]/\idlGens{x^{hn}-1} \to R = \F_p[x]/\idlGens{x^{n}-1}$ be the canonical
ring projection $x \mapsto x$ (equivalently $x^z \mapsto x^{z \bmod n}$).

\begin{proposition}
    \label{prop:orbitsum}
    With the notation of Theorem~\ref{thm:project},
    \begin{equation*}
        \rho(\sigma_u) \;=\; w \cdot \sigma_s ,
        \qquad \sigma_s = \sum_{v \in C_s^{(n)}} x^{v}.
    \end{equation*}
\end{proposition}
\begin{proof}
    Since $\rho$ sends $x^z \mapsto x^{z \bmod n}$,
    $\rho(\sigma_u) = \sum_{z \in C_u^{(hn)}} x^{z \bmod n}$. By Theorem~\ref{thm:project}(1), the multiset
    $\{z \bmod n : z \in C_u^{(hn)}\}$ covers each element of $C_s^{(n)}$ exactly $w$ times, so
    grouping equal terms gives $\rho(\sigma_u) = w \sum_{v \in C_s^{(n)}} x^{v} = w\, \sigma_s$.
    Over $\F_p$ the coefficient $w$ is reduced modulo $p$, so the image is zero if and only if $p \mid w$.
\end{proof}

Lifting and projection are not mutually inverse at the level of individual cosets. A base coset generally
lifts to a union of several cover cosets and each cover coset projects back onto the whole base coset. The next
result makes this fiber structure precise and shows it is exactly accounted for by the wrapping multiplicities.

\begin{theorem}
\label{thm:liftproject}
Assume $\gcd(p,hn)=1$. Fix a base coset $C_s^{(n)}$ of size $\kappa$, and recall that
$L(C_s^{(n)}) = \{\, z \in \Z_{hn} : z \bmod n \in C_s^{(n)} \,\} = \pi^{-1}\!\big(C_s^{(n)}\big)$ is its full
preimage in $\Z_{hn}$. Then:
\begin{enumerate}
    \item[(a)] \emph{(Projection is a left inverse of lifting.)} Every $p$-cyclotomic coset
    $C \subseteq L(C_s^{(n)})$ of $\Z_{hn}$ satisfies $\pi(C) = C_s^{(n)}$; conversely each such $C$ arises as a lift
    of $C_s^{(n)}$. In particular $\pi\big(L(C_s^{(n)})\big) = C_s^{(n)}$.

    \item[(b)] \emph{(Fiber decomposition.)} The preimage $L(C_s^{(n)})$ decomposes into $r$ disjoint cover cosets
    $C^{(1)},\dots,C^{(r)}$ with wrapping multiplicities $w_1,\dots,w_r$, and
    \begin{equation*}
        |L(C_s^{(n)})| = h\,\kappa, \qquad \sum_{i=1}^{r} w_i = h .
    \end{equation*}

    \item[(c)] \emph{(Orbit-sum conservation.)} At the level of group-algebra orbit sums, projection is
    additive over the fiber and returns the base orbit sum with total weight $h$:
    \begin{equation*}
        \sum_{i=1}^{r} \rho\big(\sigma_{C^{(i)}}\big)
        \;=\; \Big(\sum_{i=1}^{r} w_i\Big)\,\sigma_s
        \;=\; h\,\sigma_s ,
    \end{equation*}
    with $\rho$ and $\sigma_{(\cdot)}$ as in Proposition~\ref{prop:orbitsum}.
\end{enumerate}
\end{theorem}

\begin{proof}
\begin{enumerate}
    \item[(a)] If $C \subseteq L(C_s^{(n)})$ is a $p$-coset then, being $p$-invariant and lying over $C_s^{(n)}$,
Theorem~\ref{thm:project}(1) gives $\pi(C)=C_s^{(n)}$; and any lift of $C_s^{(n)}$ is by definition a
$p$-coset contained in $L(C_s^{(n)})$.
\item[(b)] Since $\pi$ is $h$-to-$1$ and $\vert{}C_s^{(n)}\vert{}=\kappa$, we have $\vert{}L(C_s^{(n)})\vert{} = h \kappa$. The set $L(C_s^{(n)})$ is a union of
$p$-cosets (it is $p$-invariant, as $\pi$ intertwines multiplication by $p$), say
$C^{(1)},\dots,C^{(r)}$, with $\sum_i \vert{}C^{(i)}\vert{} = h\kappa$. By Theorem~\ref{thm:project}(1),
$\vert{}C^{(i)}\vert{} = w_i \kappa$, so $\sum_i w_i \kappa = h\kappa$, giving $\sum_i w_i = h$.
\item[(c)] Apply Proposition~\ref{prop:orbitsum} to each $C^{(i)}$ and sum, then use part (b).
\end{enumerate}
\end{proof}

\begin{remark}
Theorem~\ref{thm:liftproject}(b) is the exact ``conservation law'' behind lifting in the sense that the $h$-fold cover of
$\Z_n$ by $\Z_{hn}$ is partitioned among the cover cosets so that their wrapping multiplicities always sum to
$h$, independently of $p$. The two extreme cases are given by a single fully-wrapped coset ($r=1$, $w_1=h$) and a
maximal split into $h$ unwrapped copies ($r=h$, all $w_i=1$). Part (2) of Theorem~\ref{thm:project} pins down
which case occurs when $\gcd(h,n)=1$, and part (4) does so when $h \mid n$.
\end{remark}

The relation $\rho(\sigma_u) = w\,\sigma_s$ is an identity between orbit sums, not
idempotents. An orbit sum $\sigma_u$ is in general not an idempotent.
For example, for $p=3$, $m=8$ the
coset $C_1^{(8)} = \{1,3\}$ has $\sigma^2 \neq \sigma$, and for $p=2$, $m=7$ the coset $\{1,2,4\}$ likewise
fails.
Consequently the formula $\rho(\widetilde{e}_u) = w\, e_s$ cannot hold for idempotents.
The map $\rho$ is a ring homomorphism, so it maps every idempotent to an idempotent, and $w\, e_s$ is idempotent
only when $w \equiv 0$ or $1 \pmod p$. The correct statement for  idempotents is
Proposition~\ref{prop:idem} below, whose survival criterion is a spectral-support condition, not
$p \nmid w$.

\begin{proposition}
\label{prop:idem}
Assume $\gcd(p, hn) = 1$. Primitive idempotents of $\widetilde{R}$ (resp.\ $R$) are
indexed by the $p$-cyclotomic cosets of $\Z_{hn}$ (resp.\ $\Z_{n}$) in the
frequency (spectral) domain via the discrete Fourier correspondence. Let
$\widetilde{e}_O$ be the primitive idempotent of $\widetilde{R}$ attached to the frequency $p$-Frobenius orbit
$O \subseteq \Z_{hn}$. Then the following assertions hold true.
\begin{enumerate}
    \item[(1)] $\rho(\widetilde{e}_O)$ is again an idempotent of $R$. It is a sum of primitive idempotents.
    If $\rho(\widetilde e_O)=w e_{O'}\neq 0$ for some $p$-Frobenius orbit $O'\subseteq\Z_n$, then $w\equiv 1\pmod p$.

    \item[(2)] $\rho(\widetilde{e}_O) \neq 0$ if and only if $O \subseteq h\,\Z_{hn}$. Equivalently, the projection annihilates precisely
    those spectral components whose frequency is not divisible by $h$ (the ``fold-away'' frequencies).

    \item[(3)] The surviving primitive idempotents of $\widetilde{R}$ are in bijection, via
    $O \mapsto \tfrac1h O$ and $\rho$, with the primitive idempotents of $R$. Hence $\rho$ restricts to a
    surjective ring map onto $R$ realized on the spectral support $h\,\Z_{hn} \cong \Z_n$.
\end{enumerate}
\end{proposition}

\begin{proof}
    \emph{(1)} Since $\rho\colon \widetilde{R} \to R$ is an $\F_p$-algebra homomorphism, it preserves both multiplication and addition. Therefore, $\rho(\widetilde{e}_O)^2 = \rho(\widetilde{e}_O^2) = \rho(\widetilde{e}_O)$. Since the image is idempotent in $R$, and $R$ is semisimple, $\rho(\widetilde{e}_O)$ must uniquely decompose into a sum of primitive orbit idempotents of $R$. 
    It can never equal a scalar multiple $w e_s$ for $w \not\equiv 1 \pmod p$. The element $w e_s$ is not an idempotent over $\F_p$ unless $w \equiv 1 \pmod p$.

    \emph{(2) and (3)} To see exactly which idempotents survive, we pass to the splitting field $K$ containing a primitive $hn$-th root of unity, $\zeta$. The primitive idempotents of $\widetilde{R} \otimes K$ are indexed by frequencies $j \in \Z_{hn}$ via the discrete Fourier transform:
    \begin{equation*}
        \widetilde{\varepsilon}_j = \frac{1}{hn} \sum_{z=0}^{hn-1} \zeta^{-jz} x^z.
    \end{equation*}
    Applying the ring projection $\rho(x^z) = x^{z \bmod n}$, we evaluate the image of $\widetilde{\varepsilon}_j$. We can uniquely decompose the summation index as $z = c n + v$, where $0 \le v < n$ tracks the base coordinate and $0 \le c < h$ tracks the cover sheet. Substituting this yields:
    \begin{align*}
        \rho(\widetilde{\varepsilon}_j) 
        &= \frac{1}{hn} \sum_{v=0}^{n-1} \sum_{c=0}^{h-1} \zeta^{-j(cn+v)} x^v \\
        &= \frac{1}{hn} \sum_{v=0}^{n-1} \zeta^{-jv} x^v \left( \sum_{c=0}^{h-1} (\zeta^{-jn})^c \right).
    \end{align*}
    The inner sum over $c$ is a geometric series. Since $\zeta$ is a primitive $hn$-th root of unity, $\zeta^{-jn}$ is an $h$-th root of unity. Thus, the inner sum equals $h$ if $h \mid j$, and equals $0$ otherwise. 
    
    If $h \nmid j$, the entire expression vanishes, meaning $\rho(\widetilde{\varepsilon}_j) = 0$. If $h \mid j$, we can write $j = h \cdot i$ for some $i \in \Z_n$. The roots simplify since $\zeta^{-jv} = \zeta^{-hiv} = (\zeta^h)^{-iv}$. Noting that $\xi = \zeta^h$ is a primitive $n$-th root of unity, the projection collapses perfectly to the base idempotent:
    \begin{equation*}
        \rho(\widetilde{\varepsilon}_{hi}) = \frac{h}{hn} \sum_{v=0}^{n-1} \xi^{-iv} x^v = \frac{1}{n} \sum_{v=0}^{n-1} \xi^{-iv} x^v = \varepsilon_i.
    \end{equation*}
    Finally, we return to the base field $\F_p$. The primitive orbit idempotent $\widetilde{e}_O$ is the sum of the splitting-field idempotents over the $p$-Frobenius orbit $O \subseteq \Z_{hn}$. Since $\gcd(p, hn) = 1$, multiplication by $p$ preserves divisibility by $h$. Therefore, the orbit $O$ either consists entirely of multiples of $h$, or it contains no multiples of $h$ whatsoever. 
    
    If $O \not\subseteq h\Z_{hn}$, every frequency in the orbit folds away to zero, so $\rho(\widetilde{e}_O) = 0$. If $O \subseteq h\Z_{hn}$, the orbit corresponds exactly to a $p$-orbit $O/h$ in $\Z_n$, and $\rho(\widetilde{e}_O) = \sum_{i \in O/h} \varepsilon_i = e_{O/h}$. This establishes both the survival criterion (2) and the explicit bijection (3).
\end{proof}

Proposition~\ref{prop:orbitsum} and Proposition~\ref{prop:idem} together
give the algebraic mechanism behind the vanishing of projected logical operators. More precisely, 
they explain the situation that when a cover-code logical
operator is projected via the chain map to the base code~\cite{symons2025sequences}, an orbit-sum
representative is scaled by the wrapping multiplicity $w$ (vanishing if and only if $p \mid w$), while the
corresponding idempotent survives if and only if its spectral support is divisible by $h$.
\medskip

\subsection{Logical Dimension of Covering Codes}
\label{subsec:cover-dimension-resolution}

In recent work, Symons, Rajput, and Browne~\cite{symons2025sequences} explored infinite sequences of bivariate bicycle codes generated by covering graphs. Given a base BB code $Q = \mathrm{BB}(a, b)$ with Tanner graph $T(Q)$, an $h$-cover code $\tilde{Q} = \mathrm{BB}(\tilde{a}, \tilde{b})$ is a BB code whose Tanner graph $T(\tilde{Q})$ is an $h$-sheeted covering graph of $T(Q)$. 

In the language of group algebras, if the base code $Q$ is defined over $R = \F_q[x, y] / \idlGens{x^\ell - 1, y^m - 1}$, then an $h$-cover code $\tilde{Q}$ is defined over an enlarged ring $\tilde{R} = \F_q[x, y] / \idlGens{x^{\tilde{\ell}} - 1, y^{\tilde{m}} - 1}$ where $\tilde{\ell} = u\ell$, $\tilde{m} = tm$, and the covering degree is $h = ut$. For $T(\tilde{Q})$ to be a valid cover of $T(Q)$, the defining check polynomials must satisfy
\begin{equation}
    \tilde{a} \equiv a \pmod{\idlGens{x^\ell - 1, y^m - 1}}, \qquad \tilde{b} \equiv b \pmod{\idlGens{x^\ell - 1, y^m - 1}}.
\end{equation}
In other words, the monomials of $\tilde{a}$ and $\tilde{b}$ are obtained from those of $a$ and $b$ by adding multiples of $\ell$ and $m$ to the $x$- and $y$-exponents, respectively.

To relate the parameters of the cover code $\tilde{Q}$ to those of the base code $Q$, the authors of~\cite{symons2025sequences} initially constructed projection and lifting chain maps. 
This established $k_h \ge k$ for covering degrees coprime to the field characteristic. They recently extended this result to all covering degrees $h$ using a quotient-ring surjection. 
We present Theorem~\ref{thm:unconditional-cover-dimension} as an independent algebraic reproof of this fact. 
We analyze the commutative algebra of the check ideals directly. Our formulation yields a new refinement through a short exact sequence and an explicit kernel $K_{\mathrm{new}}$. 
This approach requires no semisimplicity assumption on $R$ or $\tilde{R}$. It also establishes an exact equality criterion for the dimension change.

\begin{theorem}
    \label{thm:unconditional-cover-dimension}
    Let $Q = \mathrm{BB}(a, b)$ be a base BB code over $R = \F_q[x, y] / \idlGens{x^\ell - 1, y^m - 1}$ with logical dimension $k$, and let $\tilde{Q} = \mathrm{BB}(\tilde{a}, \tilde{b})$ be any $h$-cover BB code over $\tilde{R} = \F_q[x, y] / \idlGens{x^{u\ell} - 1, y^{tm} - 1}$ with logical dimension $k_h$, where $h = ut \ge 1$. Let $I = \idlGens{a, b}_R \subseteq R$ and $\tilde{I} = \idlGens{\tilde{a}, \tilde{b}}_{\tilde{R}} \subseteq \tilde{R}$ denote the ideals generated by the check polynomials. 
    
    Then the canonical covering projection induces a surjective $\F_q$-linear map between the quotient spaces, $\bar{\rho}\colon \tilde{R}/\tilde{I} \to R/I$. Defining $K_{\mathrm{new}} = \ker \bar{\rho}$ yields the short exact sequence of $\F_q$-vector spaces
    \begin{equation}
        0 \longrightarrow K_{\mathrm{new}} \longrightarrow \tilde{R}/\tilde{I} \xrightarrow{\ \bar{\rho}\ } R/I \longrightarrow 0,
    \end{equation}
    and the logical dimensions satisfy the exact relation
    \begin{equation}
        k_h = k + 2\dim_{\F_q}(K_{\mathrm{new}}) \ge k.
    \end{equation}
    In particular, the logical dimension of an $h$-cover BB code is never smaller than that of its base code, for any covering degree $h \ge 1$ and over any field characteristic $p$.
\end{theorem}

\begin{proof}
    Because $x^{u\ell} - 1 = (x^\ell - 1)(1 + x^\ell + \dots + x^{(u-1)\ell})$ and $y^{tm} - 1 = (y^m - 1)(1 + y^m + \dots + y^{(t-1)m})$, the base lattice relations generate a subideal of the covering relations in $\F_q[x, y]$:
    \begin{equation*}
        \idlGens{x^{u\ell} - 1, y^{tm} - 1} \subseteq \idlGens{x^\ell - 1, y^m - 1}.
    \end{equation*}
    Therefore, canonical reduction modulo $\idlGens{x^\ell - 1, y^m - 1}$ is a well-defined, surjective ring homomorphism
    \begin{equation*}
        \rho\colon \tilde{R} \to R, \qquad \ker \rho = \idlGens{x^\ell - 1, y^m - 1}_{\tilde{R}}.
    \end{equation*}
    By the defining condition of a covering code, we have $\rho(\tilde{a}) = a$ and $\rho(\tilde{b}) = b$. This immediately implies the inclusion $\rho(\tilde{I}) \subseteq I$. Conversely, for any element $r a + s b \in I$, we may choose preimages $\tilde{r}, \tilde{s} \in \tilde{R}$ such that $\rho(\tilde{r}) = r$ and $\rho(\tilde{s}) = s$, since $\rho$ is surjective. It follows that $r a + s b = \rho(\tilde{r} \tilde{a} + \tilde{s} \tilde{b}) \in \rho(\tilde{I})$, proving the ideal equality $\rho(\tilde{I}) = I$. By the Third Isomorphism Theorem for rings and vector spaces, $\rho$ descends to a well-defined, surjective $\F_q$-linear map between the quotient spaces: $\bar{\rho}\colon \tilde{R}/\tilde{I} \to R/I$. 
    
    By the definition of $K_{\mathrm{new}}$ from the theorem statement, this establishes the short exact sequence of $\F_q$-vector spaces. By the rank-nullity theorem, the dimensions of these vector spaces satisfy
    \begin{equation}
        \label{eq:quotient-dim-split}
        \dim_{\F_q}(\tilde{R}/\tilde{I}) = \dim_{\F_q}(R/I) + \dim_{\F_q}(K_{\mathrm{new}}).
    \end{equation}

    For the base code $Q$, the $\F_q$-linear syndrome map $\Psi_X\colon R^2 \to R$ defined by $(u, v) \mapsto au + bv$ has matrix representation $H_X = \begin{pmatrix} A & B \end{pmatrix}$. Its image is precisely the ideal $I = \idlGens{a, b}_R$, so $\rank(H_X) = \dim_{\F_q} I$. Symmetrically, for the $Z$-checks, we have $\im(H_Z) = \idlGens{\iota(a), \iota(b)}_R = \iota(I)$. Since the coordinate involution $\iota$ is an $\F_q$-linear automorphism of $R$, it preserves dimension, giving $\rank(H_Z) = \dim_{\F_q} I$. 
    
    Applying the standard CSS dimension formula to $Q$ (where $n = 2\dim_{\F_q} R$) yields
    \begin{equation*}
        k = 2\dim_{\F_q} R - \rank(H_X) - \rank(H_Z) = 2\dim_{\F_q} R - 2\dim_{\F_q} I = 2\dim_{\F_q}(R/I).
    \end{equation*}
    By the identical algebraic argument applied to the covering code $\tilde{Q}$ over $\tilde{R}$, we have
    \begin{equation*}k        k_h = 2\dim_{\F_q}(\tilde{R}/\tilde{I}).
    \end{equation*}
    Substituting these two logical dimension identities into Equation~\eqref{eq:quotient-dim-split} and multiplying by $2$, we obtain
    \begin{equation*}
        k_h = k + 2\dim_{\F_q}(K_{\mathrm{new}}) \ge k,
    \end{equation*}
    which completes the proof.
\end{proof}

Our algebraic description also provides an exact, computable criterion for when an $h$-cover code strictly increases the number of logical qubits ($k_h > k$) versus when the logical dimension remains static ($k_h = k$):

\begin{corollary}
    \label{cor:strict-growth-condition}
    With the notation and assumptions of Theorem~\ref{thm:unconditional-cover-dimension}, the change in logical dimension is given by
    \begin{equation}
        k_h - k = 2\dim_{\F_q}\!\left( \frac{\ker \rho}{\ker \rho \cap \tilde{I}} \right).
    \end{equation}
    Consequently, the cover code has the same logical dimension as the base code ($k_h = k$) if and only if the base lattice relations are already absorbed by the covering check ideal:
    \begin{equation}
        k_h = k \iff \ker \rho \subseteq \tilde{I} \iff x^\ell - 1, \; y^m - 1 \in \idlGens{\tilde{a}, \tilde{b}}_{\tilde{R}}.
    \end{equation}
\end{corollary}

\begin{proof}
    By the Second Isomorphism Theorem for vector spaces, the kernel of the quotient surjection $\bar{\rho}$ is naturally isomorphic to
    \begin{equation*}
        K_{\mathrm{new}} = \ker \bar{\rho} = \frac{\ker \rho + \tilde{I}}{\tilde{I}} \cong \frac{\ker \rho}{\ker \rho \cap \tilde{I}}.
    \end{equation*}
    Substituting this isomorphism into $k_h - k = 2\dim_{\F_q}(K_{\mathrm{new}})$ yields the stated formula. We have $k_h = k$ if and only if $K_{\mathrm{new}} = 0$, which occurs if and only if $\ker \rho \subseteq \tilde{I}$. Since $\ker \rho$ is generated as a $\tilde{R}$-ideal by $x^\ell - 1$ and $y^m - 1$, this containment holds if and only if $x^\ell - 1, y^m - 1 \in \idlGens{\tilde{a}, \tilde{b}}_{\tilde{R}}$.
\end{proof}

\begin{remark}
    \label{rem:semisimple-cover-spectrum}
    When both $R$ and $\tilde{R}$ are semisimple (when  $\gcd(p, hn) = 1$), the quotient $\tilde{R}/\tilde{I}$ decomposes via Wedderburn into a direct product of the residue fields corresponding to the common spectral zeros of $\tilde{a}$ and $\tilde{b}$. Let $\tilde{\mathcal{G}} = \mu_{\tilde{\ell}} \times \mu_{\tilde{m}} \subset \overline{\F_q}^\times \times \overline{\F_q}^\times$ denote the full covering evaluation grid, and let $\mathcal{G} = \mu_\ell \times \mu_m \subseteq \tilde{\mathcal{G}}$ denote the embedded base grid. 
    
    The spectral zeros of $\tilde{a}$ and $\tilde{b}$ restricted to the base grid $\mathcal{G}$ coincide exactly with the zeros of $a$ and $b$, accounting for the base quotient $R/I$. The kernel $K_{\mathrm{new}}$ is therefore the direct sum of the simple field components supported on the \emph{strictly new} common zeros, $\Delta \mathcal{Z} = (\mathcal{Z}(\tilde{a}) \cap \mathcal{Z}(\tilde{b})) \setminus \mathcal{G}$. Summing the $\F_q$-dimensions of these field components over all new Frobenius orbits $O \subseteq \Delta \mathcal{Z}$ yields $\dim_{\F_q}(K_{\mathrm{new}}) = \sum_{O \subseteq \Delta \mathcal{Z}} |O| = |\Delta \mathcal{Z}|$. Thus, in the semisimple cases, Theorem~\ref{thm:unconditional-cover-dimension} reduces to the explicit root-counting formula $k_h = k + 2|\Delta \mathcal{Z}|$, where $|\Delta \mathcal{Z}|$ counts the number of new common roots of $\tilde{a}$ and $\tilde{b}$ in $\tilde{\mathcal{G}} \setminus \mathcal{G}$.
\end{remark}

\begin{remark}
    \label{rem:qudit-dimension}
    Symons et al.~\cite{symons2025sequences} extended their covering graph constructions to qudit BB codes over $\F_q$. 
    They recently established the dimension bound $k_h \ge k$ for all covering degrees. 
    Theorem~\ref{thm:unconditional-cover-dimension} provides an independent algebraic proof of this logical dimension lower bound. 
    This holds without restriction on the characteristic $p$ or the covering degree $h$. 
    The behavior of the exact minimum distance remains open; Symons et al.\ conjectured that $d \le d_h \le hd$.
\end{remark}

\subsection{Distance Certificates for Covering Codes}
\label{subsec:cover-distance}

The preceding dimension result applies to arbitrary covers. We now briefly examine how the spectral certificates appearing in the colon-ideal distance bound behave under a uniform cover $(\ell,m)\mapsto(s\ell,sm)$. The total covering degree in this case is $h=s^2$. We assume throughout this subsection that $\gcd(p,s\ell m)=1$, so that both the base and covering rings are semisimple. Our results address distance certificates rather than exact distances.

Let $\rho:\widetilde R\to R$ be the covering projection that reduces coefficient exponents, $\rho(x^{\tilde z}y^{\tilde w})=x^{\tilde z\bmod\ell}y^{\tilde w\bmod m}$, extended $\F_p$-linearly. In the frequency domain, the corresponding map is the injection $\phi:\Z_\ell\times\Z_m\to\Z_{s\ell}\times\Z_{sm}$ given by $\phi(i,j)=(si,sj)$. Its image is the embedded base subgroup $G_s=s\Z_{s\ell}\times s\Z_{sm}$.

The following result records exactly which spectral information is forced by the covering condition.

\begin{theorem} 
\label{thm:lifted-bb-distance}
Let $Q=\BB(a,b)$ be a bivariate bicycle code over $R=\F_p[x,y]/\langle x^\ell-1,\ y^m-1\rangle$, and let $s\geq 2$ satisfy $\gcd(p,s\ell m)=1$. Let $\widetilde Q=\BB(\widetilde a,\widetilde b)$ be a covering code over $\widetilde R=\F_p[x,y]/\langle x^{s\ell}-1,\ y^{sm}-1\rangle$ such that $\rho(\widetilde a)=a$ and $\rho(\widetilde b)=b$. Then $\mathcal Z_{\widetilde a}\cap G_s=\phi(\mathcal Z_a)$ and $\mathcal Z_{\widetilde b}\cap G_s=\phi(\mathcal Z_b)$.

Consequently, each zero-set region of the cover restricts to the corresponding injected base region: $\mathcal T_{\widetilde a,\widetilde b}\cap G_s=\phi(\mathcal T_{a,b})$, $\mathcal U_{\widetilde a,\widetilde b}\cap G_s=\phi(\mathcal U_{a,b})$, $\mathcal U_{\widetilde b,\widetilde a}\cap G_s=\phi(\mathcal U_{b,a})$, and $\mathcal F_{\widetilde a,\widetilde b}\cap G_s=\phi(\mathcal F_{a,b})$. The covering condition imposes no restriction on the cover zero sets outside $G_s$.
\end{theorem}

\begin{proof}
Fix primitive roots $\widetilde\alpha$ and $\widetilde\beta$ of orders $s\ell$ and $sm$, respectively, and set $\alpha=\widetilde\alpha^{\,s}$ and $\beta=\widetilde\beta^{\,s}$. Then $\alpha$ and $\beta$ have orders $\ell$ and $m$. For every coefficient exponent $\tilde z$, we have $\widetilde\alpha^{\,si\tilde z}=\alpha^{i\tilde z}=\alpha^{i(\tilde z\bmod\ell)}$, and similarly $\widetilde\beta^{\,sj\tilde w}=\beta^{j\tilde w}=\beta^{j(\tilde w\bmod m)}$. It follows that $\widetilde a(\widetilde\alpha^{\,si},\widetilde\beta^{\,sj})=\rho(\widetilde a)(\alpha^i,\beta^j)=a(\alpha^i,\beta^j)$. Hence $(i,j)\in\mathcal Z_a$ if and only if $\phi(i,j)\in\mathcal Z_{\widetilde a}$. The argument for $\widetilde b$ is identical.

The identities for $\mathcal T$, $\mathcal U$, and $\mathcal F$ follow by applying the corresponding set operations to these two equalities. For example, $(\mathcal Z_{\widetilde b}\setminus\mathcal Z_{\widetilde a})\cap G_s=\phi(\mathcal Z_b)\setminus\phi(\mathcal Z_a)=\phi(\mathcal Z_b\setminus\mathcal Z_a)=\phi(\mathcal U_{b,a})$. The other cases follow in the same way.
\end{proof}

Theorem~\ref{thm:lifted-bb-distance} is a restriction theorem and does not imply that minimum-distance certificates are preserved. 
The colon-ideal bound for the cover has the form $d(\widetilde Q)\geq\min\{\widetilde E_a,\widetilde E_b,\widetilde N_{a,b}\}$, so all three terms must be controlled before obtaining a lower bound on the cover distance.
Moreover, the 2D BCH bound is stated in terms of unions of consecutive full zero strips. 
The injection map $\phi(i,j)=(si,sj)$ spaces out the inherited base roots by exactly $s-1$ gaps. 
Even if a full base strip union maps to the cover, consecutive base indices are mapped to indices spaced by $s$. 
Thus, the covering condition alone does not preserve contiguous BCH certificates. 
Preservation occurs only if the strictly new roots of the cover code fill these induced gaps to form consecutive full zero strips. 

While it is possible to explicitly prescribe the cover zero sets to form these consecutive full strips to guarantee a high minimum distance (as demonstrated in the product constructions of Appendix~\ref{sec:bch_product_codes}), enforcing such zero-strip structures in the frequency domain forces the inverse discrete Fourier transform to produce dense polynomials in the spatial domain. 
Similarly, in Example~\ref{subsec:bb510}, achieving a high certified classical distance requires dense zero sets, resulting in check weights of 53 and 54. By contrast, useful BB codes require sparse checks, which lack the degrees of freedom necessary to generate these consecutive full zero strips. 
Therefore, Theorem~\ref{thm:lifted-bb-distance} highlights a tension in covering sequences that a lift cannot preserve both the low-density parity-check structure and the full zero strips required for apparent-distance certificates simultaneously. 

The following example illustrates this phenomenon for a classical 2D cyclic ideal. 
Such ideals are precisely the objects whose distances occur in the annihilator and colon-ideal bounds.

\begin{example}
\label{ex-lifted-bb-distance}
Let $R=\F_2[x,y]/\langle x^3-1,\ y^3-1\rangle$ and consider the 2D cyclic code $\mathcal C=\langle g\rangle$, where $g=1+x+x^2$. Its defining set is $\mathcal Z(\mathcal C)=\{1,2\}\times\Z_3$. Thus, the defining set contains two consecutive full zero strip unions in the $x$ direction. Taking $\Gamma=\{1\}$, $J_1=\{1,2\}$, and $\delta_1=3$ in Theorem~\ref{thm:2D-BCH-bound} gives $d(\mathcal C)\geq 3$. Since the generator $g$ has weight $3$, this bound is exact and $d(\mathcal C)=3$.

Now consider the literal $3$-fold lift $\widetilde{\mathcal C}=\langle\widetilde g\rangle$ in $\widetilde R=\F_2[x,y]/\langle x^9-1,\ y^9-1\rangle$, where $\widetilde g=g$. The covering projection satisfies $\rho(\widetilde g)=g$, and the full cover defining set is $\mathcal Z(\widetilde{\mathcal C})=\{3,6\}\times\Z_9$. In particular, for $G_3=\{0,3,6\}\times\{0,3,6\}$, we have $\mathcal Z(\widetilde{\mathcal C})\cap G_3=\{3,6\}\times\{0,3,6\}=\phi(\mathcal Z(\mathcal C))$, as required by Theorem~\ref{thm:lifted-bb-distance}.

The cover contains full zero strip unions at the $x$-indices $3$ and $6$, but these indices are not consecutive. Consequently, the base certificate with designed distance $3$ is not inherited by the cover. Indeed, $1+x^3=(1+x)(1+x+x^2)=(1+x)\widetilde g$ belongs to $\widetilde{\mathcal C}$, so $d(\widetilde{\mathcal C})\leq 2$. On the other hand, the nonempty defining set rules out a weight-one codeword, since a nonzero scalar multiple of a monomial evaluates nontrivially at every point of the root grid. Hence $d(\widetilde{\mathcal C})=2$.

This example shows that the distance of an individual classical ideal appearing in the colon-ideal bound can decrease under a lift. It does not, by itself, imply that the colon-ideal bound of a BB code decreases since this involves the distance certificates $\widetilde E_a$, $\widetilde E_b$, and $\widetilde N_{a,b}$ from multiple ideals. 
\end{example}

\section{Examples}
\label{sec:examples}
The main formulas from the idempotent framework are easily validated against the existing literature. We gathered all the previously known BB codes and selected three at random for demonstration.

\begin{example}
    \label{ex:bb90}
    Reference~\cite{EberhardtSteffan2025} contains a $[\![90, 8, \leq 10]\!]$ code with $\ell = 3$, $m = 15$, $a(x, y) = 1 + y + y^5$, and $b(x, y) = x^2 + x + y^3$. The zero orbits are $\{(0, 5), (1, 5), (1, 10)\}$ for $a(x, y)$ and $\{(1, 0), (1, 5), (1, 10)\}$ for $b(x, y)$. The common-zero region has size $|\{(1, 5), (1, 10)\}| = 4$. Hence $n = 90$ and $k = 8$. For the distance, we compute $\mathcal{U}_{a, b}$, $\mathcal{U}_{b, a}$, and $\mathcal{F}_{a, b}$. Evaluating the single-block annihilator codes $\ann\idlGens{a} = \mathcal{C}(\mathcal{U}_{b, a} \cup \mathcal{F}_{a, b})$ and $\ann\idlGens{b} = \mathcal{C}(\mathcal{U}_{a, b} \cup \mathcal{F}_{a, b})$ yields $E_a = 10$ and $E_b = 10$. The mixed-block colon-ideal sum evaluates to $N_{a, b} = d_R(\idlGens{b:a}) + d_R(\idlGens{a:b}) = 2 + 2 = 4$. The guaranteed minimum distance lower bound is therefore $\min\{E_a, E_b, N_{a, b}\} = \min\{10, 10, 4\} = 4$. 
\end{example}

\begin{example}
    \label{ex:bb162}
    Reference~\cite{EberhardtSteffan2025} contains a $[\![162, 24, \leq 6]\!]$ code with $\ell = 9$, $m = 9$, $a(x, y) = 1 + y + y^2$, and $b(x, y) = y^3 + x^3 + x^6$. The zero orbits are $\{(0, 3), (1, 3), (1, 6), (3, 3), (3, 6)\}$ for $a(x, y)$ and $\{(1, 0), (1, 3), (1, 6)\}$ for $b(x, y)$. The common-zero region has size $|\{(1, 3), (1, 6)\}| = 12$. Hence $n = 162$ and $k = 24$. For the distance, we compute $\mathcal{U}_{a, b}$, $\mathcal{U}_{b, a}$, and $\mathcal{F}_{a, b}$. The single-block uncoupled distances evaluate to $E_a = 6$ and $E_b = 6$, while the mixed-block colon-ideal bound gives $N_{a, b} = 4$. The guaranteed lower bound is therefore $\min\{E_a, E_b, N_{a, b}\} = \min\{6, 6, 4\} = 4$.
\end{example}

\begin{example}
    Reference~\cite{symons2025sequences} contains a $[\![434, 10, \leq 26]\!]$ code with $\ell = 31$, $m = 7$, $a(x, y) = 1 + x^6 y + x^{27} y^4$, and $b(x, y) = 1 + x^{15} y^6 + x^{24} y^3$. 
    The zero orbits are $\{(1, 0)\}$ for $a(x, y)$ and $\{(1, 0)\}$ for $b(x, y)$. 
    The common-zero region has size $|\{(1, 0)\}| = 5$. Hence $n = 434$ and $k = 10$. 
    
    For the distance, we compute $\mathcal{U}_{a, b}$, $\mathcal{U}_{b, a}$, and $\mathcal{F}_{a, b}$. 
    The single-block distance $E_a$ is determined by the annihilator ideal $\ann\langle a \rangle = \mathcal{C}(\mathcal{U}_{b, a} \cup \mathcal{F}_{a, b})$ and gives a bound of 112. 
    The distance $E_b$ is determined by the ideal $\ann\langle b \rangle = \mathcal{C}(\mathcal{U}_{a, b} \cup \mathcal{F}_{a, b})$ and also gives a bound of 112. 
    The mixed-block term $N_{a, b}$ evaluates to 2.\footnote{The cyclic code corresponding to $\emptyset$ is the entire ring $R$ which contains weight-1 elements, providing one for each colon ideal.} 
    The colon-ideal minimum distance bound is therefore only 2 in this case.

    This example illustrates an interesting extreme case so we make some further remarks about it. 
    We have $\mathcal{Z}_a = \mathcal{Z}_b$ and that $\mathcal{U}_{a, b} = \mathcal{U}_{b, a} = \emptyset$ but $a \neq b$. 
    This implies that there exists a unit $c$ of $R$ such that $a = b c$, in other words $\idlGens{a}=\idlGens{b}$. 
    Note that $c$ is not unique unless $k=0.$
   
    Despite the fact that $\idlGens{a} = \idlGens{b}$, the distance does not drop to 2 as it would if $a = b$. 
    We show the distance must be at least 3 by considering two cases. 
    The case $\wt(u)=2$ covers $\wt(v)=2$ without loss of generality.
\begin{description}
    \item[Case 1] $\wt(u) = \wt(v) = 1$ \\
    If $u, v$ have weight 1, then they are monomials in $R$. 
    Thus $au + bv = 0$ implies $a$ is a monomial translate of $b$. 
    Comparison of the exponent supports of $a$ and $b$ shows that no such translation exists. 
    Hence no weight-two kernel vector exists in this case.

    \item[Case 2] $\wt(u) = 2$ \\
    In this case the logical has no support on the second block. 
    However the single-block distances $E_a = E_b = 112$ exclude a weight-two logical supported entirely in one block so this is not possible.
\end{description}
    Hence no weight-two kernel exists in either case, so $d \ge 3$.
    
    By fixing a particular unit $c$ so that $a=bc$ we can understand more about how low-weight logicals arise.
    Since $b$ has zeros, it is not invertible, so we cannot multiply by the inverse. 
    We must construct $c$ such that it has no zeros in the ring.
    Let $e_O$ be the idempotent corresponding to the common-zero orbits. 
    By~\eqref{2d_idemp_prop}, multiplying $c e_O$ isolates the zero orbits. 
    The choice of $c$ we consider is the one that acts as the identity on the common-zero region, that is $c e_O = e_O$. 
    Adding this to both sides gives $c b + c e_O = a + c e_O$. 
    Factoring out $c$ on the left and applying our design choice to the right gives $c(b + e_O) = a + e_O$. 
    Since $b + e_O$ is nonzero everywhere, it is invertible and we can multiply by the inverse to get $c = (a + e_O)(b + e_O)^{-1}$. 
    Distributing and simplifying with $e_O^2 = e_O$ gives $c = a(b + e_O)^{-1} + e_O$. 
    For this example, we find $c$ has Hamming weight 97.
    
    In general, if $(u, v)$ is a mixed-block logical, then $b (cu + v) = 0$ implies that $cu + v \in \ann \langle b \rangle$. 
    Let $f \in \ann \langle b \rangle$ be a codeword of the classical 2D cyclic code defined by $b$. 
    Then $(u, cu + f)$ is a valid logical operator with weight $\wt(u) + \wt(cu + f)$. 
    If $u$ is low-weight and $cu$ is close to a codeword of $\ann\langle b \rangle$ and $(u,v)$ is not a stabilizer then $(u,v)$ gives us a low-weight logical.
\end{example}

\section{Conclusion}
We have developed a spectral framework for bivariate
bicycle codes in the semisimple regime. Casting the defining ring
through its Wedderburn decomposition into finite-field components
indexed by $p$-Frobenius orbits, we obtained a root-counting formula for
the dimension, a colon-ideal lower bound on the minimum distance that
captures mixed-block logical operators, a check-module description of
the permutation automorphisms and $ZX$-dualities, and an account of how
these data transform under lifts and projections. Classical tools from cyclic-code theory appear throughout.

We aimed to keep the discussion within commutative ring theory, but
found it occasionally necessary to borrow from adjacent areas. Further
results are available from a purely module-theoretic viewpoint, although
it was not clear to us at the time of writing whether these yield
statements stronger than those presented here.

From a purely mathematical standpoint, most of the framework extends to
an arbitrary finite field $\mathbb{F}_q$ with little modification.
Stabilizer codes, however, are only additive, $\mathbb{F}_p$-linear
rather than $\mathbb{F}_q$-linear subspaces, and this submodule
structure requires care, particularly in classifying automorphisms. Extending the definition of BB codes to the linear closure may be an interesting direction of research.

Our definition of BB codes extends past the trinomial ansatz of the
original construction, and although we report results only for
untwisted tori, the spectral framework   carries over to
twisted tori~\cite{liang2025generalized}. Combining group-algebra codes with twists is an interesting new direction, which we will report in an upcoming paper. k
There are several additional open problems that follow directly from this work. The most immediate is the extension of the idempotent framework to the non-semisimple case. On the zero-set side this means tracking zeros with multiplicity. On the algebraic side, the Wedderburn field factors are replaced by local Artinian rings containing nilpotents. Central idempotents still separate the factors but the simple field structure is lost. The resulting analysis is likely to be less computationally friendly, relying on the Gr\"obner-basis computations the semisimple theory is designed to avoid.

A second major open direction is stabilizer weight. We have made no attempt to
reduce the check weights of the example codes, and controlling the
Hamming weight is not a natural operation in the algebraic language.
Practical implementations of QLDPC codes are often hindered by the requirement of non-local entanglement across many qubits. Photonic platforms natively support arbitrary connectivity between qubits and are ideally suited for these architectures. While some quantum architectures like Xanadu's may be able to handle higher-weight stabilizers, the stabilizer weights in the examples are likely out of reach for the majority of near-term technologies.

The freedom to replace generators by independent unit multiples moves within a family with the same logical dimension and colon-ideal certificates, while a common unit multiple leaves the check-row module unchanged. We expect that the most useful codes will emerge from combining the constructions given here with numerical data on the weights obtained by grouping orbits while sweeping over semisimple grid sizes.

While we did not frame it this way, the spectral framework and most of our results hold for abelian two-block group-algebra codes, including multivariate extensions of BB codes~\cite{voss_multivariate_2025}. The non-abelian case should be relatively straightforward given the framework established here. Codes with additional blocks~\cite{jacob2025single} are more complicated and represent an interesting and nontrivial extension.

Finally, it remains to better understand how the structure of BB codes influences decoding.

\section{Acknowledgments}
Some of this work dates back to the publication of~\cite{bravyi2024}. Much of this work existed before the advent of LLMs capable of helping. AI was not used to pursue new ideas or directions for this work. 
AI was used in the late stages of writing to verify results and consistency throughout. Suggestions were merged into the introduction and discussion. Numerical computations were performed with the coding theory library~\cite{Sabo_2026}.

\appendix
\section{BCH-Based Constructions \& Product Codes}
\label{sec:bch_product_codes}

Following the path established by the classical literature on distance lower bounds for 2D cyclic codes, this section highlights several of our results on the minimum distances of BB codes, obtained by analyzing their defining data either directly through generating polynomials or by prescribing their annihilator ideals.

We maintain the notation of the main text. In particular, let $R_x = \F_p[x]/\idlGens{x^\ell - 1}$, $R_y = \F_p[y]/\idlGens{y^m - 1}$, and $R=\F_p[x,y]/\langle x^\ell-1,\ y^m-1\rangle$.
Via the natural embeddings $R_x\hookrightarrow R$ and $R_y\hookrightarrow R$, we view $g_x(x)\in R_x$ and $g_y(y)\in R_y$ as elements of $R$. Let $C_x=\langle g_x\rangle \subseteq R_x$ and $C_y=\langle g_y\rangle \subseteq R_y$ be cyclic codes of lengths $\ell$ and $m$ respectively.

\begin{definition}
    The \emph{2D product-code ideal determined by $g_x$ and $g_y$}, denoted $\mathsf{Prod}(C_x, C_y)$, is the 2D cyclic code defined by the ideal
    \begin{equation}
        \label{eq:prod-ideal}
        \mathsf{Prod}(C_x, C_y) = \langle g_x(x) g_y(y) \rangle \subseteq R.
    \end{equation}
\end{definition}

As before, we fix a splitting field $K/ \F_p$ containing primitive $\ell$-th and $m$-th roots $\alpha$ and $\beta$. For $f\in R$, we define its zero set $\mathcal{Z}(f)\subseteq \Z_\ell \times \Z_m$ as the set of coordinates $(i,j)$ where $f(\alpha^i,\beta^j)=0$.

\begin{lemma}
\label{lem:zeros_prod_strips}
Let $g_x\in R_x$ and $g_y\in R_y$, with 1D zero sets $\mathcal{Z}_x\subseteq\Z_\ell$ and $\mathcal{Z}_y\subseteq\Z_m$. The zero set of their product in $R$ is a union of strips:
\[
\mathcal{Z}\bigl(g_x(x)g_y(y)\bigr) = \bigl(\mathcal{Z}_x\times \Z_m\bigr)\ \cup\ \bigl(\Z_\ell\times \mathcal{Z}_y\bigr).
\]
\end{lemma}
The result follows directly from the multiplicative property of the evaluation map. A point $(\alpha^i, \beta^j)$ is a zero of the product if and only if it is a root of $g_x$ or a root of $g_y$, which geometrically traces out the horizontal and vertical strips corresponding to $Z_x$ and $Z_y$.

Since $\mathsf{Prod}(C_x,C_y)$ is the tensor product $C_x \otimes C_y$, its parameters follow from standard results for tensor products of linear codes.

\begin{proposition}
    \label{prop:prod-distance}
    Let $\mathcal{C}$ be the product-code $\mathcal{C} := \mathsf{Prod}(C_x,C_y) \subseteq R$. 
    Then its dimension is $\dim_{\F_q} \mathcal{C} = (\dim C_x)(\dim C_y)$, 
    and its minimum distance is $d_{R_x}(C_x)\,d_{R_y}(C_y)$. 
    Any distance lower bounds for the $x$ and $y$ codes immediately carry over to the product. In particular, for BCH design distances $\delta_x$ and $\delta_y$ we have
    \[d_R(\mathcal{C}) \geq \delta_x \cdot \delta_y.\]
\end{proposition}
The distance identity of the proposition is established by viewing any non-zero codeword in the product code as an $\ell \times m$ matrix. 
Every non-zero row belongs to $C_y$ and thus contains at least $d_{R_y}(C_y)$ non-zero entries. 
By column-wise linearity, each of these entries belongs to a column vector in $C_x$ weighing at least $d_{R_x}(C_x)$, yielding a total weight of at least the product of the two distances. 
Furthermore, the specific BCH lower bound $d_R(\mathcal{C}) \ge \delta_x \delta_y$ is an immediate consequence of the 2D BCH bound (Theorem~\ref{thm:2D-BCH-bound}), as the union of strips established in the preceding lemma provides the exact contiguous zero columns required by the theorem.

\begin{example}
\label{ex:designdistance}
In this example, 
by using Proposition~\ref{prop:prod-distance}, we  construct a BB code with a partial guarantee on its minimum distance.

Recall the notation of the colon-ideal distance lower bound given in Theorem~\ref{thm:orbit-partition-colon-bound}, where we denoted the distances of these annihilators by $E_a$ and $E_b$ respectively. 
To break the symmetry between the blocks while maintaining control over the mixed-block distance, we pair a strong 1D cyclic code on one axis with a weak 1D cyclic code on the other. We also require the zero sets of the weak codes to be strictly nested inside the strong codes to perfectly isolate the uncoupled regions.

Let $\ell=7$, $m=15$, and $q=2$. We require the weak codes to have minimum distance $\gamma_x=\gamma_y=2$. Over $\mathbb{F}_2$, the only cyclic code of length $7$ with minimum distance exactly $2$ is the single-parity-check code. Its zero set is exactly $\{0\}$.\footnote{A binary cyclic code contains a weight-two word $x^{s_1}+x^{s_2}$ if and only if $\beta^{i(s_1-s_2)}=1$ for every $i$ in its zero set, where $\beta$ is a primitive $\ell$-th root of unity. For $\ell=7$ prime this forces the zero set to be $\{0\}$. Consequently a \emph{narrow-sense} strong code, whose zero set omits $0$, can never contain a distance-$2$ subcode, and the nesting hypothesis would be vacuous.} Since the weak zero sets must be nested inside the strong zero sets, the strong zero sets must also contain the root $0$. We  take the strong codes to be the even-weight subcodes of the respective Hamming codes:
\[
\begin{aligned}
\mathcal{Z}_x^{\mathrm{str}} &= \{0,1,2,4\},   &\quad C_x^{\mathrm{str}} &= [7,3,4],
&\quad \delta_x &= 4, \\
\mathcal{Z}_x^{\mathrm{wk}}  &= \{0\},         &\quad C_x^{\mathrm{wk}}  &= [7,6,2],
&\quad \gamma_x &= 2, \\
\mathcal{Z}_y^{\mathrm{str}} &= \{0,1,2,4,8\}, &\quad C_y^{\mathrm{str}} &= [15,10,4],
&\quad \delta_y &= 4, \\
\mathcal{Z}_y^{\mathrm{wk}}  &= \{0\},         &\quad C_y^{\mathrm{wk}}  &= [15,14,2],
&\quad \gamma_y &= 2 .
\end{aligned}
\]
Each strong zero set contains a run of three consecutive residues $\{0,1,2\}$. 
The BCH bound certifies the designed distances $\delta_x=\delta_y=4$, and both are attained. The nesting $\mathcal{Z}_x^{\mathrm{wk}} \subset \mathcal{Z}_x^{\mathrm{str}}$ and $\mathcal{Z}_y^{\mathrm{wk}} \subset \mathcal{Z}_y^{\mathrm{str}}$ now holds by construction.

We prescribe the annihilator ideals using the 2D product-code construction of Eq.~\ref{eq:prod-ideal}:
\[
\ann \idlGens{a} = \mathsf{Prod}(C_x^{\mathrm{str}}, C_y^{\mathrm{wk}}), \qquad
\ann \idlGens{b} = \mathsf{Prod}(C_x^{\mathrm{wk}}, C_y^{\mathrm{str}}).
\]

By Proposition~\ref{prop:prod-distance} the single-block distances $E_a$ and $E_b$
are lower bounded by the BCH designed distances of these annihilator ideals. 
By counting consecutive runs in the four zero sets above we see 
$\delta_x \gamma_y = 4 \cdot 2 = 8$ and $\gamma_x \delta_y = 2 \cdot 4 = 8$, yielding
\[
E_a \ge 8, \qquad E_b \ge 8 .
\]

Next we compute the mixed-block term $N_{a,b}=d_R(\idlGens{b:a})+d_R(\idlGens{a:b})$.
By Lemma~\ref{lem:zeros_prod_strips}, the zero sets of the prescribed
annihilator ideals are
\[
\begin{aligned}
\mathcal Z(\ann\idlGens{a})
&=
(\mathcal Z_x^{\mathrm{str}}\times\Z_m)
\cup
(\Z_\ell\times\mathcal Z_y^{\mathrm{wk}}),\\
\mathcal Z(\ann\idlGens{b})
&=
(\mathcal Z_x^{\mathrm{wk}}\times\Z_m)
\cup
(\Z_\ell\times\mathcal Z_y^{\mathrm{str}}).
\end{aligned}
\]
The zero set of an annihilator is the active support of the
annihilated element. Hence
\[
\begin{aligned}
\mathcal Z_a
&=
(\Z_\ell\setminus\mathcal Z_x^{\mathrm{str}})
\times
(\Z_m\setminus\mathcal Z_y^{\mathrm{wk}}),\\
\mathcal Z_b
&=
(\Z_\ell\setminus\mathcal Z_x^{\mathrm{wk}})
\times
(\Z_m\setminus\mathcal Z_y^{\mathrm{str}}).
\end{aligned}
\]
Using the nesting of the constituent zero sets, we therefore obtain zero sets for the ideals $\idlGens{b:a}$ and $\idlGens{a:b}$
\[
\begin{aligned}
\mathcal U_{b,a}
&=\mathcal Z_b\setminus\mathcal Z_a\\
&=
(\mathcal Z_x^{\mathrm{str}}
 \setminus\mathcal Z_x^{\mathrm{wk}})
\times
(\Z_m\setminus\mathcal Z_y^{\mathrm{str}}),\\
\mathcal U_{a,b}
&=\mathcal Z_a\setminus\mathcal Z_b\\
&=
(\Z_\ell\setminus\mathcal Z_x^{\mathrm{str}})
\times
(\mathcal Z_y^{\mathrm{str}}
 \setminus\mathcal Z_y^{\mathrm{wk}}).
\end{aligned}
\]
Explicitly, $\mathcal{U}_{b,a} = \{1,2,4\} \times \{3,5,6,7,9,10,11,12,13,14\}$ and $\mathcal{U}_{a,b} = \{3,5,6\} \times \{1,2,4,8\}$. 
The ideals $\idlGens{b:a}$ and $\idlGens{a:b}$ define 2D cyclic codes in $R$.
A general fact about 2D cyclic codes defined as a Cartesian product is their distance equals the minimum of the distances of its two 1D components. 
The $x$-factors $\{1,2,4\}$ and $\{3,5,6\}$ both have design distance 3 and the $y$-factors 
$\{3,5,6,7,9,10,11,12,13,14\}$ and $\{1,2,4,8\}$ have design distance 7 and 3 respectively\footnote{It happens that in this particular example the design distances for each of the four codes equal the code's true distance.}.
Evaluating the corresponding 1D polynomials gives
\[d_R(\idlGens{b:a}) = \min(3,7) = 3, \qquad d_R(\idlGens{a:b}) = \min(3,3) = 3,\]
Consequently, applying Theorem~\ref{thm:colon-lower-bound}, the resulting code $Q=\mathrm{BB}(a,b)$ satisfies $d_Z \ge \min(E_a, E_b, N_{a,b}) \ge \min(8, 8, 3+3) = 6.$
Finally, due to the equality of $X$ and $Z$ distances we have
\[
d(Q) \ge 6. 
\]

We note that this distance lower bound holds for any choice of stabilizer generator polynomials $a$ and $b$. 
By Theorem~\ref{thm:ideals-by-orbits}, the prescribed annihilator
ideals have the orbit-idempotent representations
\[
\ann\idlGens{a}=e_{\mathcal Z_a}R,
\qquad
\ann\idlGens{b}=e_{\mathcal Z_b}R.
\]
One pair realizing these annihilators is
\[
a=1-e_{\mathcal Z_a}=e_{\mathcal S_a},
\qquad
b=1-e_{\mathcal Z_b}=e_{\mathcal S_b}.
\]
As discussed below, multiplying $a$ and $b$ by a unit yields other choices.
\end{example}

\begin{remark}
    \label{rem:slack-orbits}
    In Example~\ref{ex:designdistance}, the prescribed strips do not exhaust the spectral grid $\Z_\ell \times \Z_m$. Using $\mathcal{Z}_a$ and the definition of $S_a$ it is not difficult to compute $S_a$. Doing this for $b$ as well we find that $\lvert S_a \cup S_b \rvert = 75$ of the $105$ points, so $\mathcal{T}_{a,b}$ contains $30$ points. These points are tiled by exactly three $p$-Frobenius orbits, namely
    \[
    \{3,5,6\}\times\{3,6,9,12\}, \qquad
    \{3,5,6\}\times\{5,10\}, \qquad
    \{3,5,6\}\times\{7,11,13,14\},
    \]
    of sizes $12$, $6$ and $12$. By selectively including or excluding these slack orbits from the zero sets $\mathcal{Z}_a$ and $\mathcal{Z}_b$ we can explore adjustments to the logical dimension $k.$
    The colon-ideal distance lower bound can be applied to the codes generated this way as a fast test to see if the distance is sufficiently high.  
\end{remark}

Constructions that start by prescribing annihilator ideals allow control over aspects of the code's distance as we have seen. 
It remains to find stabilizer generators $a$ and $b$ of low weight. If the canonical idempotent-complement constructors are too dense, we may replace them with sparser elements sharing the same annihilators without altering the logical dimension or the colon-ideal distance guarantees.
A naive brute force search is time-consuming. 
Our algebraic approach suggests restricting the search to the unit group of $R.$
The BB codes defined by generators $\{\mathrm{BB}(ua,ub): u \in R^{\times}\}$ are all equal since multiplication by $u$ is an invertible transformation on $R$.

Alternatively, using distinct unit multiples for each generator 
we may produce new codes with the same $n$, $k$, and lower bounded distance. 
Theorem~\ref{thm:annihilator-preserving-replacement} considers this in general.
\begin{theorem}
    \label{thm:annihilator-preserving-replacement}
    Assume $\gcd(p,\ell m)=1$. If two BB codes $Q=\mathrm{BB}(a,b)$ and $Q^\prime=\mathrm{BB}(a^\prime,b^\prime)$ over $R$ satisfy $\ann\idlGens{a^\prime}=\ann\idlGens{a}$ and $\ann\idlGens{b^\prime}=\ann\idlGens{b}$, then $k(Q^\prime)=k(Q)$ and $Q^\prime$ obeys the same colon-ideal minimum distance bounds as $Q$.
\end{theorem}

The proof relies on the fact that, in a semisimple ring, equality of the annihilator ideals forces the ideals $\idlGens{a}$ and $\idlGens{a^{\prime}}$ themselves to be equal, which holds respectively for $b$ and $b^{\prime}.$
From this we have $\idlGens{a:b}=\idlGens{a^{\prime}:b^{\prime}}$ (and the same holds for the reversed order). 
This establishes the colon-ideal minimum distance bound statement. 

As $\idlGens{a} = \idlGens{a^{\prime}}$ the constructors $a^\prime$ and $b^\prime$ are merely unit multiples of $a$ and $b$. 
Let $u,u^\prime\in R^{\times}$ be any units and set $a^\prime=ua$, $b^\prime=u^\prime b$. 
Then $\langle a^\prime\rangle=\langle a\rangle$ and $\langle b^\prime\rangle=\langle b\rangle$, so $\ann\idlGens{a^\prime}=\ann\idlGens{a}$ and $\ann\idlGens{b^\prime}=\ann\idlGens{b}$.
By Theorem~\ref{thm:BB-dimension} we get the equality of $k(Q)$ and $k(Q^\prime).$

Theorem~\ref{thm:annihilator-preserving-replacement} shows
choosing distinct units  
$u \neq u^\prime$ produces generators that define a code that may differ from the original but have  
the same logical dimension and satisfy the same minimum distance bounds as $\mathrm{BB}(a,b)$\footnote{The question of when codes generated this way are distinct has to do with the slope and addressing it is beyond the current discussion. We also stress that this theorem only establishes preservation of the distance lower bound, and gives no information about the exact distance.}.
 
\begin{example}
\label{ex:sparsify-replacement}

We illustrate this on the $[\![90,8,10]\!]$ code of~\cite{bravyi2024}, which lies over $R = \F_2[x,y]/\langle x^{15}-1,\ y^3-1\rangle$ which has two weight three generators $a = x^{9}+y+y^{2}$ and $b = 1+x^{2}+x^{7}.$
Since $\gcd(2,\,15\cdot 3)=1$, the ring $R$ is semisimple and the hypothesis of Theorem~\ref{thm:annihilator-preserving-replacement} holds. 
For example, the following generator pairs define codes with parameters $[\![90, 8]\!]$ and the same colon-ideal distance lower bound as the original code.  

\begin{align*}
    a^\prime_1 &= x^{3}+x^{3}y+x^{6}+x^{6}y^{2}+x^{12}, \\
    b^\prime_1 &= y+xy+x^{2}y+x^{5}y+x^{9}y+x^{11}y+x^{12}y; \\[4pt]
    a^\prime_2 &= x^{2}y+x^{5}+x^{5}y+x^{8}+x^{8}y^{2}+x^{11}+x^{11}y, \\
    b^\prime_2 &= y+x^{3}y+x^{4}y+x^{5}y+x^{8}y+x^{12}y+x^{14}y; \\[4pt]
    a^\prime_3 &= x+xy^{2}+x^{2}y^{2}+x^{4}y^{2}+x^{5}+x^{5}y^{2}+x^{8}y+x^{10}+x^{14}y^{2}, \\
    b^\prime_3 &= 1+xy+x^{5}+x^{7}+x^{9}+x^{9}y+x^{11}y+x^{13}+x^{14}.
\end{align*}
We verify $\dim_{\F_2}\ann\langle a,b\rangle = 4$, hence the dimension is $k=8.$ 

The weights for the three codes above: $(5,7)$, $(7,7)$, and $(9,9)$ are actually somewhat low weight considering the size of $\ell$ and $m$. 
Generally, the polynomials produced this way are high weight.
As discussed above, we may fix one of these codes with generators $(a,b)$ and search for a lower weight pair in $(ua, ub)$ for $u$ a unit. 
 
\end{example}

\bibliographystyle{plain}
\bibliography{refs}

\end{document}